\documentclass[11pt]{article}

\usepackage{amsfonts}
\usepackage{amsmath}
\usepackage[dvips]{graphicx}
\usepackage{graphicx,indentfirst,amsmath}
\usepackage{caption}
\usepackage{subcaption}
\usepackage{graphics}
\usepackage{comment}
\usepackage{lscape}
\usepackage{rotating}
\usepackage{color}
\usepackage[square]{natbib}

\newtheorem{theorem}{Theorem}[section]

\newtheorem{lemma}{Lemma}[section]
\newtheorem{proposition}{Proposition}[section]
\newtheorem{remark}{Remark}[section]
\newenvironment{proof}[1][]{\textbf{#1} }{\ \rule{0.5em}{0.5em}}
\makeatletter

\newcommand{\Rmnum}[1]{\expandafter\@slowromancap\romannumeral #1@}
\makeatother

\def\<{\langle}
\def\>{\rangle}

\begin{document}

\bibliographystyle{plainnat}

\begin{center}
 {\bf  \Large \textbf{Approximate Bayesian estimation in large coloured graphical Gaussian models}}\\
 { \text{Qiong Li, Xin Gao, H\'{e}l\`{e}ne Massam}\\
 \text{Department of Mathematics and Statistics, York University}}
\end{center}

{}
\begin{center}
\end{center}
\abstract
Distributed estimation methods have recently been used to compute the maximum likelihood estimate of the precision matrix for large graphical Gaussian models. Our aim, in this paper, is to give a  Bayesian estimate of the precision matrix for large  graphical Gaussian models with, additionally, symmetry constraints imposed by an underlying graph which is coloured. We take the sample posterior mean of the precision matrix as our estimate. We study its asymptotic behaviour under the regular asymptotic regime when the number of variables $p$ is fixed and under the double asymptotic regime when both $p$ and $n$ grow to infinity. We show in particular, that when the number of parameters of the local models is uniformly bounded, the standard convergence rate we obtain for the asymptotic consistency, in the Frobenius norm, of our estimate of the precision matrix compares well with the rates in the current literature for the maximum likelihood estimate.

\section{INTRODUCTION}
In this paper, we consider graphical Gaussian models with symmetry constraints. Symmetry restrictions
for the multivariate Gaussian distribution have a long history dating back to \citet{Wilks46} and the reader is referred to \citet{Gehrmann12} for a complete list of references.
Graphical Gaussian models with symmetry restrictions were first considered by \citet{Hylleberg93}: the symmetry restrictions in that paper could be described by a group action. Subsequently, \citet{Andersen95} and \citet{Madsen00} also considered such models.
More recently, \citet{Lau08} considered graphical Gaussian models  with symmetry constraints not necessarily described by a group action. Rather those symmetries are described by coloured graphs ${\cal G}=({\cal V}, {\cal E})$ with skeleton $G=(V,E)$ where $V$ is the set of vertices, $E$ the set of undirected edges, ${\cal V}$ is the set of colour classes for the vertices and ${\cal E}$ the set of colour classes for the edges. The symmetry is given by the equality of certain entries either in the covariance, the correlation or the precision matrices. Models for the multivariate random variable $X=(X_i, i\in V)$ Markov with respect to $G$ and with covariance, precision or correlation matrix following equality constraints given by ${\cal G}$ are called coloured graphical Gaussian models. These models have two main advantages. First they may reflect true or imposed symmetries. For example, variables could represent characteristics of twins (see Frets heads data set, \citet{Frets21}) and therefore the variance of the corresponding variables can be  assumed to be equal. Second, since conditional independences imply that certain entries of the precision matrix are set to zero, these restrictions combined with the symmetry restrictions reduce the number of free parameters and facilitate inference in high-dimensional models. \citet{Lau08} developed algorithms to compute the maximum likelihood estimate of the covariance, correlation or precision matrix.

In \citet{Helene15}, the authors considered  the coloured graphical Gaussian model with symmetry restrictions on the precision matrix and they did so from a Bayesian perspective. In this paper also, we only consider such models which were called RCON models by \citet{Lau08}. A pleasant feature of these models is that,
given a sample $X_1,\ldots,X_n$ from the coloured graphical Gaussian model with precision matrix $K$, the distribution of
the sufficient statistics is a natural exponential family with canonical parameter $K$ and therefore a convenient prior is the \citet{Diaconis79} prior distribution (henceforth abbreviated DY conjugate prior), which we call the coloured $G$-Wishart since it is similar to the $G$-Wishart which is the DY conjugate prior for graphical Gaussian models Markov with respect to an undirected graph.
\citet{Helene15} gave a method to sample from this posterior (or prior, of course) distribution in order to estimate $K$ with the sample mean of the posterior DY conjugate distribution. However, as the dimension of the model increases, the computational times also increase and it is not practically feasible to compute the posterior mean for high-dimensional models.

In order to be able to give a Bayesian estimate of the posterior mean of the precision matrix for high-dimensional models, in this paper, we consider distributed estimation, thus  providing a Bayesian alternative to the distributed estimation of $K$ by maximum likelihood. The idea behind distributed estimation is that the estimation of the parameter $K$ is parsed out to smaller models from which we can estimate part of the parameter of the initial global model. The estimates of parts of the global parameter are then combined together to yield an estimate of the global model. More precisely, if we want to estimate the precision matrix $K$ in a graphical Gaussian model with underlying graph $G=(V,E)$, for each $i\in V$ we consider the set of neighbours $ne(i)$ of $i$. At this point, for each $i$, we could consider either a local conditional or marginal model: the conditional model of $X_i$ given $X_{ne(i)}$ or the marginal model of ${\cal N}_i=\{i\}\cup ne(i)$.  None of the parameters of the local conditional model are equal to part of $K$ while, as we will show in Section 2, the precision matrix $K^i$ of the local ${\cal N}_i$-marginal model is such that
\begin{equation}
\label{conserve}
K^i_{il}=K_{il}, l\in {\cal N}_i.
\end{equation}
As we will see also in Section 2, in order to make the local marginal model as a natural exponential family while keeping the property \eqref{conserve}, we will consider a ``relaxed" ${\cal N}_i$-marginal model. This method was first developed by \citet{Meng14} for graphical Gaussian models and we adapt it here to coloured graphical Gaussian models.

 Having obtained our Bayesian estimate of $K$ using local marginal models and local coloured $G$-Wishart, we will then study its asymptotic properties. We will do so first under the traditional asymptotic conditions, i.e. when the sample size $n$ goes to infinity and the number of variables $p$ is fixed and second under the double asymptotic regime when both $n$ and $p$ go to infinity.

 The study of the asymptotic properties, for $p$ fixed, of the Bayesian estimate goes back to \citet{Bickel69} who proved the convergence of the normalized posterior density to the appropriate normal density as well as the consistency and efficiency of the posterior mean. Since then, a lot of research has been devoted to Bayesian asymptotics for $p$ fixed. One of the most recent and well-known work in that area is \citet{Ghosal95}. For both $p$ and $n$ going to infinity, \citet{Ghosal00} studied the  consistency and asymptotic normality, under certain conditions, of the posterior distribution of the
canonical parameter for an exponential family when the dimension of the parameter
grows with the sample size. \citet{Ghosal00} also indicates that under additional conditions, the difference between the normalized posterior mean of the canonical parameter and the normalized sample mean tends to $0$ in probability.

 We will prove in this paper first that, for $p$ fixed, our estimate is consistent and asymptotically normally distributed, second that, for both $p$ and $n$ going to infinity, under certain boundedness conditions and for $\frac{p^{13}(\log p)^2}{\sqrt{n}}\to 0$, our estimate tends, in Frobenius norm,  to the true value of the parameter with probability tending to $1$. For $p$ fixed our arguments are classical arguments adapted to our distributed estimate. Under the double asymptotic regime, there are three main features to our proofs. For each local model, we follow an argument similar to that given in \citet{Ghosal00}. We therefore need to verify that our DY conjugate prior and our sampling distribution satisfy the conditions and properties assumed by \citet{Ghosal00} in his arguments.  The second feature is that, in the process of proving that the norm of the difference between our estimate and the true value of the parameter tends to $0$, we need to prove that asymptotically, our sampling distribution  satisfies the so-called cumulant-boundedness condition. To do so, we use an argument similar to that developed by \citet{Xin15} who, in turn, were inspired by the sharp deviation bounds given by \citet{Spokoiny13} for $n$ fixed. Finally, we have to combine the results obtained for each local model to show our result for the estimate of the global parameter.

 From the condition $\frac{p^{13}(\log p)^2}{\sqrt{n}}\to 0$ mentioned above, it would appear that for $p$ large, $n$ would have to be extremely large to achieve asymptotic consistency of distributed Bayesian estimate with high probability. This condition is given under the assumption that the number of parameters in each local model is also allowed to grow with $p$ and $n$. When doing estimation, we are actually given the model and it is then reasonable to assume that the number of parameters in the local models is uniformly bounded. Under that more relaxed assumption, as we shall see in Section 5, we obtain a relative rate of growth of $p$ and $n$ such that $\frac{\log^4 p\log \log p}{\sqrt{n}}\to 0$ which, as will be shown, compares well to the rate given by \citet{Meng14} for the distributed estimation of the maximum likelihood estimate in graphical Gaussian models.

 In Section 2, we  recall definitions and basic properties of coloured graphical models and distributed computing. We briefly recall the scheme for sampling from the posterior coloured $G$-Wishart. In Section 3, we study the asymptotic properties of our estimate when $p$ is fixed. In Section 4, we study the asymptotic properties under the double asymptotic regime.
 In Section 5, we examine how the results of Section 4 are modified when we assume that the number of parameters in the local marginal models is uniformly bounded.
  In Section 6, we illustrate the efficacy of our method to obtain the posterior mean of $K$ using several simulated examples. We demonstrate numerically how our method can scale up to any dimension by looking at coloured graphical Gaussian model governed by large coloured cycles and also by a coloured  $10\times 10$ grid.

 \section{Preliminaries}
\subsection{Coloured graphical models}
Let $X_{1}, X_{2}, \ldots, X_{n}$ be independent and identically distributed $p$-dimensional random variables following a multivariate normal distribution $N_{p}(0, \Sigma)$ with $X_{i}=(X_{i1},X_{i2},\ldots ,X_{ip})^t$, $i=1,2,\ldots,n$. Let $K=\Sigma^{-1}$ be the precision matrix and $G=(V,E)$ be an undirected graph where $V=\{1,2,\ldots,p\}$ and $E$ are the sets of vertices and edges, respectively. For $X=(X_{i},i\in V)$, we say that the distribution of $X$ is Markov with respect to $G$ if $X_{i}\perp X_{j}|X_{V\backslash \{i,j\}}$ is implied by the absence of an edge between $i$ and $j$ in the graph $G$. Such models for $X$ are called graphical Gaussian models. Since, as it is well-known, conditional independence of the variables $X_{i}$ and $X_{j}$ is equivalent to $K_{ij}=0$, if we denote $P_{G}$ as the cone of positive definite matrices with zero $(i,j)$ entry whenever the edge $(i,j)$ does not belong to $E$, then the graphical Gaussian model Markov with respect to $G$ can be represented as
\begin{eqnarray}\label{Markov}
\mathcal{N}_{G}=\{N(0, \Sigma)|K\in P_{G}\}.
\end{eqnarray}
\citet{Lau08} introduced the coloured graphical Gaussian models with additional symmetry on $K$ as follows. Let $\mathcal{V}$$=\{V_{1},  V_{2}, \ldots, V_{T}\}$ form a partition of $V$ and $\mathcal{E}$$=\{E_{1},  E_{2}, \ldots, E_{S}\}$ form a partition of $E$. If all the vertices belonging to an element $V_l, l=1,\ldots,T$, of $\mathcal{V}$ have the same colour, we say $\mathcal{V}$ is a colouring of $V$. Similarly if all the edges belonging to an element $E_t, t=1,\ldots, S$, of $\mathcal{E}$ have the same colour, we say that $\mathcal{E}$ is a colouring of $E$. We call $\mathcal{G}=(\mathcal{V}, \mathcal{E})$ a coloured graph. Furthermore, if the model \eqref{Markov} is imposed with the following additional restrictions
\begin{enumerate}
\item[(a)] if $m$ is a vertex class in $\mathcal{V}$, then for all $i\in m$, $K_{ii}$ are equal, and
\item[(b)] if $s$ is an edge class in $\mathcal{E}$, then for all $(i,j)\in s$, $K_{ij}$ are equal,
\end{enumerate}
then the model is defined as a coloured graphical Gaussian model RCON $({\mathcal{V}, \mathcal{E}})$ and denoted as
\begin{eqnarray*} \label{colored}
\mathcal{N}_{\mathcal{G}}=\{N(0, \Sigma)|K\in P_{\mathcal{G}}\}
\end{eqnarray*}
where $P_{\mathcal{G}}$ is the cone of positive symmetric matrix with zero and colour constraints.

We operate within a Bayesian framework. The prior for $K$ will be the coloured $G$-Wishart with density
$$\pi(K|\delta, D)=\frac{1}{I_{G}(\delta, D)}|K|^{(\delta-2)/2}\exp\{-\frac{1}{2} tr(KD)\}\mathbf{1}_{K \in P_{\mathcal{G}}},$$
where $\delta > 0$ and $D$, a symmetric positive definite $p\times p$ matrix, are the hyper parameters of the prior distribution, $\mathbf{1}_{A}$ denotes the indicator function of the set $A$ and $I_{\mathcal{G}}(\delta, D)$ is the normalizing constant, namely,
$$I_{ \mathcal{G}}(\delta, D)=\int_{P_{\mathcal{G}}}|K|^{(\delta-2)/2}\exp\{-\frac{1}{2} tr(KD)\}dK.$$
In the previous expression, $tr(\cdot)$ represents the trace and $|K|$ represents the determinant of a matrix $K$.

\citet{Helene15} proposed a sampling scheme for the coloured $G$-Wishart distribution. This sampling method is based on Metropolis-Hastings algorithm and Cholesky decomposition of matrices. To develop this sampling method, the authors make the change of variable from $K$ to $\Psi=\Phi Q^{-1}$ where $D^{-1}=Q^tQ$ and $K=\Phi^t\Phi$ are the Cholesky decomposition of $D^{-1}$ and $K$ respectively with $Q$ and $\Phi$ upper triangular matrices with real positive diagonal entries. The superscript $t$ denotes the transpose. Then the zero and colour constraints on the entries of $K$ associated with a coloured graph $\mathcal{G}$ determine the sets of free entries in $\Psi$. The proposal distribution is a product of normal distributions and chi-square distributions of the free elements in $\Psi$. Finally, we complete the non-free elements in $\Psi$ as a function of the free elements and obtain the sampled $K=Q^t(\Psi^t\Psi)Q$.

\subsection{Local relaxed marginal model}
For a given vertex $i\in V$, define the set of immediate neighbors of vertex $i$ as $ne(i)=\{j|(i,j)\in E\}$. For each $i\in V$, we consider two types of neighbourhood of $i$, the so-called one-hop and two-hop neighbourhood. The one-hop neighbourhood $N_i=\{i\}\cup ne(i)$ is made up of $i$ and the vertices directly connected to it. The two-hop neighbourhood $N_i=\{i\}\cup ne(i)\cup \{k\mid (k,j)\in E,j\in ne(i)\}$ consists of $i$, its neighbours and the neighbours of the neighbours. Without risk of confusion, we let $N_{i}$ denote either a one-hop or two-hop neighbourhood. Consider the local marginal model for $X_{N_{i}}=\{X_{v}, v\in N_{i}\}$ which is abbreviated as $X^{i}$. This is a Gaussian model with precision matrix denoted by $\mathcal{K}^{i}$. Then
\begin{eqnarray}\label{inverse}
\mathcal{K}^{i}=(\Sigma_{N_{i},N_{i}})^{-1}=K_{N_{i},N_{i}}- K_{N_{i},V\backslash N_{i}}[K_{V\backslash N_{i},V\backslash N_{i}}]^{-1}K_{V\backslash N_{i},N_{i}}.
\end{eqnarray}
Based on the collection of vertices $N_{i}$ and its complement set $V\backslash N_{i}$, we partition $N_{i}$ further into two subsets. One is the buffer set $B_{i}=\{j|j \in N_{i}$ and $ne(j)\cap (V\backslash N_{i})\neq \emptyset \}$, which are the vertices having edges connecting to the complement of $N_i$ in $V$. The other is the protected set $\mathcal{P}_{i}=N_{i}\backslash B_{i}$, which are the vertices in $N_{i}$ that are not directly connected to $V\backslash N_{i}$.
Since the distribution of $X$ is Markov with respect to $G$, then $X_{\mathcal{P}_{i}} \perp X_{V\backslash N_{i}}|X_{B_{i}}$ and it follows that
\begin{eqnarray}\label{conditional}
K_{\mathcal{P}_{i}, V\backslash N_{i}}=0.
\end{eqnarray}
Then equation \eqref{inverse} becomes
\begin{eqnarray*}\label{global}
&&\left( \begin{array}{cc}
\mathcal{K}^{i}_{\mathcal{P}_{i},\mathcal{P}_{i}} &\mathcal{K}^{i}_{\mathcal{P}_{i},B_{i}} \\
\mathcal{K}^{i}_{B_{i},\mathcal{P}_{i}} & \mathcal{K}^{i}_{B_{i},B_{i}}\\\end{array} \right)\\\nonumber
&=&\left( \begin{array}{cc}
K_{\mathcal{P}_{i},\mathcal{P}_{i}} &K_{\mathcal{P}_{i},B_{i}} \\
K_{B_{i},\mathcal{P}_{i}} & K_{B_{i},B_{i}}\\\end{array} \right)-\left( \begin{array}{c}
K_{\mathcal{P}_{i},V\backslash N_{i}}  \\
K_{B_{i},V\backslash N_{i}}  \\\end{array} \right)(K_{V\backslash N_{i},V\backslash N_{i}})^{-1}\left( \begin{array}{cc}
K_{V\backslash N_{i}, \mathcal{P}_{i}} & K_{V\backslash N_{i},B_{i}} \\ \end{array} \right)\\\nonumber
&=&\left( \begin{array}{cc}
K_{\mathcal{P}_{i},\mathcal{P}_{i}} &K_{\mathcal{P}_{i},B_{i}} \\
K_{B_{i},\mathcal{P}_{i}} & K_{B_{i},B_{i}}\\\end{array} \right)-\left( \begin{array}{cc}
0 &0 \\
0 & K_{B_{i},V\backslash N_{i}}(K_{V\backslash N_{i},V\backslash N_{i}})^{-1}K_{V\backslash N_{i},B_{i}}\\\end{array} \right)\\\nonumber
\end{eqnarray*}
where the 0's in the matrix above follows from the identity \eqref{conditional}. Therefore, we obtain the following relationships
\begin{eqnarray*}
\mathcal{K}^{i}_{\mathcal{P}_{i},\mathcal{P}_{i}}=K_{\mathcal{P}_{i},\mathcal{P}_{i}}, \hspace{24mm}\mathcal{K}^{i}_{\mathcal{P}_{i},B_{i}}=K_{\mathcal{P}_{i},B_{i}}\hspace{8mm}\label{relax1} \\ \mathcal{K}^{i}_{B_{i},B_{i}}=K_{B_{i},B_{i}}-K_{B_{i},V\backslash N_{i}}(K_{V\backslash N_{i},V\backslash N_{i}})^{-1}K_{V\backslash N_{i},B_{i}}.\label{relax2}
\end{eqnarray*}
This shows that the local parameters of $\mathcal{K}^{i}$ indexed by $(\mathcal{P}_{i},\mathcal{P}_{i})$ and $(\mathcal{P}_{i},B_{i})$ are equal to the corresponding global ones but the same does not hold for those indexed by
$(B_i,B_i)$. This important observation motivates us to use the $N_i$-marginal local models to estimate those parameters which are identical in both local and global models.

We denote by ${\cal G}_i$ the coloured graph with vertex set $N_i$ and edge set
$$E_i=E\cap \{
\{
{\cal P}_i\times {\cal P}_i\}
\cup \{{\cal P}_i\times B_i\}
\cup \{B_i\times {\cal P}_i\}
\}
\cup
 \{B_i\times B_i\}.$$
In ${\cal G}_i$, the colours of the vertices in $N_i\setminus B_i$ are the same as the corresponding ones in ${\cal G}$. The colours of the edges in $E_i\backslash \{B_{i},B_{i}\}$ are the same as the corresponding ones in ${\cal G}$. The colours of the vertices in $B_i$ and the edges in $B_i\times B_i$ are arbitrary without constraints.
Let $K^i$ denote the precision matrix of this relaxed local marginal model. We thus keep the important relationships
$${\cal K}^i_{{\cal P}_i, {\cal P}_i}=
K_{{\cal P}_i, {\cal P}_i},\;\;
{\cal K}^i_{{\cal P}_i, B_i}=
K_{{\cal P}_i, B_i}$$
and have a local Gaussian model with canonical parameter $K^i$ on which we can put a local coloured $G$-Wishart distribution.
In each local model Markov with respect to $\mathcal{G}_{i}$, $i \in \{1,2,\ldots,p\}$, we use the method proposed by \citet{Helene15} to obtain the Bayesian estimator $\tilde{K}^{i}$, the sample posterior mean of $K^{i}$ with prior distribution the coloured $G$-Wishart.

Next, we will show how to construct a distributed Bayesian estimate by combining local Bayesian estimates. Let $\theta=(\theta_{V_{1}},\theta_{V_{2}},\ldots,\theta_{V_{T}},\theta_{E_{1}},\theta_{E_{2}},\ldots,\theta_{E_{S}})^{t}$ denote the global parameter, that is the ``free" entries of $K$ which represent the vertex class or the edge class, and let $\theta_0$ be its true value. In each local model $\mathcal{G}_{i}$, we define the local parameter as $\theta^{i}=(\theta^{i}_{1},\theta^{i}_{2},\ldots,\theta^{i}_{S_{i}})^t$, the vector of free entries of $K^i$, and the corresponding local estimator as $\tilde{\theta}^{i}$.
The true value of $\theta^{i}$ is denoted by $\theta^{i}_{0}$.
Furthermore, we collapse all the local parameters into one vector
\begin{eqnarray*}\label{thetabar}
\bar{\theta}=((\tilde{\theta}^{1})^{t},(\tilde{\theta}^{2})^{t},\ldots,(\tilde{\theta}^{p})^{t} )^t
\end{eqnarray*}
and
its true value is denoted as $\bar{\theta}_{0}$. After obtaining the local estimators, a distributed estimate of $\tilde{\theta}$ can be constructed as
$$\tilde{\theta}_{V_{k}}=g_{V_{k}}(\bar{\theta})=\frac{1}{|V_{k}|}\sum\limits_{i\in V_{k}}\sum\limits^{S_{i}}_{j=1}\tilde{\theta}^{i}_{j}\mathbf{1}_{\theta^{i}_{j}=\theta_{V_{k}}}, \hspace{8mm}k=1,2,\ldots,T,$$
and $$\tilde{\theta}_{E_{k}}=g_{E_{k}}(\bar{\theta})=\frac{1}{2|E_{k}|}\sum\limits_{i \in G_{k}}\sum\limits^{S_{i}}_{j=1}\tilde{\theta}^{i}_{j}\mathbf{1}_{\theta^{i}_{j}=\theta_{E_{k}}}, \hspace{8mm}k=1,2,\ldots,S,$$
where $G_{k}=\{i|\exists h\in N_{i}, (i,h)\in E_{k}\}$.
Define the global distributed Bayesian estimate $$\tilde{\theta}=g(\bar{\theta})=(g_{V_{1}}(\bar{\theta}), g_{V_{2}}(\bar{\theta}), \ldots, g_{V_{T}}(\bar{\theta}), g_{E_{1}}(\bar{\theta}), g_{E_{2}}(\bar{\theta}), \ldots, g_{E_{S}}(\bar{\theta}))^{t}.$$

\section{The Bayesian estimator $\tilde{\theta}$ when $p$ is fixed and $n\rightarrow \infty$}
Let $\xrightarrow{\pounds}$ and $\xrightarrow{p}$ denote the convergence in distribution and in probability, respectively. In each local model corresponding to the vertex $i \in \{1,2,\ldots,p\}$, let $L^{i}(\theta^{i})$ and $l^{i}(\theta^{i})$ denote the likelihood and log likelihood, respectively. The Fisher information is denoted by $I^{i}(\theta^{i})=E_{\theta^{i}}[\frac{\partial }{\partial \theta^{i}}l^{i}(\theta^{i}|X^{i})[\frac{\partial }{\partial \theta^{i}}l^{i}(\theta^{i}|X^{i})]^{t}]$. Define a $S_{i}$-dimensional vector $U_{ij} = \frac{1}{\sqrt{n}}[I^{i}(\theta_{0}^{i})]^{-1}\frac{\partial l^{i}(\theta^{i}|X^{i}_{j})}{\partial
\theta^{i}}\big|_{\theta^{i}=\theta_{0}^{i}}$ for $j=1,\ldots, n$ and $i=1,\ldots, p$, a $\sum\limits_{i=1}^{p}S_{i}$-dimensional vector $U_{j}=(U^{t}_{1j}, U^{t}_{2j}, \ldots, U^{t}_{pj})^t$ and $\bar{G}= nCov(U_{1})$. For each $r$, $r=1,2,\ldots,S_{i}$, let $\delta^{i}_{r}$ be the $S_{i} \times S_{i}$ indicator matrix with $(\delta^{i}_{r})_{hl}=1$ if $K^{i}_{hl}=\theta^{i}_{r}$ and 0 otherwise. The following theorem shows that the global estimator has the property of asymptotic normality when the number of variables $p$ is fixed and the sample size $n$ goes to infinity.
\begin{theorem}{}{}%
\label{pfix}
Let $\theta_{0}$, $\tilde{\theta}$, $\bar{\theta}$ and $\bar{G}$ be defined above. Then
$$\sqrt{n}({\tilde{\theta}}-\theta_{0})\xrightarrow{\pounds} N(0,A) \ \ \ \  \text{as}\ \ \ \  n \rightarrow \infty$$ where $A=\frac{\partial g(\bar{\theta})}{\partial \bar{\theta}^t}\bar{G}(\frac{\partial g(\bar{\theta})}{\partial \bar{\theta}^t})^t$.
\end{theorem}

All proofs of theorems and lemmas necessary to their proof are given in the Appendix.
We now establish a result similar to Theorem \ref{pfix} but with the MLE replacing the posterior mean. Based on the same local models, we compute the local MLE $\hat{\theta}^{i}$ of $\theta^{i}$ and obtain a distributed MLE, which is denoted by $\hat{\theta}$.
\begin{theorem}{}{}%
\label{covariance}
Let $\hat{\theta}$ be the distributed MLE. Then $$\sqrt{n}({\hat{\theta}}-\theta_{0})\xrightarrow{\pounds} N(0,A) \ \ \ \  \text{as}\ \ \ \  n \rightarrow 0$$
where $A$ is defined as in Theorem \ref{pfix} above.
\end{theorem}

The distributed MLE is computed by the method of \citet{Meng14} using the local relaxed marginal models defined above. We thus see that the distributed Bayesian estimator $\tilde{\theta}$ has the same limiting distribution as the distributed MLE $\hat{\theta}$.

\section{The Bayesian estimator under the double asymptotic regime $p \rightarrow \infty$ and $n\rightarrow \infty$}
In this section, we study the consistency of the global estimator $\tilde{\theta}$ when both $p$ and $n$ go to infinity. For a vector $x=(x_{1}, x_{2}, \ldots, x_{p})$, let $||x||$ stand for its Euclidean norm $(\sum\limits_{i=1}^{p}x_{i}^{2})^{1/2}$. For a square $p\times p$ matrix $A$, let $||A||$ be its operator norm defined by $\sup\{||Ax||:||x||\leq 1\}$, let $||A||_{F}$ be its Frobenius norm defined by $||A||_{F}=(\sum\limits_{j=1}^{p}\sum\limits_{k=1}^{p}|a_{jk}|^2)^{\frac{1}{2}}$, and let $\lambda(A)$, $\lambda_{min}(A)$ and $\lambda_{max}(A)$ be the eigenvalues, the smallest eigenvalues and largest eigenvalues of $A$, respectively. The vector obtained by stacking columnwise the entries of $A$ is denoted by $vec(A)$. Let $I_{p}$ be the identity matrix with $p$ dimension.
In the local model Markov with respect to ${\cal G}_i$ as defined in
Section 2.2 above, we write the density of $X^{i}_{j}$, $j=1,2,\ldots, n$, as
\begin{eqnarray*}\label{density}
f(X^{i}_{j}; K^{i})&=&\frac{|K^{i}|^{\frac{1}{2}}\exp\big\{-\frac{1}{2}tr(K^{i} X^{i}_{j}(X^{i}_{j})^{t})\big\}}{(2\pi)^{\frac{p_{i}}{2}}}\mathbf{1}_{K^{i} \in P_{\mathcal{G}_{i}}}
\end{eqnarray*}
where $p_{i}=|N_{i}|$. The normalized local coloured $G$-Wishart distribution of $K^{i}$ is denoted by
$$\pi^{i}(K^{i}|\delta^{i}, D^{i})=\frac{1}{I^{i}_{\mathcal{G}_{i}}(\delta^{i}, D^{i})}| K^{i}|^{(\delta^{i}-2)/2}\exp\{-\frac{1}{2} tr(K^{i} D^{i}) \}\mathbf{1}_{K^{i} \in P_{\mathcal{G}_{i}}},$$ where $I^{i}_{\mathcal{G}_{i}}(\delta^{i}, D^{i})$ is the normalizing constant.
In order to obtain our results, we will follow an argument similar to that of \citet{Ghosal00} which gives the asymptotic distribution of the posterior mean when both the dimension $p$ of the model and the sample size $n$ go to $\infty$. \citet{Ghosal00} considers a random variable $X$ with density belonging to the natural exponential family
$$f(x;\theta)\propto\exp[x^t\theta-\psi(\theta)],$$ where $x$ is the canonical statistic, $\theta$ is the canonical parameter and $\psi(\theta)$ is the cumulant generating function.
To follow the notations of \citet{Ghosal00}, we define an $S_{i}$-dimensional vector
\begin{eqnarray}\label{sufficient}
Y^{i}_{j}=-\frac{1}{2}(tr( \delta^{i}_{1} X^{i}_{j}(X^{i}_{j})^{t}), tr( \delta^{i}_{2} X^{i}_{j}(X^{i}_{j})^{t}), \ldots, tr(\delta^{i}_{S_{i}} X^{i}_{j}(X^{i}_{j})^{t}))^{t},
\end{eqnarray}
where $\delta^{i}_{1}, \delta^{i}_{2}, \cdots, \delta^{i}_{S_{i}}$ are indicator matrices for each colour class, $j=1,2,\ldots,n$. The distribution of $Y^{i}_{j}$ is as follows
\begin{eqnarray*}\label{density1}
f(Y^{i}_{j}; K^{i})\propto\exp\big[-\frac{1}{2}tr(K^{i}X^{i}_{j}(X^{i}_{j})^{t})+\frac{1}{2}\log |K^{i}| \big]
=\exp\big[(Y_{j}^{i})^t\theta^{i}-\psi(\theta^{i})\big]
\end{eqnarray*}
where $\psi(\theta^{i})=-\frac{1}{2}\log |K^{i}|$ is the cumulant generating function. From standard properties of natural exponential families, we have that
\begin{eqnarray}\label{F}
\mu^{i} = \psi'(\theta^{i}_{0})\hspace{8mm} \text{and} \hspace{8mm}F^{i}=\psi''(\theta^{i}_{0})
\end{eqnarray}
are the mean vector and the covariance matrix of $Y_{j}^{i}$, $j=1,2,\ldots,n$, respectively. Let $J^{i}$ be a square root of $F^{i}$, i.e. $J^{i}(J^{i})^{t}=F^{i}$.
Let
\begin{eqnarray}\label{stand}
V_{j}^{i}=(J^{i})^{-1}(Y_{j}^{i}-E_{\theta^{i}}(Y_{j}^{i}))
\end{eqnarray} be the standardized version of the canonical statistic. Following \citet{Ghosal00}, for any constant $c$, $c>0$, we define
$$B^{i}_{1n}(c)=\sup\{E_{\theta^{i}}|a^{t}V_{j}^{i}|^{3}: a \in \mathbb{R}^{S_{i}}, ||a||=1, ||J^{i}(\theta^{i}-\theta^{i}_{0})||^{2}\leq \frac{cS_{i}}{n}\}$$
and
$$B^{i}_{2n}(c)=\sup\{E_{\theta^{i}}|a^{t}V_{j}^{i}|^{4}: a \in \mathbb{R}^{S_{i}}, ||a||=1, ||J^{i}(\theta^{i}-\theta^{i}_{0})||^{2}\leq \frac{cS_{i}}{n}\}.$$
Define also $$u^{i}=\sqrt{n}J^{i}(\theta^{i}-\theta^{i}_{0}),$$ then $\theta^{i}=\theta^{i}_{0}+n^{-1/2}(J^{i})^{-1}u^{i}$. Therefore, the likelihood ratio can be written as a function of $u^{i}$ in the following form
\begin{eqnarray*}\label{znu}
Z^{i}_{n}(u^{i})=\frac{\prod\limits_{j=1}^{n}f(Y_{j}^{i};\theta^{i}) }{\prod\limits_{j=1}^{n}f(Y_{j}^{i};\theta_{0}^{i})}=\exp \{\sqrt{n}(\bar{Y}^{i})^t(J^{i})^{-1}u^{i}-n[\psi(\theta^{i}_{0}+n^{-\frac{1}{2}}(J^{i})^{-1}u^{i})-\psi(\theta^{i}_{0})]\},
\end{eqnarray*}
where $\bar{Y}^{i}=\frac{1}{n}\sum\limits^{n}_{j=1}Y^{i}_{j}$. Furthermore, we denote
\begin{eqnarray}\label{delta}
\Delta^{i}_{n}=\sqrt{n}(J^{i})^{-1}(\bar{Y}^{i}-\mu^{i}).
\end{eqnarray}
The following three conditions will be assumed.
\begin{enumerate}
\item[(1)] The orders of $\log p$ and $\log n$ are the same, i.e. $\frac{\log p}{\log n}\rightarrow \zeta>0$ as $n\rightarrow \infty$.
\item[(2)] There exists two constants $\kappa_{1}$ and $\kappa_{2}$ such that $0< \kappa_{1}\leq \lambda_{min}(K_{0})< \lambda_{max}(K_{0})\leq \kappa_{2}<\infty$.
\item[(3)] For any $i \in \{1,2,\ldots,p\}$, the numbers $\tau^{i}$ of the entries $K_{jk}^{i}$ in the same colour class is bounded.
\item[(4)] As $p\rightarrow 0$, the sample size satisfies the rate $\frac{p^{13}(\log p)^{2}}{n^{\frac{1}{2}}}\rightarrow 0$.
\end{enumerate}

\begin{remark}{}{}%
Condition (2) implies $0< \frac{1}{\kappa_{2}}\leq \lambda_{min}(\Sigma_{0})< \lambda_{max}(\Sigma_{0})\leq \frac{1}{\kappa_{1}}<\infty$. By the interlacing property of eigenvalues, we have that $0< \frac{1}{\kappa_{2}}\leq \lambda_{min}((\Sigma_{0})_{N_{i}, N_{i}})< \lambda_{max}((\Sigma_{0})_{N_{i},N_{i}})\leq \frac{1}{\kappa_{1}}<\infty$ where $N_{i}$ is defined as in section 2.2. Therefore, $0< \kappa_{1}\leq \lambda_{min}((\Sigma_{0})_{N_{i}, N_{i}})^{-1}< \lambda_{max}((\Sigma_{0})_{N_{i},N_{i}})^{-1}\leq \kappa_{2}<\infty$. By the definition \eqref{inverse}, for any $i \in \{1,2,\ldots,p\}$, we have $0< \kappa_{1}\leq \lambda_{min}(K^{i}_{0})< \lambda_{max}(K^{i}_{0})\leq \kappa_{2}<\infty$.
\end{remark}

Our aim in this section is to prove that under the Conditions (1)-(4) when both $p$ and $n$ go to infinity, the distributed estimator $\tilde{\theta}$ tends to $\theta_{0}$ in Frobenius norm with probability tending to 1. \citet{Ghosal00} considered the consistency of the posterior mean for the exponential family. The rate of the convergence depends on three expressions which added together yield an upper bound of the overall error $||\tilde{\theta}-\theta_{0}||$. In each expression, the only random component is $||\Delta^{i}_{n}||$ and $||\Delta^{i}_{n}||=O_{p}(p)$. However, we have an infinite number of local models. In order to use Bonferroni inequality to bound the overall error probability, we need to know the exact tail probability of $P(||\Delta^{i}_{n}||> cp)$, where $c$ is a constant. This leads us to establish a new large deviation result for $||\Delta^{i}_{n}||$ in Lemma \ref{max}. We now state the asymptotic consistency of our proposed estimator in Theorem \ref{converge}.
\begin{theorem}{}{}%
\label{converge}Under Conditions (1)-(4), there exists a constant $c^{*}$ such that
$$||\tilde{\theta}-\theta_{0}||\leq c^{*}\frac{p^{\frac{3}{2}}}{\sqrt{n}}$$
with probability greater than $1-10.4\exp\{-\frac{1}{6}p^2\log p+\log p\}$.
\end{theorem}

\section{The Bayesian estimator under the double asymptotic regime when the dimension of the local models bounded}
We saw in Section 4 what the asymptotic behaviour of $\tilde{\theta}$ is when $S_{i}$ is unbounded under the double asymptotic regime. In this Section, we assume that $S_{i}$ is bounded and we will see that for $\tilde{\theta}$ to be close to $\theta_{0}$, $n$ must grow as a power of $\log p$ rather than as a power of $p$. Indeed, we assume the following conditions:

(4$^{*}$) As $p\rightarrow \infty$, the sample size satisfies the rate $\frac{\log^4 p \log \log p}{\sqrt{n}}\rightarrow 0$.

(5) The number of parameters in each local model is bounded by a constant $S^{*}$, i.e. $S_{i}\leq S^{*}$, $i \in \{1,2,\ldots,p\}$.

The main result is Theorem \ref{local} below.
\begin{theorem}{}{}%
\label{local}Under Conditions (1), (2), (4*) and (5), there exists a constant $c_{1}^{*}$ such that
$$||\tilde{\theta}-\theta_{0}||\leq c_{1}^{*}\frac{p^{\frac{1}{2}}\log p}{\sqrt{n}}$$
with probability greater than $1-10.4\exp\{-\frac{1}{6}\log^2 p+\log p\}$.
\end{theorem}

For the convenience of the reader, we now point out the main difference between the proofs of Theorems \ref{converge} and \ref{local}.

(a) Under Conditions (1), (2) and (5), for any $i \in \{1,2,\ldots,p\}$, the quantities $\log|F^{i}|$ in Proposition \ref{trace}, $B_{1n}^{i}(c)$ and $B_{2n}^{i}(c)$, $i=1,2,\ldots,p,$ in Proposition \ref{exp} are all uniformly bounded because the number of parameters in each local model is uniformly bounded and the eigenvalues of $K^{i}$ are uniformly bounded from above and below.

(b) The equivalent of Lemma \ref{Xin} under our new boundedness condition is Lemma \ref{Xin1} where $||\gamma^{i}||< p$ is replaced by the condition $||\gamma^{i}||< \log p$.

(c) The equivalent of Lemma \ref{max} is Lemma \ref{max1}, the large deviation result is established for $||\Delta^{i}_{n}||^2>3a^2(\log p)^2$ rather than $||\Delta^{i}_{n}||^2>3a^2p^2$.

(d) When $S_{i}$ is unbounded, in Theorem 2.1 of \citet{Ghosal00}, the fact that $||\Delta^{i}_{n}||^2>3a^2p^2$ with probability $1-\varepsilon$ implies $n||\hat{\theta}^{i}-\theta_{0}^{i}||^2>bp^2$ with the same probability $1-\varepsilon$, where $\hat{\theta}^{i}$ is the MLE and $b$ is a constant. For $S_{i}$ bounded, using the new large deviation result $||\Delta^{i}_{n}||^2>3a^2(\log p)^2$ in (c) above, we have the new result of $n||\hat{\theta}^{i}-\theta_{0}^{i}||^2>b'(\log p)^2$ with probability greater than $1-10.4\exp\{-\frac{1}{6}(\log p)^2\}$, where $b'$ is a constant (See Lemma \ref{znubound}).

(e) As a consequence of our choice $||\gamma^{i}||< \log p$ in (b) above, the threshold $M(p)=p^2\log p$ in the proof of Theorem \ref{converge} can be replaced by $M(p)=(\log p)^2(\log \log p)$.

\begin{remark}{}{}%
We note that the error bound $\frac{p^{\frac{1}{2}}\log p}{\sqrt{n}}$ in Theorem \ref{local} is smaller than that $\frac{p^{\frac{3}{2}}}{\sqrt{n}}$ in Theorem \ref{converge}. Also the rate of growth for the sample size is in terms of powers of $\log p$ rather than $p$ as in Section 4.
\end{remark}

\begin{remark}{}{}%
As in \citet{Meng14}, we assume that the graph structure is known. When $S_{i}<S^{*}$, the error bound in our case is of the order $\frac{p\log^2p}{n}$ which compares well with the order $\frac{p\log p}{n}$ in \citet{Meng14}. The sample size requirement $\frac{\log^4 p \log \log p}{\sqrt{n}}\rightarrow 0$ is slightly more demanding than Meng's condition of $n>\log p$.
\end{remark}

\section{Simulations}
In order to evaluate the performance of our proposed distributed Bayesian estimate of $K$, we  conduct a number of experiments using simulated data. For each experiment, we compute the distributed estimator using relaxed local marginal models built on the ``one- hop"  and on the ``two-hop" neighbourhoods of each $i\in \{1,2,\ldots,p\}$. We choose the coloured $G$-Wishart distribution as the prior with hyperparameters $D^{i}=I_{p_{i}}$ and $\delta^{i}=3$ for all $i\in \{1,2,\ldots,p\}$. The corresponding estimators are called the
MBE-1hop and MBE-2hop estimates of $K$ respectively. We consider seven different coloured graphical Gaussian models. The underlying graph of three of those  models are cycles of length $p=20$ with alternate vertex and edge colours  as indicated in Figure \ref{fig:1} (a), (b) and (c). Three other models have the same type of underlying coloured graphs but the cycles are of length $p=30$. The  underlying graph of the seventh model is a $10\times 10$ grid with colours as shown in Figure \ref{fig:1} (d). For both the cycles and the $10\times 10$ grid, black edges or vertices indicate that there is no colour constraint. For the sake of comparison, for models with underlying graphs the cycles of order $20$ or $30$, we also compute the global Bayesian estimate of $K$, denoted GBE, using the method given in \citet{Helene15}. Since asymptotically, the posterior mean is expected to be close to the maximum likelihood estimate of $K$, for all models, we also compute the global MLE  of $K$, denoted GMLE.

The values of $(K_{ij})_{1\leq i,j\leq p}$ used for the simulation for models with underlying graphs as given in Figure \ref{fig:1} (a), (b) and (c) are given in
 Table \ref{table:parameters}.
 For the $10\times 10$ coloured grid-graph of Figure \ref{fig:1} (d), we chose $K_{i+10(j-1), i+1+10(j-1)}=1$ for $i=1, 2,\ldots, 9$ and $j=1, 2,\ldots, 10$, $K_{i+10(j-1), i+10j}=1+0.01i+0.1j$ for $i=1, 2,\ldots, 10$ and $j=1, 2,\ldots, 9$ and $K_{i, i}=10+0.01i$ for $i=1, 2,\ldots, 100$. The posterior mean estimates are based on 5000 iterations after the first 1000 burn-in iterations.

The posterior mean estimates are based on 5000 iterations after the first 1000 burn-in iterations. Table \ref{table:2} shows the normalized mean square error $NMSE(K,\hat{K})=\frac{||\hat{K}-K||^2}{||K||^2}$ for the six models with the coloured cycles as underlying graphs. Values are averaged over 100 data sets from the normal $N(0,K^{-1})$ distribution. Standard deviations are shown in parentheses. From these results, we see that our MBE-1hop and MBE-2hop estimates perform very well compared to the global estimate GBE. In Figure \ref{fig:2} we give the graphs of $NMSE(\hat{K},K)$ in function of sample size, for different sample sizes ranging from 50 to 100 for the four models with underlying graphs the coloured cycles of length $p=20$ and the $10\times 10$ grid. We see that the MLE  and the GBE consistently yields the smallest  and largest NMSE respectively with the MBE-1hop and MBE-2hop in between with the NMSE of the MBE-2hop estimate always smaller than that of the MBE-1hop. As expected, as $n$ increases, all NMSE tend to the same value.

  Computations are performed on a 2 core 4 thread processor with i5-4200U, 2.3 GHZ chips and 8 GB of RAM, running on Windows 8. The average computing times for the estimates of $K$ are given in minutes for the six models with cycles as underlying graphs. We can see that the computation times for the MBE-1hop and MBE 2-hop are much smaller than for the GBE.

\begin{figure}
         \centering
         \begin{subfigure}[b]{0.3\textwidth}
                 \includegraphics[width=\textwidth]{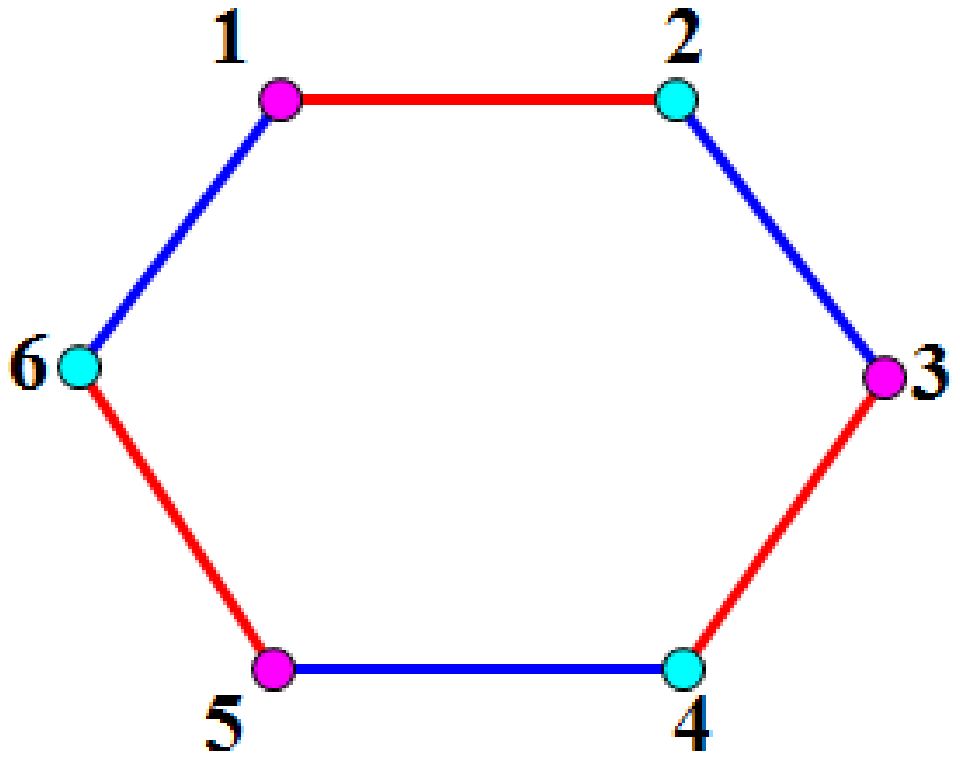}
                 \caption{}
                 \label{fig:vertex and edge}
         \end{subfigure}%
         ~ 
         \begin{subfigure}[b]{0.3\textwidth}
                 \includegraphics[width=\textwidth]{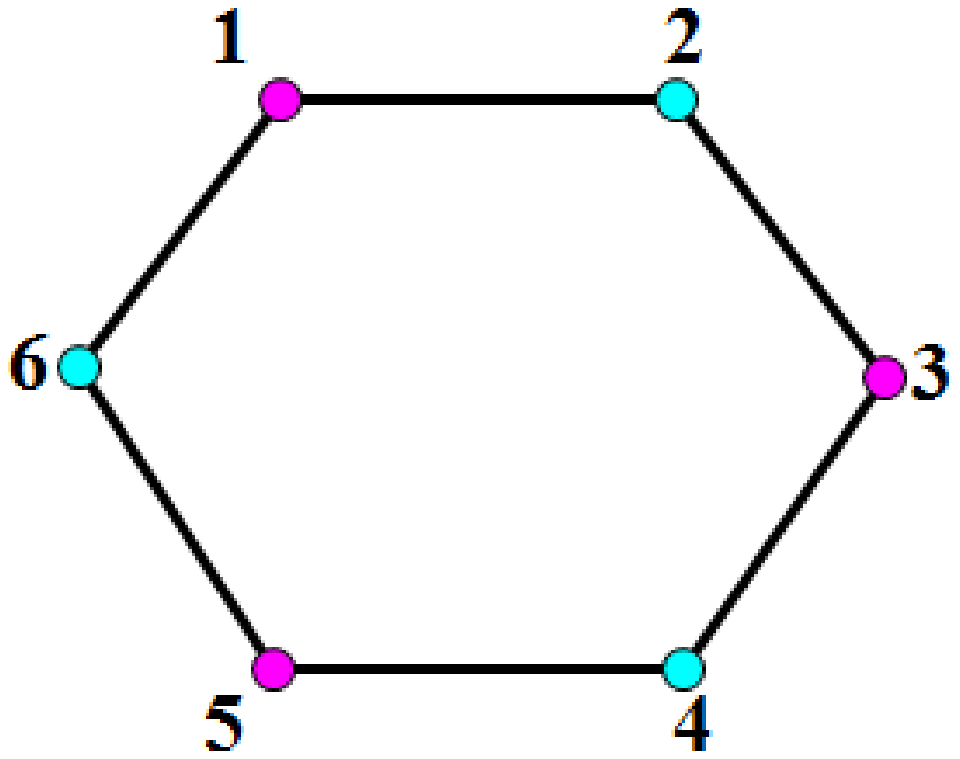}
                 \caption{}
                 \label{fig:vertex}
         \end{subfigure}
         ~ 
         \begin{subfigure}[b]{0.3\textwidth}
                 \includegraphics[width=\textwidth]{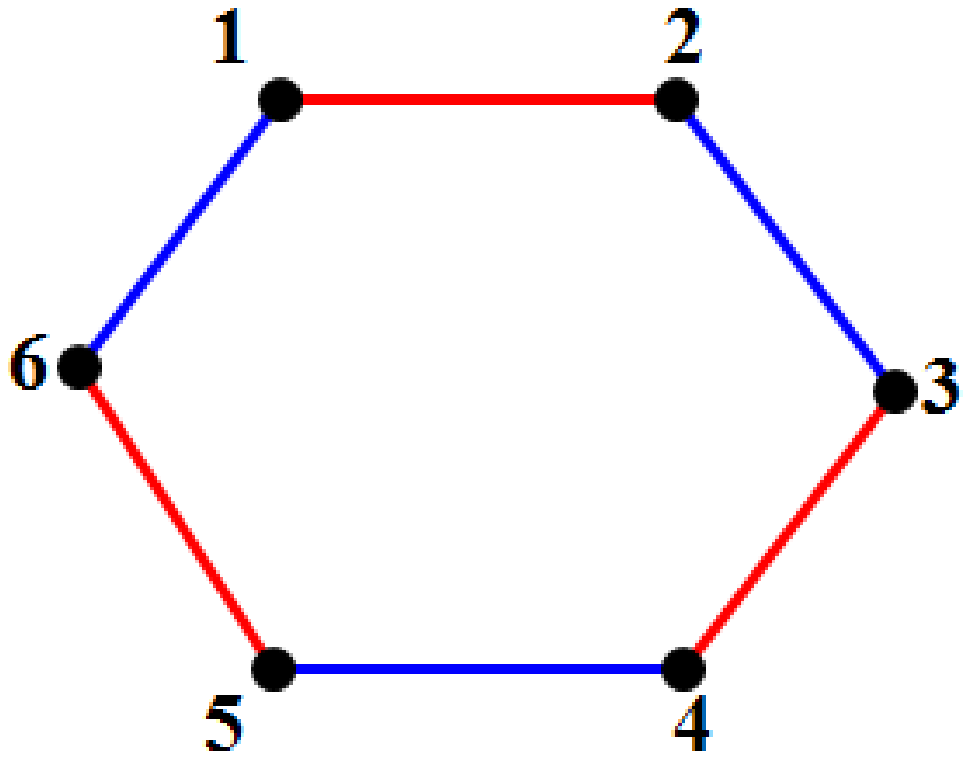}
                 \caption{}
                 \label{fig:edge}
         \end{subfigure}
         ~
         \begin{subfigure}[b]{0.3\textwidth}
                 \includegraphics[width=\textwidth]{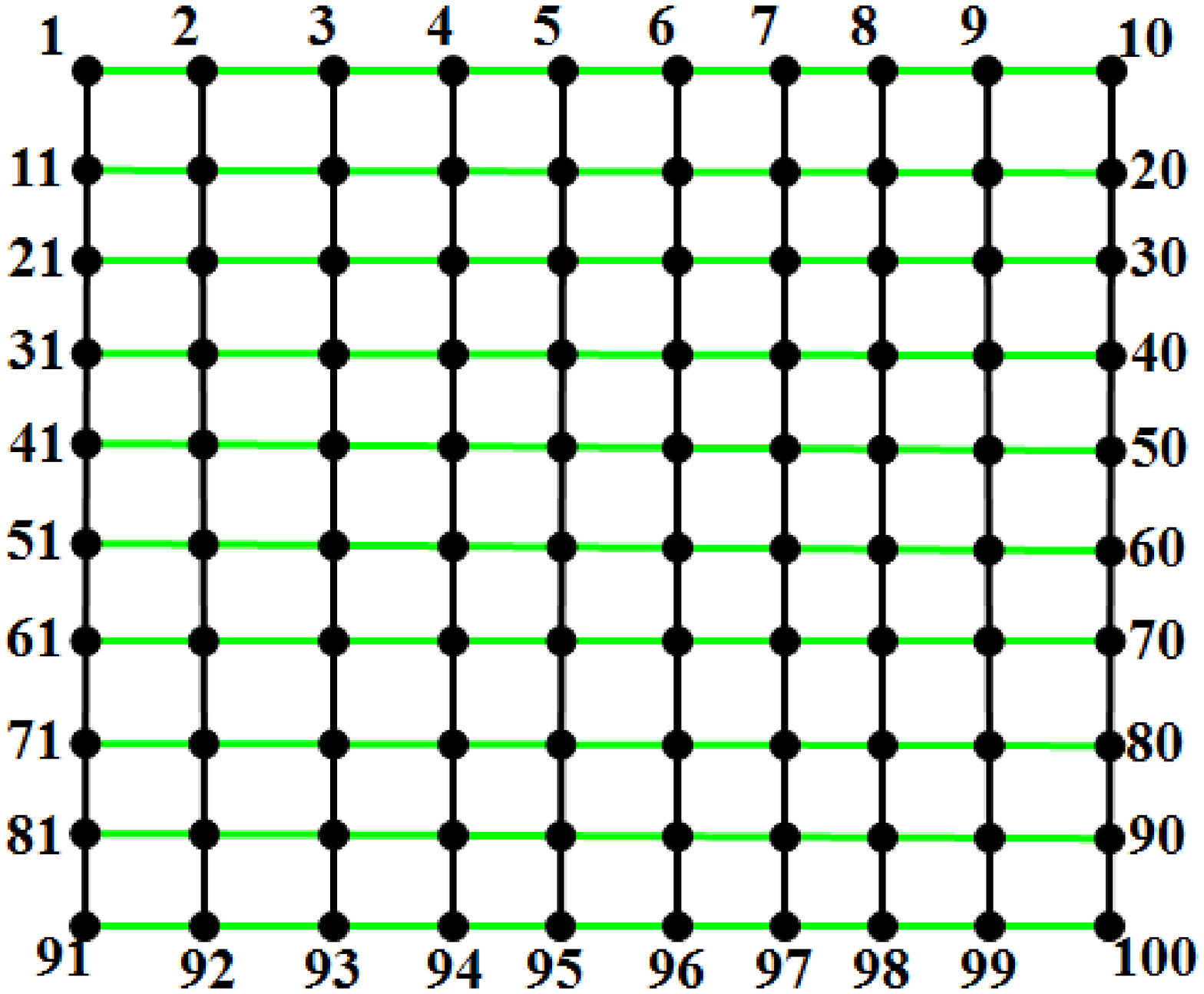}
                 \caption{}
                 \label{fig:vertex}
         \end{subfigure}
         ~ 

         \caption{Cycles of length 6 with the three different patterns of colouring that we use for the cycles of length $p=20$ and $p=30$. Black vertices or edges indicate different arbitrary colours.}
         \label{fig:1}

\end{figure}

\begin{table}[h!]
\centering
\caption{The parameters chosen for the matrix $K$ for producing Figure \ref{fig:1}.}
\begin{tabular}{ccccc}
     \hline
     parameters &Figure \ref{fig:1} (a) &Figure \ref{fig:1} (b) &Figure \ref{fig:1} (c) \\
     \hline
     $K_{ii}$ ($i=1, 3,\ldots, 2p-1$) &  0.1  & 0.1 & 0.1+0.1i\\ \hline
          $K_{ii}$ ($i=2, 4,\ldots, p$) &  0.03  & 0.3 &0.03+0.01i \\ \hline
          $K_{i,i+1}=K_{i+1,i}$ ($i=1, 3,\ldots, 2p-1$) & 0.01  & 0.01+0.001i &0.01\\ \hline
          $K_{i,i+1}=K_{i+1,i}$ ($i=2, 4,\ldots, p-2$) & 0.02 & 0.01+0.002i & 0.02 \\ \hline
          $K_{1p}=K_{p1}$ & 0.02   & 0.01 & 0.02 \\ \hline
   \end{tabular}
\label{table:parameters}
\end{table}

\begin{table}[h!]
\centering
\caption{$NMSE(K,\hat{K})$ for the three coloured models when $p=20$ and $p=30$.}
\begin{tabular}{ccccc}
     \hline
     &&\multicolumn{3}{c}{NMSE} \\
     \hline
     $p$ & $ \mathcal{G}$ & MBE\_1hop &MBE\_2hop & GBE   \\ \hline
          &  (a) & 0.0162 (0.0155) & 0.0032 (0.0027) & 0.0110 (0.0102)  \\
     20   &  (b) & 0.0256 (0.0153) & 0.0148 (0.0058)& 0.0237 (0.0189) \\
          &  (c) & 0.0375 (0.0283) & 0.0305 (0.0142) & 0.0308 (0.0241) \\
     \hline
          &  (a) & 0.0098 (0.0070) &0.0017(0.0014)&0.0317 (0.0571)   \\
     30   &  (b) & 0.0234 (0.0088) &0.0151(0.0054) &0.0482 (0.0533)   \\
          &  (c) & 0.0379 (0.0127) &0.0308 (0.0086) &0.0823 (0.0257)   \\
     \hline
   \end{tabular}
\label{table:2}
\end{table}

\begin{table}[h!]
\centering
\caption{Timing for the three coloured models when $p=20$ and $p=30$.}
\begin{tabular}{ccccc}
     &&\multicolumn{3}{c}{Timing} \\
     \hline
     $p$ & $ \mathcal{G}$ & MBE\_1hop & MBE\_2hop & GBE \\ \hline
          &  (a)  &0.365&3.410 &21.875  \\
     20   &  (b)  &1.047& 3.353&16.249  \\
          &  (c)  &0.944&3.054 &15.513  \\
     \hline
          &  (a)  & 1.442&4.952 &83.965  \\
     30   &  (b)  & 1.538&4.557 &80.255   \\
          &  (c)  & 1.504&4.509 &79.918  \\
    \hline
   \end{tabular}
\label{table:3}
\end{table}

\begin{figure}
         \centering
         \begin{subfigure}[b]{0.4\textwidth}
                 \includegraphics[width=\textwidth]{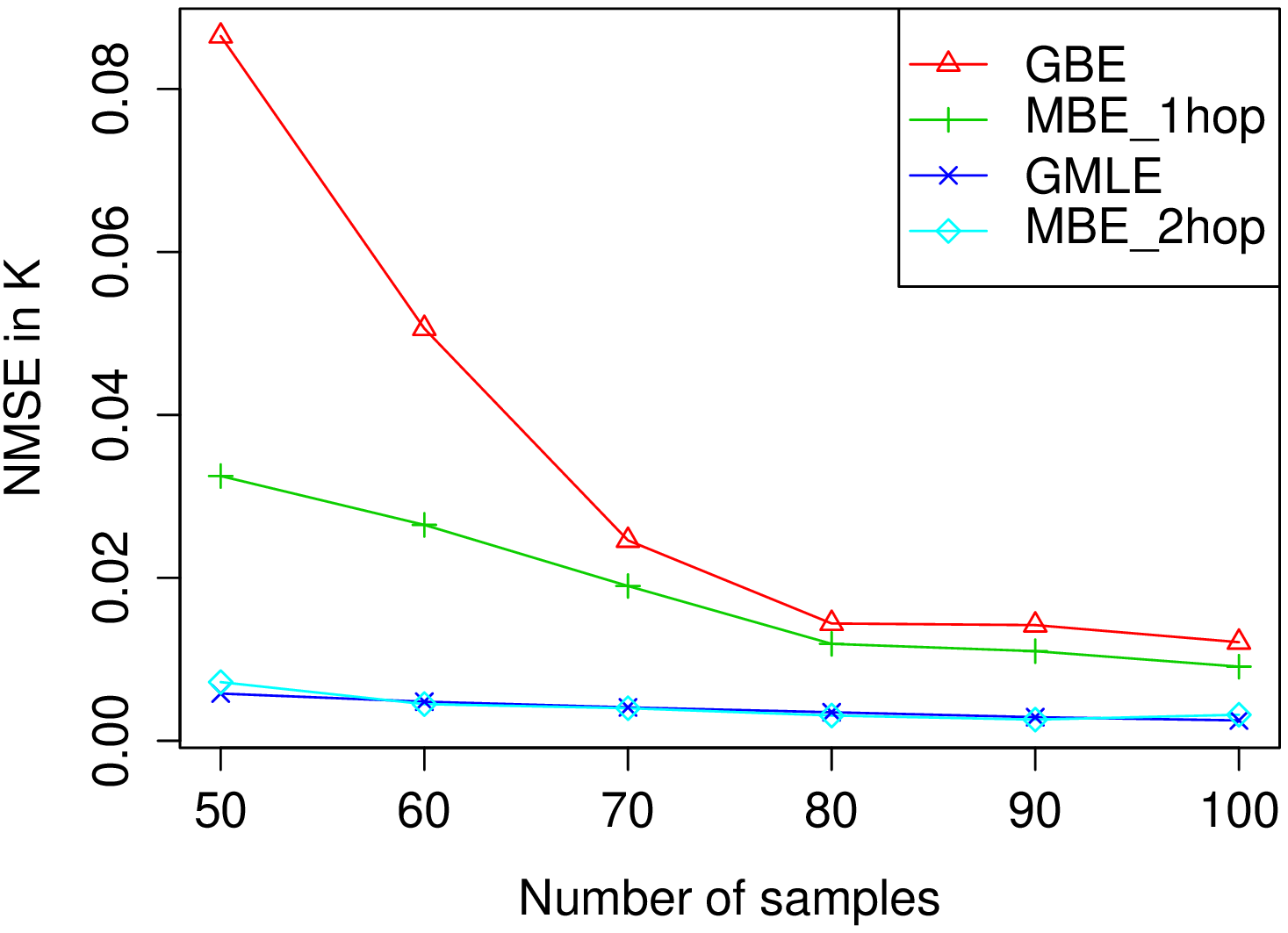}
                 \caption{}
                 \label{fig:vertex and edge 20}
         \end{subfigure}%
         ~ 
         \begin{subfigure}[b]{0.4\textwidth}
                 \includegraphics[width=\textwidth]{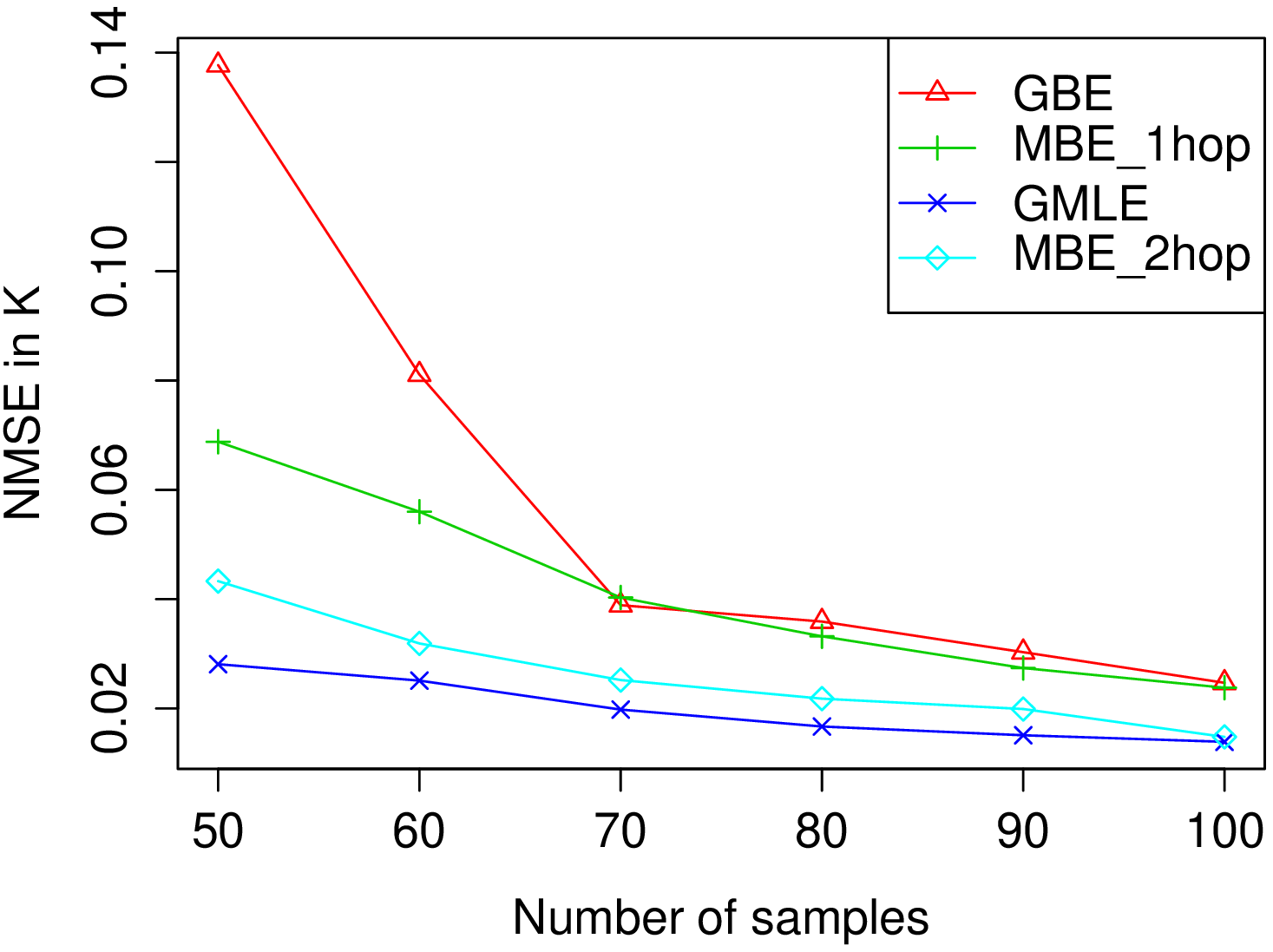}
                 \caption{}
                 \label{fig:vertex 20}
         \end{subfigure}
         ~ 
         \begin{subfigure}[b]{0.4\textwidth}
                 \includegraphics[width=\textwidth]{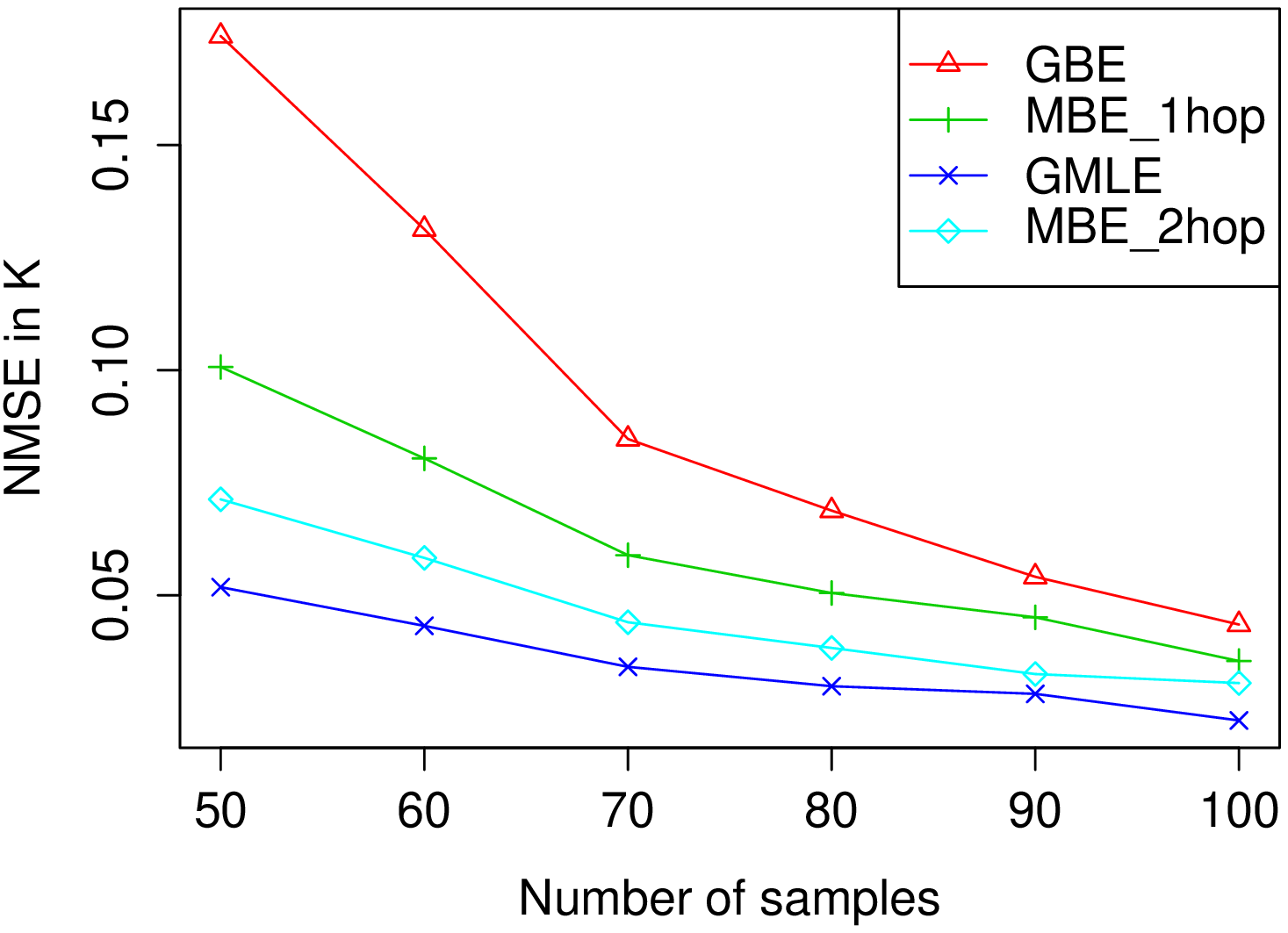}
                 \caption{}
                 \label{fig:edge 20}
         \end{subfigure}
         ~ 
         \begin{subfigure}[b]{0.4\textwidth}
                 \includegraphics[width=\textwidth]{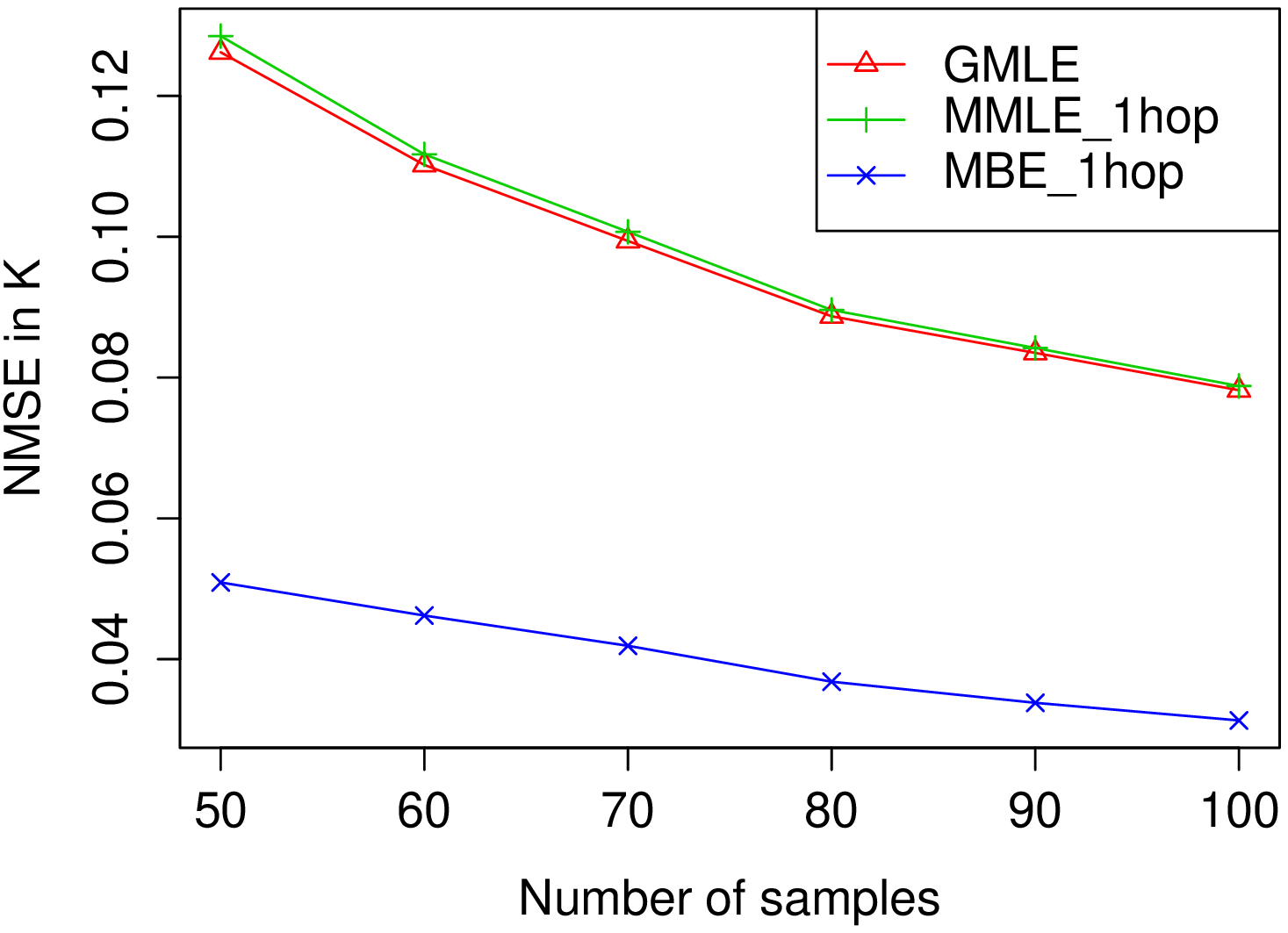}
                 \caption{}
                 \label{fig:lattice}
         \end{subfigure}

         \caption{NMSE in $K$ for different colored graphical models. (a) NMSE for the colored graph in Figure \ref{fig:1} (a) when $p=20$. (b) NMSE for the colored graph in Figure \ref{fig:1} (b) when $p=20$. (c) NMSE for the colored graph in Figure \ref{fig:1} (c) when $p=20$. (d) NMSE for the colored lattice graph in Figure \ref{fig:1} (d) when $p=100$.}
         \label{fig:2}

\end{figure}

\section{Conclusion}
In this paper, we have given a distributed method to compute the posterior mean of  the precision matrix of a coloured graphical Gaussian model, using the DY conjugate prior. To our knowledge, this is the first time that a distributed method has been used in a Bayesian framework. It is also the first time, we believe, that a  Bayesian estimate of $K$ for a high-dimensional coloured graphical Gaussian model, has been given.

We have first shown that, for our distributed method,  the local models should be marginal, rather than conditional,  because from these local marginal models, we can directly extract all the components of the precision matrix. Second, we have studied the asymptotic properties of our Bayesian estimate $\tilde{\theta}$ both when the number $p$ of variables is fixed and when it is allowed to grow to infinity. When $p$ is fixed, in our first main result, Theorem \ref{pfix}, we prove the asymptotic normality of $\sqrt{n}(\tilde{\theta}-\theta_0)$ where $\theta_0$ is the true value of the parameter, using classical methods of asymptotic theory. When both $p$ and $n$ tend to infinity, part of our methodology is to extend the method of \citet{Ghosal00} to our distributed method but we also developed novel methods to prove that the Frobenius norm of $\tilde{\theta}-\theta_0$ becomes arbitrarily small with probability tending to $1$, as $n$ and $p$ tend to infinity. Indeed, in Lemma \ref{Xin}, parallel to Lemma 1 of \citet{Xin15}, we show that if the third derivatives of the cumulant generating functions of $\Delta^i_n=\sqrt{n}(J^i)^{-1}(\bar{Y}^i-\mu_i), \;i=1,\ldots, p$, are uniformly bounded, then  the subnormality condition \eqref{sub}
given in Lemma \ref{Xin} holds. This, in turn allows us to prove the large deviation result given in Lemma \ref{max}, which, after being combined for all local models, yields our second main result, Theorem \ref{converge}. Though the convergence rate given in Theorem \ref{converge} seems to indicate that for this distributed estimate to achieve convergence to $\theta_0$, $n$ has to be very large, in fact, we show in Theorem \ref{local}, our third main result, that when we assume that the number of parameters in the local models is uniformly bounded, then the relative growth rate for $p$ and $n$ compares well with that given by \citet{Meng14} under similar conditions.
The examples in Section 6, show, through computations, that our method is efficient with sample sizes actually  smaller than suggested by our asymptotic results.

\section{Appendix}

In the following, we provide the proofs to all theorems.

\begin{proof}{\textbf{Proof of Theorem \ref{pfix} }}{}%
For any $i \in \{1,2,\ldots,p\}$, we have that
$\sqrt{n}(\tilde{\theta}^{i}-\theta^{i}_{0})=\sqrt{n}(\tilde{\theta}^{i}-T^{i})
+\sqrt{n}(T^{i}-\theta^{i}_{0})$
where $T^{i}=\theta^{i}_{0}+\frac{1}{n}[I^{i}(\theta^{i}_{0})]^{-1}\frac{\partial l^{i}(\theta^{i})}{\partial \theta^{i}}\big|_{\theta^{i}=\theta_{0}^{i}} $. It then follows from Theorem 8.3 in \citet{Lehmann98} that $\sqrt{n}(\tilde{\theta}^{i}-T^{i})\xrightarrow{p} 0$. Furthermore, we have
\begin{eqnarray*}\label{BM}
\sqrt{n}(T^{i}-\theta^{i}_{0})
=\frac{1}{\sqrt{n}}[I^{i}(\theta^{i}_{0})]^{-1}\frac{\partial l^{i}(\theta^{i})}{\partial \theta^{i}}\big|_{\theta^{i}=\theta_{0}^{i}}
=\sum\limits^{n}_{j=1}U_{ij}
\end{eqnarray*}
with $E[U_{j}]=0$ for $j=1,2,\cdots,n$. Next,
we compute the covariance matrix $Cov(U_{1})$ with $(i,k)$ entry
\begin{eqnarray}\label{variance}Cov(U_{i1},U_{k1})
&=&\frac{1}{n}[I^{i}(\theta^{i}_{0})]^{-1}E\big[\frac{\partial l^{i}(\theta^{i}|X^{i}_{1})}{\partial \theta^{i}}\big|_{\theta^{i}=\theta_{0}^{i}} (\frac{\partial l^{k}(\theta^{k}|X^{k}_{1})}{\partial \theta^{k}}\big|_{\theta^{k}=\theta_{0}^{k}})^{t}\big][I^{k}(\theta^{k}_{0})]^{-1}.
\end{eqnarray}
Based on the definition of the indicator matrix $\delta_{r}^{i}$,
\begin{eqnarray*}\label{k and theta}
\theta_{0}^{i}=(\frac{1}{|\tau^{i}_{1}|}tr(\delta^{i}_{1} K_{0}^{i}), \frac{1}{|\tau^{i}_{2}|}tr(\delta^{i}_{2} K_{0}^{i}), \cdots, \frac{1}{|\tau^{i}_{S_{i}}|}tr( \delta^{i}_{S_{i}} K_{0}^{i} ))^t,
\end{eqnarray*}
where $\tau^{i}_{r}$ is the numbers of elements belonging to the $r$-th colour class in $K_{0}^{i}$.
Since $X^{i}$ has a multivariate normal distribution $N(0, (K_{0}^{i})^{-1})$, we have
\begin{eqnarray*}
\frac{\partial l^{i}(\theta^{i}|X^{i}_{1})}{\partial \theta_{j}^{i}}\big|_{\theta^{i}=\theta_{0}^{i}}=\frac{1}{2}tr(\delta_{j}^{i} (K_{0}^{i})^{-1})-\frac{1}{2}tr(\delta_{j}^{i}X^{i}_{1}(X^{i}_{1})^{t}).
\end{eqnarray*}
Therefore, the $(q,m)$ entry of $E\big[\frac{\partial l^{i}(\theta^{i}|X^{i}_{1})}{\partial \theta^{i}}\big|_{\theta^{i}=\theta_{0}^{i}} (\frac{\partial l^{k}(\theta^{k}|X^{k}_{1})}{\partial \theta^{k}}\big|_{\theta^{k}=\theta_{0}^{k}})^{t}\big]$ in \eqref{variance} is
\begin{eqnarray*}&&E\Big[\frac{\partial l^{i}(\theta^{i}|X^{i}_{1})}{\partial \theta_{q}^{i}}\big|_{\theta^{i}=\theta_{0}^{i}}\frac{\partial l^{k}(\theta^{k}|X^{k}_{1})}{\partial \theta_{m}^{k}}\big|_{\theta^{k}=\theta_{0}^{k}}\Big]\\
&=&\frac{1}{4}tr(\delta_{q}^{i} \Sigma_{0}^{i}) \times tr( \delta_{m}^{k} \Sigma_{0}^{k})-\frac{1}{4}tr(\delta_{q}^{i} \Sigma_{0}^{i})\times tr( \delta_{m}^{k} E[ X^{k}_{1}(X^{k}_{1})^{t}])\\
&&\vspace{2cm}-\frac{1}{4}tr(\delta_{m}^{k} \Sigma_{0}^{k})\times tr(\delta_{q}^{i} E[X^{i}_{1}(X^{i}_{1})^{t}])+\frac{1}{4}E[tr(\delta_{q}^{i} X^{i}_{1}(X^{i}_{1})^{t})\times tr(\delta_{m}^{k} X^{k}_{1}(X^{k}_{1})^{t})]\\
\end{eqnarray*}
where $\Sigma^{i}_{0}=(K_{0}^{i})^{-1}$ and $\Sigma^{k}_{0}=(K_{0}^{k})^{-1}$. According to Isserlis' Theorem, we have that
\begin{eqnarray*}
E[X_{1a}X_{1b}X_{1c}X_{1d}]
=(\Sigma_{0})_{ab}(\Sigma_{0})_{cd}
+(\Sigma_{0})_{ac}(\Sigma_{0})_{bd}+(\Sigma_{0})_{ad}(\Sigma_{0})_{bc}.\\
\end{eqnarray*}
Therefore,
each entry of $nCov(U_{1})$ is well-defined. By Multivariate Central Limit Theorem, we have $\sqrt{n}(\bar{\theta}-\bar{\theta}_{0})\xrightarrow{\pounds} N(0,\bar{G})$ as $n\rightarrow \infty$, where
$\bar{G}=nCov(U_{1})$. As $\tilde{\theta}=g(\bar{\theta})$, based on Delta method, we have that
$\sqrt{n}(\tilde{\theta}-\theta_{0})\xrightarrow{\pounds} N(0, A)$
where $A=\frac{\partial g(\bar{\theta})}{\partial \bar{\theta}^t}\bar{G}(\frac{\partial g(\bar{\theta})}{\partial \bar{\theta}^t})^t$.
\end{proof}

\begin{proof}{\textbf{Proof of Theorem \ref{covariance} }}{}%
For any $i \in \{1,2,\ldots,p\}$, we use the well known result for MLE as follows
\begin{eqnarray}\label{MLE}
\sqrt{n}(\hat{\theta}^{i}-\theta_{0}^{i})=\frac{1}{\sqrt{n}}[I^{i}(\theta^{i}_{0})]^{-1}\sum\limits^{n}_{j=1}\frac{\partial l^{i}(\theta^{i}|X^{i}_{j})}{\partial \theta^{i}}\big|_{\theta^{i}=\theta_{0}^{i}}+R^{i}
\end{eqnarray}
where $R^{i}\xrightarrow{p} 0$ as $n\rightarrow \infty$. Comparing identity \eqref{MLE} with \eqref{BM} in Theorem \ref{pfix}, the result of Theorem \ref{covariance} follows.
\end{proof}

\begin{proof}{\textbf{Proof of Theorem \ref{converge} } }{}%
In this theorem, we study the consistency of $\tilde{\theta}$ in the context of Frobenius norm. In order to do this, first, we evaluate the norm $||\tilde{\theta}^{i}-\theta^{i}_{0}||^2$ in each local model. Since $||\sqrt{n}J^{i}(\tilde{\theta^{i}}-\theta_{0}^{i})||^2=n(\tilde{\theta^{i}}-\theta_{0}^{i})^t(J^{i})^tJ^{i}(\tilde{\theta^{i}}-\theta_{0}^{i})\geq n\lambda_{min}(F^{i})||\tilde{\theta}^{i}-\theta^{i}_{0}||^2$, we obtain
\begin{eqnarray}
||\tilde{\theta}^{i}-\theta^{i}_{0}||^2&\leq& \frac{1}{n\lambda_{min}(F^{i})}||\sqrt{n}J^{i}(\tilde{\theta^{i}}-\theta_{0}^{i})||^2\nonumber\\
&=& \frac{1}{n\lambda_{min}(F^{i})}||\Delta^{i}_{n}+\int u^{i} [\pi_{*}^{i}(u^{i})  - \phi(u^{i}; \Delta^{i}_{n}, I_{S_{i}})] du^{i}||^2 \hspace{8mm}\text{by Lemma \ref{allnorm}}\nonumber\\
&\leq &\frac{1}{n\lambda_{min}(F^{i})}\Big(||\Delta^{i}_{n}||^2+||\int u^{i} [\pi_{*}^{i}(u^{i}) - \phi(u^{i};\Delta^{i}_{n}, I_{S_{i}})] du^{i}||^2\Big)\label{T1}\\\nonumber
\end{eqnarray}
where $\phi(\cdot;v,\Sigma)$ stands for the multivariate normal density of $N(v, \Sigma)$ and $\pi_{*}^{i}(u^{i})$ stands for the posterior distribution of $u^i$.
Next, for every element of the vector $\int u^{i} [\pi_{*}^{i}(u^{i}) - \phi(u^{i};\Delta^{i}_{n}, I_{S_{i}})] du^{i}$ in \eqref{T1}, we will find out its upper bound. Denote $u^{i}=(u_{1}^{i},u_{2}^{i},\cdots,u_{S_{i}}^{i})^t$. Then for the $j$-th element of $\int u^{i} [\pi_{*}^{i}(u^{i}) - \phi(u^{i};\Delta^{i}_{n}, I_{S_{i}})] du^{i}$, we have that
\begin{eqnarray}\label{absolute}
\int u_{j}^{i} [\pi_{*}^{i}(u^{i}) - \phi(u^{i};\Delta^{i}_{n}, I_{S_{i}})] du^{i} &\leq & \int ||u^{i}|| [\pi_{*}^{i}(u^{i}) - \phi(u^{i};\Delta^{i}_{n}, I_{S_{i}})] du^{i}.\\\nonumber
\end{eqnarray}
Let $\tilde{Z}^{i}_{n}(u^{i})=\exp[(u^{i})^t\Delta^{i}_{n}-\frac{1}{2}||u^{i}||^2]$ and $M(p)=p^2\log p$. According to the argument of Theorem 2.3 in \citet{Ghosal00}, the integral $\int ||u^{i}|| [\pi_{*}^{i}(u^{i}) - \phi(u^{i};\Delta^{i}_{n}, I_{S_{i}})] du^{i}$ in \eqref{absolute} can be bounded by a sum of three integrals as follows.
\begin{eqnarray}
&&\int||u^{i}||\times|\pi_{*}^{i}(u^{i})-\phi(u^{i};\Delta^{i}_{n}, I_{S_{i}})|du^{i} \nonumber\\
&\leq &\frac{\int_{||u^{i}||^2\leq cM(p)}||u^{i}||\cdot |\pi^{i}(\theta^{i}_{0}+n^{-\frac{1}{2}}(J^i)^{-1}u^{i})Z^{i}_{n}(u^{i})-\pi^{i}(\theta^{i}_{0})\tilde{Z}^{i}_{n}(u^{i})| du^{i}}{\int \pi^{i}(\theta^{i}_{0})\tilde{Z}^{i}_{n}(u^{i})du^{i}}\label{TT1}\\
&&+[\int \pi^{i}(\theta^{i}_{0})\tilde{Z}^{i}_{n}(u^{i})du^{i}]^{-1} \int_{||u^{i}||^2> cM(p)}||u^{i}|| \cdot Z^{i}_{n}(u^{i}) \pi^{i}(\theta_{0}+n^{-\frac{1}{2}}(J^{i})^{-1}u^{i}) du^{i}\label{TT2}\\
&&+\int_{||u^i||^2> cM(p)} ||u^{i}||\phi(u^{i}; \Delta^{i}_{n}, I_{S_{i}}) du^{i},\label{TT3} \\ \nonumber
\end{eqnarray}
where $c$ is defined in Lemma \ref{secondtermnorm}.
By Lemmas \ref{firsttermnorm}, \ref{secondtermnorm} and \ref{thirdtermnorm}, $\int||u^{i}||\times|\pi_{*}^{i}(u^{i})-\phi(u^{i};\Delta^{i}_{n}, I_{S_{i}})|du^{i}$ can be bounded by $$A(p,n,c)=c_{5}(c)\frac{p^{13}\log p}{\sqrt{n}}+\exp[- c_{9}(c)p^2\log p]+\frac{2}{\sqrt{2\pi}}p^{-4a^2+4}+\sqrt{3a^2}\frac{2}{\sqrt{2\pi}} p^{-4a^2+3}$$ with probability greater than $1-10.4\exp\{-\frac{1}{6}p^2\}$. Consequently,
\begin{eqnarray*}
\int u_{j}^{i} [\pi_{*}^{i}(u^{i}) - \phi(u^{i};\Delta^{i}_{n}, I_{S_{i}})] du^{i} &\leq & A(p,n,c)\\
\end{eqnarray*}
with probability greater than $1-10.4\exp\{-\frac{1}{6}p^2\}$.
Since the dimension of $\int u^{i} [\pi_{*}^{i}(u^{i}) - \phi(u^{i};\Delta^{i}_{n}, I_{S_{i}})] du^{i}$ is $S_{i}$, from the inequality \eqref{T1} and Lemma \ref{max}, we get
\begin{eqnarray*}
||\tilde{\theta}^{i}-\theta^{i}_{0}||^2 &\leq & \frac{1}{\lambda_{min}(F^{i})}\Big(\frac{3a^2p^2}{n}+\frac{S_{i}}{n}A(p,n,c)\Big)\\
\end{eqnarray*}
with probability greater than $1-10.4\exp\{-\frac{1}{6}p^2\}$.
Finally, we will estimate the Frobenius norm $||\tilde{\theta}-\theta_{0}||$ for the distributed estimator $\tilde{\theta}$ in terms of $||\tilde{\theta}^{i}-\theta^{i}_{0}||$ from the local model.
By Proposition \ref{Fbound}, for any $i \in \{1,2,\ldots,p\}$, $\lambda_{min}(F^{i})\geq \frac{1}{\kappa^2_{2}}$. Therefore, we have
\begin{eqnarray*}
||\tilde{\theta}-\theta_{0}||& \leq &||\bar{\theta}-\bar{\theta}_{0}||
\leq (\sum\limits^{p}_{i=1}||\tilde{\theta}^{i}-\theta^{i}_{0}||^2)^{\frac{1}{2}} \hspace{16mm}\text{by triangle inequality}\nonumber\\
&\leq & \Big\{\sum\limits^{p}_{i=1}\Big[\frac{1}{\lambda_{min}(F^{i})}\big(\frac{3a^2p^2}{n}+\frac{S_{i}}{n}A(p,n,c)\big)\Big]\Big\}^{\frac{1}{2}}\nonumber\\
&\leq & \Big\{\kappa^2_{2}\big[\frac{3a^2p^3}{n}+\frac{p^2(p+1)}{2n}A(p,n,c)\big]\Big\}^{\frac{1}{2}}\nonumber\\ \nonumber
\end{eqnarray*}
with probability greater than $1-10.4p\exp\{-\frac{1}{6}p^2\}$ by the Bonferroni inequality.
Furthermore, Condition (4) implies $A(p,n,c)\rightarrow 0$. Therefore, there exists a constant $c^{*}$ such that
\begin{eqnarray*}
||\tilde{\theta}-\theta_{0}||
&\leq & \Big\{\kappa^2_{2}\big[\frac{3a^2p^3}{n}+\frac{p^2(p+1)}{2n}o(1)\big]\Big\}^{\frac{1}{2}}\\
&\leq & c^{*}\frac{p^{\frac{3}{2}}}{\sqrt{n}}
\end{eqnarray*}
with probability greater than $1-10.4\exp\{-\frac{1}{6}p^2+\log p\}\rightarrow 1$.
\end{proof}

\begin{proof}{\textbf{Proof of Theorem \ref{local} } }{}%
The proof follows the same line as that of Theorem \ref{converge}. Our aim is to find the upper bound for the three terms \eqref{TT1}, \eqref{TT2} and \eqref{TT3}.

1. A bound for \eqref{TT1}: Under the Condition (5), the Lipschitz continuity in Proposition \ref{lipschiz} becomes $|\log \pi^{i}(\theta^{i})-\log \pi^{i}(\theta^{i}_{0})| \leq M_{1} S^{*}||\theta^{i}-\theta^{i}_{0}||$ when $||\theta^{i}-\theta^{i}_{0}|| \leq\sqrt{||(F^{i})^{-1}||cM(p)/n}$. We choose $M(p)=\log^2p(\log \log p)$, then $\varphi^{i}_{n}(c)=O(\frac{\log p (\log \log p)^{\frac{1}{2}}}{\sqrt{n}})$ and $f^{i}(||\Delta^{i}_{n}||,c)=O(\frac{\log^3 p(\log \log p)^{\frac{1}{2}}}{\sqrt{n}})$ in Lemma \ref{firsttermnorm}. Therefore, following the same proof of Lemma \ref{firsttermnorm}, we have that there exists a constant $c'_{5}(c)$ such that
$R_{1}(||\Delta^{i}_{n}||,c)\leq c'_{5}(c)\frac{\log^4 p\log \log p}{\sqrt{n}}$ with probability greater than $1-10.4\exp\{-\frac{1}{6}\log^2 p\}$.

2. A bound for \eqref{TT2}:
According to Lemma \ref{znubound} and following the same proof of Lemma 2.2 of \citet{Ghosal00}, on $||u^{i}||^2>cM(p)$, we have $Z^{i}_{n}(u^{i})\leq \exp[-\frac{1}{4}c\log^2 p(\log \log p)]$ with probability greater than $1-10.4\exp\{-\frac{1}{6}\log^2 p\}$. Following the same proof as that of Lemma \ref{secondtermnorm}, there exists a constant $c$ and a constant $c'_{9}(c)$ such that
$$R_{2}(||\Delta^{i}_{n}||,c)\leq \exp[- c'_{9}(c)M(p)]$$
with probability greater than $1-10.4\exp\{-\frac{1}{6}\log^2 p\}$.

3. A bound for \eqref{TT3}: According to Lemma \ref{thirdtermnorm}, for $M(p)=\log^2p(\log \log p)$, we have
$$\int_{||u^{i}||^2>cM(p)}||u^{i}||\phi(u^{i};\Delta^{i}_{n},I_{S_{i}})du^{i}\leq \frac{2S^2_{i}}{\sqrt{2\pi}}e^{-\frac{4a^2M(p)}{S_{i}}}+\frac{2\sqrt{3a^2}S^{i}\log p}{\sqrt{2\pi}}e^{-\frac{4a^2M(p)}{S_{i}}}$$
with probability greater than $1-10.4\exp\{-\frac{1}{6}\log^2 p\}$.

Combining the above results, we have
$$A(p,n,c)=c'_{5}(c)\frac{\log^4 p\log \log p}{\sqrt{n}}+e^{- c'_{9}(c)M(p)}+\frac{2S^2_{i}}{\sqrt{2\pi}}e^{-\frac{4a^2M(p)}{S_{i}}}+\frac{2\sqrt{3a^2}S^{i}\log p}{\sqrt{2\pi}}e^{-\frac{4a^2M(p)}{S_{i}}}$$
with probability greater than $1-10.4\exp\{-\frac{1}{6}\log^2 p\}$. It follows
\begin{eqnarray*}
||\tilde{\theta}-\theta_{0}||
&\leq & \Big\{\kappa^2_{2}\big[\frac{p3a^2\log^2 p}{n}+\frac{S_{i}}{n}A(p,n,c\big]\Big\}^{\frac{1}{2}}\\
&\leq & c_{1}^{*}\frac{\sqrt{p}\log p}{\sqrt{n}}
\end{eqnarray*}
with probability greater than $1-10.4p\exp\{-\frac{1}{6}\log^2p\}$ by the Bonferroni inequality. This completes the proof.
\end{proof}

Here we provide the lemmas and their proofs. Additional technical lemmas and propositions are provided in a supplementary file. We let
\begin{eqnarray}\label{U}
\bar{U}_{j}^{i}=(J^{i})^{-1}(Y^{i}_{j}-\mu^{i})
\end{eqnarray}
for $i=1,2,\ldots,p$, and $j=1,2,\ldots,n$. We now want to show the large deviation result for $\Delta_{n}^{i}$. To do so, we need to show that the cumulant boundedness condition is satisfied by $\bar{U}_{j}^{i}$ (Lemma \ref{boundcondition}). This will allow us to show that $\Delta_{n}^{i}$ satisfy the exponential moment condition (Lemma \ref{Xin}). In Lemma \ref{max}, we obtain the large deviation result for $\Delta_{n}^{i}$.

\begin{lemma}{}{}%
\label{boundcondition}For any $i\in \{1,2,\ldots,p\}$, there exist constants $\eta$ and $C_{2}$ such that under Condition (2) and (3), for $||\gamma^{i}||\leq \eta$ and for
all $1\leq k,l,m \leq S_{i}$, the absolute value of all the third derivatives of the cumulant generating function $G_{\bar{U}^{i}_{j}}^{i}(\gamma^{i})$ of $\bar{U}^{i}_{j}$
satisfy $$\Big|\frac{\partial^3 G_{\bar{U}^{i}_{j}}^{i}(\gamma^{i})}{\partial \gamma^{i}_{k}\partial \gamma^{i}_{l}\partial \gamma^{i}_{m}}\Big|\leq C_{2}, \hspace{4mm} j=1,2\ldots, n.$$
\end{lemma}
\begin{proof}{\textbf{Proof} }{}%
Let $Y_{j}^{i}$ be defined in \eqref{sufficient} of Section 4 and $G_{Y_{j}^{i}}^{i}(\gamma^{i})=\log E(e^{(\gamma^{i})^tY_{j}^{i}})$ be the cumulant generating function of $Y_{j}^{i}$. Let $\gamma^{i}$ be a $S_{i}$-dimensional vector, by Theorem 3.2.3 in \citet{Mui82}, the moment generating function of $Y_{j}^{i}$ is
$$M^{i}(\gamma^{i})=E\{\exp [(\gamma^{i})^tY^{i}_{j}]\}=|I_{p_{i}}+T^{i}(\gamma^{i})\Sigma_{0}^{i}|^{-\frac{1}{2}}$$
where $T^{i}(\gamma^{i})$ is a ${p_{i}\times p_{i}}$ matrix with $T_{\alpha\beta}^{i}=\gamma_{k}^{i}$ if $K^{i}_{\alpha\beta}=\theta_{k}^{i}$. Therefore,
the cumulant generating function $G_{Y_{j}^{i}}^{i}(\gamma^{i})$ of $Y^{i}_{j}$ is given by
$$G_{Y_{j}^{i}}^{i}(\gamma^{i})=\log M^{i}(\gamma^{i})=-\frac{1}{2}\log |I_{p_{i}}+T^{i}(\gamma^{i})\Sigma_{0}^{i}|.$$
It is easy to obtain the first, second and third derivative of the cumulant generating function $G_{Y_{j}^{i}}^{i}(\gamma^{i})$, which can be expressed as
\begin{eqnarray*}
\frac{\partial G_{Y_{j}^{i}}^{i}(\gamma^{i})}{\partial \gamma_{k}^{i}}
&=& -\frac{1}{2}tr\Big([I_{p_{i}}+T^{i}(\gamma^{i})\Sigma_{0}^{i}]^{-1} (\delta^{i}_{k}\Sigma_{0}^{i}) \Big),\\
\frac{\partial^2 G_{Y_{j}^{i}}^{i}(\gamma^{i})}{\partial \gamma_{k}^{i}\partial  \gamma_{l}^{i}}
&=& \frac{1}{2}tr\Big(\delta^{i}_{k}\Sigma_{0}^{i}[I_{p_{i}}+T^{i}(\gamma^{i})\Sigma_{0}^{i}]^{-1}(\delta^{i}_{l}\Sigma_{0}^{i} )[I_{p_{i}}+T^{i}(\gamma^{i})\Sigma_{0}^{i}]^{-1}\Big) \text{ and }\\
\frac{\partial^3 G_{Y_{j}^{i}}^{i}(\gamma^{i})}{\partial \gamma_{k}^{i}\partial \gamma_{l}^{i}\partial \gamma_{m}^{i}}
&=& -\frac{1}{2}tr\Big (\delta^{i}_{k}\Sigma_{0}^{i}[I_{p_{i}}+T^{i}(\gamma^{i})\Sigma_{0}^{i}]^{-1}(\delta^{i}_{m}\Sigma_{0}^{i} )(I_{p_{i}}+T^{i}(\gamma^{i})\Sigma_{0}^{i})^{-1}(\delta^{i}_{l}\Sigma_{0}^{i} )\\
\hspace{8mm}&&\times[I_{p_{i}}+T^{i}(\gamma^{i})\Sigma_{0}^{i}]^{-1}+\delta^{i}_{k}\Sigma_{0}^{i}[I_{p_{i}}+T^{i}(\gamma^{i})\Sigma_{0}^{i}]^{-1}(\delta^{i}_{l}\Sigma_{0}^{i} )[I_{p_{i}}+T^{i}(\gamma^{i})\Sigma_{0}^{i}]^{-1}\\
\hspace{8mm}&&\times(\delta^{i}_{m}\Sigma_{0}^{i} )[I_{p_{i}}+T^{i}(\gamma^{i})\Sigma_{0}^{i}]^{-1}\Big),
\end{eqnarray*}
respectively. First, Condition (2) implies $\lambda_{max}(\Sigma_{0}^{i})\leq \frac{1}{\kappa_{1}}$. By Proposition \ref{entrybound}, the absolute value of each element of $\Sigma_{0}^{i}$ is bounded by $\frac{1}{\kappa_{1}}$. Next, by $\sum\limits_{j=1}^{p}|\lambda_{j}(A)|\leq ||A||_{F}$ and $||AB||\leq ||AB||_{F}\leq ||A||_{F}||B||$ for any two $p\times p$ symmetric matrix, we have that $|\lambda_{j}(T^{i}(\gamma^{i})\Sigma_{0}^{i})| \leq ||T^{i}(\gamma^{i})||_{F}||\Sigma_{0}^{i}||\leq \eta \frac{1}{\kappa_{1}}$. It implies $1-\eta \frac{1}{\kappa_{1}}\leq \lambda_{j}(I_{p_{i}}+T^{i}(\gamma^{i})\Sigma_{0}^{i})\leq 1+\eta \frac{1}{\kappa_{1}}$. Moreover, according to Lemma \ref{positive}, $I_{p_{i}}+T^{i}(\gamma^{i})\Sigma_{0}^{i}$ is a positive definite. Therefore, by Proposition \ref{entrybound} again, the absolute value of each element of $[I_{p_{i}}+T^{i}(\gamma^{i})\Sigma_{0}^{i}]^{-1}$ is bounded. Finally, combining the above results and Condition (3),
for any $i \in \{1,2,\ldots,p\}$, there exists a constant $C_{1}$ such that $|\frac{\partial^3 G_{Y_{j}^{i}}^{i}(\gamma^{i})}{\partial \gamma_{k}^{i}\partial \gamma_{l}^{i}\partial \gamma_{m}^{i}}|\leq C_{1}$ for any $k, m, l$. Since
the cumulant generating function of $\bar{U}^{i}_{j}$ is
\begin{eqnarray*}
G_{\bar{U}_{j}^{i}}^{i}(\gamma^{i})&=&\log E[e^{(\gamma^{i})^t(J^i)^{-1}(Y^{i}_{j}-\mu^{i})}]
= G_{Y_{j}^{i}}^{i}((J^i)^{-1}\gamma^{i})-(\gamma^{i})^t(J^i)^{-1}\mu^{i}.
\end{eqnarray*}
It follows that there exists a constant $C_{2}$ such that $|\frac{\partial^3 G_{\bar{U}_{j}^{i}}^{i}(\gamma^{i})}{\partial \gamma_{k}^{i}\partial \gamma_{l}^{i}\partial \gamma_{m}^{i}}|\leq C_{2}$ for $||\gamma^{i}||\leq \eta$.
\end{proof}

\begin{lemma}{}{}%
\label{Xin} Let $\Delta_{n}^{i}$ and $\bar{U}^{i}_{j}$ be as defined in \eqref{delta} of Section 4 and \eqref{U}, respectively. Let $C_{2}$ be as in Lemma \ref{boundcondition}. Then, under Condition (2)-(4), for any arbitrary constant $a$ such that $a^2>1$, we have that if $\frac{C_{2}p^3}{3\sqrt{n}}\leq a-1$, then as $n\rightarrow \infty$,
\begin{eqnarray}\label{sub}
G_{\Delta_{n}^{i}}^{i}(\gamma^i)=\log \big(E\{\exp[(\gamma^i)^t\Delta_{n}^{i}]\}\big)\leq a^2||\gamma^{i}||^2/2 \hspace{4mm}\text{for} \hspace{4mm}||\gamma^{i}|| < p.
\end{eqnarray}
\end{lemma}
\begin{proof}{\textbf{Proof} }{}%
By a Taylor expansion of $G_{\bar{U}_{j}^{i}}^{i}(\gamma^i)$ around 0, there exists a vector $\gamma^{i,*}$ on the line segment between 0 and $\gamma^i$ such that
\begin{eqnarray*}
G_{\bar{U}_{j}^{i}}^{i}(\gamma^i)&=&G_{\bar{U}_{j}^{i}}^{i}(0)+\sum\limits_{k=1}^{S_{i}}\big(\frac{\partial G_{\bar{U}_{j}^{i}}^{i}(\gamma^{i})}{\partial \gamma_{k}^{i}}|_{\gamma^{i}=0}\big)\gamma_{k}^{i}+\frac{1}{2}\sum\limits_{k=1}^{S_{i}}\sum\limits_{l=1}^{S_{i}}\big(\frac{\partial^2 G_{\bar{U}_{j}^{i}}^{i}(\gamma^{i})}{\partial \gamma_{k}^{i}\partial \gamma_{l}^{i}}|_{\gamma^{i}=0}\big)\gamma_{k}^{i}\gamma_{l}^{i}\\
&&+\frac{1}{6}\sum\limits_{k=1}^{S_{i}}\sum\limits_{l=1}^{S_{i}}\sum\limits_{m=1}^{S_{i}}\big(\frac{\partial^3 G_{\bar{U}_{j}^{i}}^{i}(\gamma^{i})}{\partial \gamma_{k}^{i}\partial \gamma_{l}^{i}\partial \gamma_{m}^{i}}|_{\gamma^{i}=\gamma^{i,*}}\big)\gamma_{k}^{i}\gamma_{l}^{i}\gamma_{m}^{i}.
\end{eqnarray*}
Since $\bar{U}^{i}_{j}$ has zero mean and identity covariance matrices, then $\frac{\partial G_{\bar{U}_{j}^{i}}^{i}(\gamma^{i})}{\partial \gamma_{k}^{i}}|_{\gamma^{i}=0}=0$, $\frac{\partial^2 G_{\bar{U}_{j}^{i}}^{i}(\gamma^{i})}{\partial \gamma_{k}^{i}\partial \gamma_{l}^{i}}|_{\gamma^{i}=0}=1$ for $k=l$ and $\frac{\partial^2 G_{\bar{U}_{j}^{i}}^{i}(\gamma^{i})}{\partial \gamma_{k}^{i}\partial \gamma_{l}^{i}}|_{\gamma^{i}=0}=0$ for $k\neq l$. Furthermore, since $G_{\bar{U}_{j}^{i}}^{i}(0)=0$, we have
\begin{eqnarray*}
G_{\bar{U}_{j}^{i}}^{i}(\gamma^i)&=&\frac{1}{2}\big(\gamma^{i})^t\gamma^{i}+\frac{1}{6}\sum\limits_{k=1}^{S_{i}}\sum\limits_{l=1}^{S_{i}}\sum\limits_{m=1}^{S_{i}}(\frac{\partial^3 G_{\bar{U}_{j}^{i}}^{i}(\gamma^{i})}{\partial \gamma_{k}^{i}\partial \gamma_{l}^{i}\partial \gamma_{m}^{i}}|_{\gamma^{i}=\gamma^{i,*}}\big)\gamma_{k}^{i}\gamma_{l}^{i}\gamma_{m}^{i}.
\end{eqnarray*}
By the definition \eqref{delta} of Section 4, we have $\Delta^{i}_{n}=\frac{1}{\sqrt{n}}\sum\limits_{j=1}^{n}\bar{U}^{i}_{j}.$ Since the moment generating function of $\bar{U}_{j}^{i}$ is $\exp G_{\bar{U}_{j}^{i}}^{i}(\gamma^i)$, then the moment generating function of $\Delta^{i}_{n}$ is
\begin{eqnarray*}
E[e^{(\gamma^{i})^t\Delta^{i}_{n}}]&=&E[e^{(\gamma^{i})^t\frac{1}{\sqrt{n}}\sum\limits_{j=1}^{n}\bar{U}_{j}^{i}}]=\prod_{j=1}^{n}E[e^{(\frac{\gamma^{i}}{\sqrt{n}})^t\bar{U}_{j}^{i}}]\\
&=&\exp\Big\{\frac{1}{2}(\gamma^{i})^t\gamma^{i}+\frac{1}{6}\frac{1}{\sqrt{n}}\sum\limits_{k=1}^{S_{i}}\sum\limits_{l=1}^{S_{i}}\sum\limits_{m=1}^{S_{i}}\big(\frac{\partial^3 G_{\bar{U}_{j}^{i}}^{i}(\frac{\gamma^{i}}{\sqrt{n}})}{\partial \gamma_{k}^{i}\partial \gamma_{l}^{i}\partial \gamma_{m}^{i}}|_{\gamma^{i}=\gamma^{i,*}}\big)\gamma_{k}^{i}\gamma_{l}^{i}\gamma_{m}^{i}\Big\}.\\
\end{eqnarray*}
Since $||\gamma^{i,*}||<||\gamma^{i}||$, we have $||\frac{\gamma^{i,*}}{\sqrt{n}}||<||\frac{\gamma^{i}}{\sqrt{n}}||<\frac{p}{\sqrt{n}}$. Moreover,  Condition (4) implies $\frac{p}{\sqrt{n}}\rightarrow 0$, and thus $||\frac{\gamma^{i}}{\sqrt{n}}||\leq \eta$ for $n$ large enough. Therefore, by Lemma \ref{boundcondition}, there exists a constant $C_{2}$ such that $\Big|\frac{\partial^3 G_{\bar{U}_{j}^{i}}^{i}(\frac{\gamma^{i}}{\sqrt{n}})}{\partial \gamma^{i}_{k}\partial \gamma^{i}_{l}\partial \gamma^{i}_{m}}\Big|\leq C_{2}$. It follows
\begin{eqnarray*}
E[e^{(\gamma^{i})^t\Delta^{i}_{n}}]&\leq&\exp\Big\{\frac{1}{2}(\gamma^{i})^t\gamma^{i}+\frac{1}{6}\frac{C_{2}}{\sqrt{n}}\sum\limits_{k=1}^{S_{i}}
\sum\limits_{l=1}^{S_{i}}\sum\limits_{m=1}^{S_{i}}\gamma_{k}^{i}\gamma_{l}^{i}\gamma_{m}^{i}\Big\}\\
&=&\exp\Big\{\frac{1}{2}(\gamma^{i})^t\gamma^{i}\big[1+\frac{1}{3}\frac{C_{2}}{\sqrt{n}}\sum\limits_{m=1}^{S_{i}}\gamma_{m}^{i}\big]\Big\}.\\
\end{eqnarray*}
Therefore, for any arbitrary constant $a$ such that $a^2>1$, if $\frac{1}{3}\frac{C_{2}}{\sqrt{n}}\sum\limits_{m=1}^{S_{i}}\gamma_{m}^{i}\leq a^2-1$, then we have $$\log E[e^{(\gamma^{i})^t\eta^{i}}]\leq a^2||\gamma^{i}||^2/2.$$ Actually, the inequality $\frac{1}{3}\frac{C_{2}}{\sqrt{n}}\sum\limits_{m=1}^{S_{i}}\gamma_{m}^{i}\leq a^2-1$ holds under Condition (4). Since $||\gamma^{i}||< p$, we have $|\gamma_{m}^{i}| \leq ||\gamma^{i}|| < p$ for any $1\leq m\leq S_{i}$. Therefore, according to Condition (4), we have
\begin{eqnarray*}
\frac{1}{3}\frac{C_{2}}{\sqrt{n}}\sum\limits_{m=1}^{S_{i}}\gamma_{m}^{i}=O\big(\frac{p^3}{\sqrt{n}}\big)=o(1).
\end{eqnarray*}
It implies $\frac{1}{3}\frac{C_{2}}{\sqrt{n}}\sum\limits_{m=1}^{S_{i}}\gamma_{m}^{i}\leq a^2-1$ for any constant $a$ with $a^2>1$.
\end{proof}

\begin{lemma}{}{}%
\label{max}Under Condition (2)-(4), for any $i \in \{1,2,\ldots,p\}$ and $n$ sufficiently large, there exists a constant $a$, $a^2>1$, such that $$P\{||\Delta^{i}_{n}||^2 > 3a^2p^2\} \leq 10.4\exp\{-\frac{1}{6}p^2\}$$
where $\Delta^{i}_{n}$ is defined as in \eqref{delta} of Section 4.
\end{lemma}
\begin{proof}{\textbf{Proof} }{}%
According to Lemma \ref{Xin}, we have
$$\log \big(E\{\exp [(\gamma^{i})^t \Delta^{i}_{n}]\}\big)\leq a^2 ||\gamma^{i}||^2/2 \hspace{10mm} for \hspace{10mm} ||\gamma^{i}||\leq p$$
where $a$ is a constant with $a^2>1$. Let $g=ap$ and $t^i_{1}=a\gamma^{i}$, then the subsequent inequality holds
$$\log (E\{\exp [(t^{i}_{1})^t\frac{\Delta^{i}_{n}}{a}]\})\leq ||t^{i}_{1}||^2/2 \hspace{10mm} for \hspace{10mm} ||t_{1}^{i}||\leq g.$$
Next we apply the large deviation result from Corollary 3.2 in \citet{Spokoiny13}. Following the notations in \citet{Spokoiny13}, we introduce $w^{i}_{c}$ satisfying the equation
$\frac{w^{i}_{c}(1+w^{i}_{c})}{[1+(w^{i}_{c})^2]^{\frac{1}{2}}}=gS_{i}^{-1/2}.$ Based on $w^{i}_{c}$, we define $x^{i}_{c}=0.5S_{i}[(w^{i}_{c})^2-\log (1+(w^{i}_{c})^2)].$ Since $g^2=a^2p^2>\frac{p^2+p}{2}\geq S_{i}$, by the arguments in \citet{Spokoiny13}, we have $x^{i}_{c}>\frac{1}{4}g^2=\frac{1}{4}a^2p^2$. Let $x=\frac{1}{6}p^2$, then $\frac{S_{i}}{6.6}\leq \frac{p^2+p}{2\times 6.6}<x < x^{i}_{c}$. By Corollary 3.2 in \citet{Spokoiny13}, the following inequality holds
$$P(||\frac{\Delta^{i}_{n}}{a}||^2\geq S_{i}+6.6\times\frac{1}{6}p^2)\leq 2 e^{-\frac{1}{6}p^2}+8.4e^{-x^{i}_{c}},$$
which implies $P(||\frac{\Delta^{i}_{n}}{a}||^2\geq 3p^2)\leq 10.4 e^{-\frac{1}{6}p^2}$. Hence, $P(||\Delta^{i}_{n}||^2\geq 3a^2p^2)\leq 10.4 e^{-\frac{1}{6}p^2}$, which means $||\Delta^{i}_{n}||^2=O_{p}(p^2)$.\end{proof}

The next four lemmas are used to complete the proof of Theorem \ref{converge}.

\begin{lemma}{}{}%
\label{firsttermnorm}Under Conditions (2)-(4), for any given $i \in \{1,2,\ldots,p\}$ and for any given constant $c$, there exists a constant $c_{5}(c)$ such that
\begin{eqnarray}
\frac{\int_{||u^{i}||^2\leq cM(p)}||u^{i}||\cdot|\pi^{i}(\theta^{i}_{0}+n^{-\frac{1}{2}}(J^{i})^{-1}u^{i})Z^{i}_{n}(u^{i})-\pi^{i}(\theta^{i}_{0})\tilde{Z}^{i}_{n}(u^{i})| du^{i}}{\int \pi^{i}(\theta^{i}_{0})\tilde{Z}^{i}_{n}(u^{i})du^{i}} \leq c_{5}\frac{p^{13}\log p}{\sqrt{n}}\label{long}\\ \nonumber
\end{eqnarray}
with probability greater than $1-10.4\exp\{-\frac{1}{6}p^2\}$.
\end{lemma}
\begin{proof}{\textbf{Proof} }{}%
Let $Q^{i}$ denote the set $\{u^{i}; ||u^{i}||^2 \leq cM(p)\}$. We get that
\begin{eqnarray*}
&&[\int \pi^{i}(\theta^{i}_{0})\tilde{Z}^{i}_{n}(u^{i})du^{i}]^{-1}\int_{Q^{i}}||u^{i}||\cdot|\pi^{i}(\theta^{i}_{0}+n^{-1/2}(J^i)^{-1}u^{i})Z^{i}_{n}(u^{i})-\pi^{i}(\theta^{i}_{0})\tilde{Z}^{i}_{n}(u^{i})| du^{i} \\
&=& [\int \pi^{i}(\theta^{i}_{0})\tilde{Z}^{i}_{n}(u^{i})du^{i}]^{-1}\int_{Q^{i}}||u^{i}||\cdot|\frac{\pi^{i}(\theta^{i}_{0}+n^{-1/2}(J^i)^{-1}u^{i})}{\pi^{i}(\theta^{i}_{0})}Z^{i}_{n}(u^{i})-\tilde{Z}^{i}_{n}(u^{i})| \pi^{i}(\theta^{i}_{0}) du^{i}\\
&\leq & \frac{\sup\limits_{u^{i} \in Q^{i}}\Big\{||u^{i}||\cdot|\frac{\pi^{i}(\theta^{i}_{0}+n^{-1/2}(J^i)^{-1}u^{i})}{\pi^{i}(\theta^{i}_{0})}-1|\Big\}\int_{Q^{i}}Z^{i}_{n}(u^{i})du^{i}}{\int \tilde{Z}^{i}_{n}(u^{i})du^{i}}+\frac{\int_{Q^{i}}||u^{i}||\cdot|Z^{i}_{n}(u^{i})-\tilde{Z}^{i}_{n}(u^{i})| du^{i}}{\int \tilde{Z}^{i}_{n}(u^{i})du^{i}}. \\
\end{eqnarray*}
Since
$cM(p) \geq ||u^{i}||^2 = ||\sqrt{n}J^i(\theta^{i}-\theta^{i}_{0})||^2
\geq n\lambda_{min}(F^{i})||\theta^{i}-\theta^{i}_{0}||^2$,
then $||\theta^{i}-\theta^{i}_{0}|| \leq \sqrt{\frac{cM(p)||(F^i)^{-1}||}{n}}.$
By Proposition \ref{Fbound}, we have $\kappa_{1}^2\leq ||(F^i)^{-1}||\leq \kappa_{2}^2$. Based on Condition (4), $\frac{p^2\log p}{n}\rightarrow 0$. Therefore, $||\theta^{i}-\theta^{i}_{0}||\rightarrow 0$.
Using the fact $|e^{x}-1|\leq 2|x|$ for sufficiently small $|x|$ and Proposition \ref{lipschiz}, we obtain
\begin{eqnarray*}
\sup_{u^{i} \in Q^{i}}\Big\{||u^{i}||\cdot|\frac{\pi^{i}(\theta^{i}_{0}+n^{-1/2}(J^i)^{-1}u^{i})}{\pi^{i}(\theta^{i}_{0})}-1|\Big\}
\leq 2 \sqrt{cM(p)}M_{1}p ||\theta^{i}-\theta^{i}_{0}||
\leq \frac{2cM_{1}\kappa_{2}M(p)p}{\sqrt{n}}
\end{eqnarray*}
where $M_{1}$ is a constant. We also have that
\begin{eqnarray*}
\frac{\int_{Q^{i}}Z^{i}_{n}(u^{i})du^{i}}{\int \tilde{Z}^{i}_{n}(u^{i})du^{i}}
&=& \frac{\int_{Q^{i}} \tilde{Z}^{i}_{n}(u^{i})du^{i}+\int_{Q^{i}}[Z^{i}_{n}(u^{i})-\tilde{Z}^{i}_{n}(u^{i})]du^{i}}{\int \tilde{Z}^{i}_{n}(u^{i})du^{i}}\\
&\leq& 1+\Big(\int \tilde{Z}^{i}_{n}(u^{i})du^{i}\Big)^{-1}\int_{Q^{i}}|Z^{i}_{n}(u^{i})-\tilde{Z}^{i}_{n}(u^{i})|du^{i}.\\
\end{eqnarray*}
According to Lemma 2.3 in \citet{Ghosal00}, we can obtain
\begin{eqnarray}\label{ff}\Big(\int \tilde{Z}^{i}_{n}(u^{i})du^{i}\Big)^{-1}\int_{Q^{i}}|Z^{i}_{n}(u^{i})-\tilde{Z}^{i}_{n}(u^{i})|du^{i}\leq f^{i}(||\Delta^{i}_{n}||,c)
\end{eqnarray}
where \begin{eqnarray*}
f^{i}(||\Delta^{i}_{n}||,c)&=& \varphi^{i}_{n}(c)[p_{i}^2+\big(1-2\varphi^{i}_{n}(c)\big)^{-1}||\Delta^{i}_{n}||^2]\Big(1-2\varphi^{i}_{n}(c)\Big)^{-(\frac{p_{i}^2}{2}+1)}\\
&&\times\exp \Big\{\frac{\varphi^{i}_{n}(c)||\Delta^{i}_{n}||^2}{1-2\varphi^{i}_{n}(c)}\Big\},
\end{eqnarray*}
and $$\varphi^{i}_{n}(c)=\frac{1}{6}[n^{-\frac{1}{2}}\big(cM(p)\big)^{\frac{1}{2}}B^{i}_{1n}(0)+n^{-1}cM(p)B^{i}_{2n}(c\frac{M(p)}{S_{i}})].$$
Furthermore, since $||u^{i}||\leq \sqrt{cM(p)}$, by the inequality \eqref{ff}, it is easy to see that
$$\frac{\int_{Q^{i}}||u^{i}||\cdot|Z^{i}_{n}(u^{i})-\tilde{Z}^{i}_{n}(u^{i})| du^{i}}{\int \tilde{Z}^{i}_{n}(u^{i})du^{i}} \leq \sqrt{cM(p)}f(||\Delta^{i}_{n}||,c).$$
Combining the above results, we can show that the LHS in \eqref{long} is bounded by
\begin{eqnarray*}
R_{1}(||\Delta^{i}_{n}||,c) =\frac{2cM_{1}\kappa_{2}M(p)p}{\sqrt{n}}[1+f^{i}(||\Delta^{i}_{n}||,c)]+\sqrt{cM(p)}f^{i}(||\Delta^{i}_{n}||,c).
\end{eqnarray*}
According to Proposition \ref{exp}, we have $B^{i}_{1n}(0)=O(p^9)$ and $B^{i}_{2n}(c\frac{M(p)}{S_{i}})=O(p^{12})$. Therefore, there exist two constants $c_{1}$ and $c_{2}$ such that
\begin{eqnarray}
\varphi^{i}_{n}(c)&\leq&\frac{1}{6}[\frac{\big(cM(p)\big)^{\frac{1}{2}}c_{1}p^9}{\sqrt{n}}+n^{-1}cM(p)c_{2}p^{12}]
=\frac{1}{6}\frac{p^{10}\sqrt{\log p}}{\sqrt{n}}[\sqrt{c}c_{1}+c_{2}c\frac{p^{4}\sqrt{\log p}}{\sqrt{n}}]\label{cc}.\\ \nonumber
\end{eqnarray}
Since the first term in \eqref{cc} is the dominating term, then there exists a constant $c_{3}(c)$ such that $\varphi^{i}_{n}(c)\leq c_{3}(c)\frac{p^{10}\sqrt{\log p}}{\sqrt{n}}$.
By Condition (4), we have that $\varphi^{i}_{n}(c)\rightarrow 0$. Furthermore, using the fact $(1-x)^{-1} \leq 2$ and $-\log (1-x) \leq 2x$ for sufficiently small $x$, we have $[1-2\varphi^{i}_{n}(c)]^{-1}\leq 2$ and $e^{-(\frac{p^2_{i}}{2}+1)\log \big(1-2\varphi^{i}_{n}(c)\big)}\leq e^{(\frac{p^2_{i}}{2}+1)4\varphi^{i}_{n}(c)}$.
Therefore, the following inequality holds
\begin{eqnarray*}
f^{i}(||\Delta^{i}_{n}||,c)\leq \varphi^{i}_{n}(c)[p^2+2||\Delta^{i}_{n}||^2]\exp\{(\frac{p^2}{2}+1)4\varphi^{i}_{n}(c)\}\exp \Big\{2\varphi^{i}_{n}(c)||\Delta^{i}_{n}||^2\Big\}.
\end{eqnarray*}
According to Lemma \ref{max}, we see that $P(||\Delta^{i}_{n}||^2\leq 3a^2p^2)>1-10.4\exp\{-\frac{1}{6}p^2\}$. Therefore,
\begin{eqnarray*}
f^{i}(||\Delta^{i}_{n}||,c)
&\leq & c_{3}(c)\frac{p^{10}\sqrt{\log p}}{\sqrt{n}}[p^2+6a^2p^2]\exp \Big\{c_{3}(c)\frac{p^{10}\sqrt{\log p}}{\sqrt{n}}(6a^2p^2+2p^2+4)\Big\}\nonumber\\
\end{eqnarray*}
with probability greater than $1-10.4\exp\{-\frac{1}{6}p^2\}$.
By Condition (4), we have that $\frac{p^{10}\sqrt{\log p}}{\sqrt{n}}(6a^2p^2+2p^2+2)\rightarrow 0$. Therefore, $\exp \Big\{c_{3}(c)\frac{p^{10}\sqrt{\log p}}{\sqrt{n}}(4a^2p^2+2p^2+2)\Big\}<2$. It follows
\begin{eqnarray*}
f^{i}(||\Delta^{i}_{n}||,c)
&\leq & 2(1+6a^2)c_{3}(c)\frac{p^{12}\sqrt{\log p}}{\sqrt{n}}\label{b2}\\ \nonumber
\end{eqnarray*}
with probability greater than $1-10.4\exp\{-\frac{1}{6}p^2\}$. Let $c_{4}(c)=2(1+6a^2)c_{3}(c)$, then $f^{i}(||\Delta^{i}_{n}||,c)\leq c_{4}(c)\frac{p^{12}\sqrt{\log p}}{\sqrt{n}}
$ with probability greater than $1-10.4\exp\{-\frac{1}{6}p^2\}$. Furthermore, we can get
\begin{eqnarray}
R_{1}(||\Delta^{i}_{n}||,c) &=&\frac{2cM_{1}\kappa_{2}M(p)p}{\sqrt{n}}[1+c_{4}(c)\frac{p^{12}\sqrt{\log p}}{\sqrt{n}}]+\sqrt{cM(p)}c_{4}(c)\frac{p^{12}\sqrt{\log p}}{\sqrt{n}}\nonumber\\
&=&\frac{2cM_{1}\kappa_{2}p^3\log p}{\sqrt{n}}+\frac{p^{13}\log p}{\sqrt{n}}[c_{4}(c)\frac{2cM_{1}\kappa_{2}p^{2}\log^{\frac{1}{2}} p}{\sqrt{n}}+\sqrt{c}c_{4}(c)]\label{b3}\\ \nonumber
\end{eqnarray}
with probability greater than $1-10.4\exp\{-\frac{1}{6}p^2\}$.
It is easy to see that the third term in \eqref{b3} is the dominating term. Therefore, there exists a constant $c_{5}(c)$ such that $R_{1}(||\Delta^{i}_{n}||,c)\leq c_{5}(c)\frac{p^{13}\log p}{\sqrt{n}}$ with probability greater than $1-10.4\exp\{-\frac{1}{6}p^2\}$.
The proof is completed.
\end{proof}

\begin{lemma}{}{}%
\label{secondtermnorm} Under Condition (2)-(4), there exists a constant $c$ large enough and a constant $c_{9}(c)$ such that for any given $i \in \{1,2,\ldots,p\}$,
\begin{eqnarray*}
\frac{\int_{||u^{i}||^2> cM(p)}||u^{i}||\pi^{i}(\theta^{i}_{0}+n^{-\frac{1}{2}}(J^{i})^{-1}u^{i})Z^{i}_{n}(u^{i}) du^{i}}{\int \pi^{i}(\theta^{i}_{0})\tilde{Z}^{i}_{n}(u^{i})du^{i}}\leq \exp[-c_{9}(c)p^2\log p]\\
\end{eqnarray*}
with probability greater than $1-10.4\exp\{-\frac{1}{6}p^2\}$.
\end{lemma}
\begin{proof}{\textbf{Proof} }{}%
Let
\begin{eqnarray*}
R_{2}(||\Delta^{i}_{n}||,c)&=&\frac{\int_{||u^{i}||^2> cM(p)}||u^{i}||\pi^{i}(\theta^{i}_{0}+n^{-\frac{1}{2}}(J^{i})^{-1}u^{i})Z^{i}_{n}(u^{i}) du^{i} }{\int \pi^{i}(\theta^{i}_{0})\tilde{Z}^{i}_{n}(u^{i})du^{i}} \nonumber\\
&=&\frac{\int_{||u^{i}||^2> cM(p)}||u^{i}||\frac{\pi^{i}(\theta^{i}_{0}+n^{-\frac{1}{2}}(J^{i})^{-1}u^{i})}{\pi^{i} (\theta^{i}_{0})}Z^{i}_{n}(u^{i}) du^{i}}{(2\pi)^{S_{i}/2} \exp [\frac{||\Delta^{i}_{n}||^2}{2}]}. \nonumber\\
\end{eqnarray*}
According to Lemma 2.2 in \citet{Ghosal00}, we have that $Z^{i}_{n}(u^{i})\leq \exp[-\frac{1}{4}cp^2 \log p]$ with probability greater than $1-10.4\exp\{-\frac{1}{6}p^2\}$. Let $\pi_{0}^i(\theta^{i})$ denotes the non-normalized local coloured $G$-Wishart distribution. Then we obtain that
\begin{eqnarray}
&&R_{2}(||\Delta^{i}_{n}||,c)\nonumber\\
&\leq & \frac{\exp [-\frac{1}{4}cp^2 \log p]}{(2\pi)^{S_{i}/2} \exp [\frac{||\Delta^{i}_{n}||^2}{2}]} \nonumber\\
&&\hspace{8mm} \times\int_{||\sqrt{n}J^{i}(\theta^{i}-\theta^{i}_{0})||^2> cM(p)}||\sqrt{n}J^{i}(\theta^{i}-\theta^{i}_{0})||\frac{\pi_{0}^i(\theta^{i})}{\pi_{0}^i (\theta^{i}_{0})}n^{S_{i}/2}|J^{i}| d\theta^{i} \nonumber\\
&\leq& \exp [\frac{S_{i}}{2}\log n+\frac{1}{2}\log |F^{i}|-\frac{1}{4}cp^2 \log p - \log \pi_{0}^i(\theta^{i}_{0})\nonumber\\
&&\hspace{2mm} +\log \int_{||\sqrt{n}J^{i}(\theta^{i}-\theta^{i}_{0})||^2> cM(p)} ||\sqrt{n}J^{i}(\theta^{i}-\theta^{i}_{0})|| \pi_{0}^i(\theta^{i}) d\theta^{i}]\label{b1}\\\nonumber
\end{eqnarray}
with probability greater than $1-10.4\exp\{-\frac{1}{6}p^2\}$. By Proposition \ref{prior} and Lemma \ref{He}, we have that
\begin{eqnarray*}
R_{2}(||\Delta^{i}_{n}||,c)&\leq & \exp [\frac{S_{i}}{2}\log n+\frac{1}{2}\log |F^{i}|-\frac{1}{4}cp^2 \log p + \frac{1}{2}p_{i}\kappa_{2}-\frac{\delta^{i}-2}{2}p_{i}\log \kappa_{1}\nonumber\\
&&\hspace{2mm} +M_{7}p^2 \log p]\\
\end{eqnarray*}
with probability greater than $1-10.4\exp\{-\frac{1}{6}p^2\}$.
By Condition (1), $\log n$ and $\log p$ are of the same order.
Furthermore, Proposition \ref{trace} implies $\log|F^{i}|=O(p^2)$. Therefore, there exists a constant $c_{6}$ such that $\log |F^{i}|\leq c_{6}p^2$.
It follows the RHS in \eqref{b1} is bounded by the following term
\begin{eqnarray*}
&&\exp [\frac{p(p+1)}{4}\log p+\frac{1}{2}c_{6}p^2-\frac{1}{4}cp^2 \log p+\frac{1}{2}p_{i}\kappa_{2}-\frac{\delta^{i}-2}{2}p_{i}\log \kappa_{1}+M_{7}p^2 \log p]\\
\end{eqnarray*}
with probability greater than $1-10.4\exp\{-\frac{1}{6}p^2\}$.
Furthermore, there exists a constant $c_{8}$ such that
\begin{eqnarray*}
R_{2}(||\Delta^{i}_{n}||,c)&\leq \exp [\frac{p(p+1)}{4}\log p-\frac{1}{4}cp^2 \log p+M_{7}p^2 \log p+c_{8}p^2 \log p]\\
\end{eqnarray*}
with probability greater than $1-10.4\exp\{-\frac{1}{6}p^2\}$. We can choose a constant $c$ big enough such that $c_{9}(c)=\frac{1}{4}-\frac{1}{4}c+c_{8}+M_{7}<0$. It immediately implies
$R_{2}(||\Delta^{i}_{n}||,c)\leq \exp[- c_{9}(c)p^2\log p]$ with probability greater than $1-10.4\exp\{-\frac{1}{6}p^2\}$.
\end{proof}

\begin{lemma}{}{}%
\label{thirdtermnorm}Under Condition (2)-(4), for any given $i \in \{1,2,\ldots,p\}$ and for any constant $c$ such that $c>11a^2$ and $a^2>1$, we have
\begin{eqnarray*}
\int_{||u^{i}||^2> cM(p)} ||u^{i}|| \phi(u^{i}; \Delta^{i}_{n}, I_{S_{i}}) du^{i} \leq \frac{2}{\sqrt{2\pi}}p^{-4a^2+4}+\sqrt{3a^2}\frac{2}{\sqrt{2\pi}} p^{-4a^2+3}\\
\end{eqnarray*}
with probability greater than $1-10.4\exp\{-\frac{1}{6}p^2\}$.
\end{lemma}
\begin{proof}{\textbf{Proof} }{}%
First we observe that
\begin{eqnarray*}
&&\int_{||u^{i}||^2> cM(p)} ||u^{i}|| \phi(u^{i}; \Delta^{i}_{n}, I_{S_{i}}) du^{i}\\
&\leq&\int_{||u^{i}||^2> cM(p)} (||u^{i}-\Delta^{i}_{n}||) \phi(u^{i}; \Delta^{i}_{n}, I_{S_{i}}) du^{i} + \int_{||u^{i}||^2> cM(p)} ||\Delta^{i}_{n}||\phi(u^{i}; \Delta^{i}_{n}, I_{S_{i}}) du^{i}.\\
\end{eqnarray*}
Let $v^{i}=u^{i}-\Delta^{i}_{n}$, since $||v^{i}||^2+||\Delta^{i}_{n}||^{2}\geq ||v^{i}+\Delta^{i}_{n}||^{2}=||u^{i}||^2>cM(p)$, then immediately $||v^{i}||^2>cM(p)-||\Delta^{i}_{n}||^{2}$. By Lemma \ref{max}, we can see that $||\Delta^{i}_{n}||^2\leq 3a^2p^2$ with probability greater than $1-10.4\exp\{-\frac{1}{6}p^2\}$ with $a^2>1$. As $c$ is chosen that $c>11a^2$, we can get
$||v^{i}||^2>cM(p)-||\Delta^{i}_{n}||^{2}>cM(p)-3a^2p^2>(11a^2-3a^2)p^2\log p=8a^2p^2\log p$ with probability greater than $1-10.4\exp\{-\frac{1}{6}p^2\}$.
Thus the following inequality holds with probability greater than $1-10.4\exp\{-\frac{1}{6}p^2\}$.
\begin{eqnarray*}
&&\int_{||u^{i}||^2> cM(p)} (||u^{i}-\Delta^{i}_{n}||) \phi(u^{i}; \Delta^{i}_{n}, I_{S_{i}}) du^{i} \\
&=& \int_{||v^{i}+\Delta^{i}_{n}||^2> cM(p)} ||v^{i}|| \phi(v^{i}; 0, I_{S_{i}}) dv^{i}\leq \sum\limits^{S_{i}}_{j=1} \int_{||v^{i}||^2>8a^2M(p)} |v^{i}_{j}|\phi(v^{i}; 0, I_{S_{i}}) dv^{i},\\
\end{eqnarray*}
where $v^{i}_{j}$ is the $j$-th element of $v^{i}$.
We also have that
\begin{eqnarray*}
&&\sum\limits^{S_{i}}_{j=1} \int_{||v^{i}||^2>8a^2M(p)} |v^{i}_{j}|\phi(v^{i}; 0, I_{S_{i}}) dv^{i}\\
&&\leq \sum\limits^{S_{i}}_{j=1}\sum\limits^{S_{i}}_{k=1}\int_{\mathbb{R}^{S^{i}-1}} \int_{(v^{i}_{k})^2> 8a^2\frac{M(p)}{S_{i}}} |v^{i}_{j}|\phi(v^{i}; 0, I_{S_{i}}) dv^{i}\leq 2p^4\frac{1}{\sqrt{2\pi}}p^{-4a^2}
\end{eqnarray*}
with probability greater than $1-10.4\exp\{-\frac{1}{6}p^2\}$,
and
\begin{eqnarray*}
\int_{||u^{i}||^2> cM(p)}||\Delta^{i}_{n}|| \phi(u^{i}; \Delta^{i}_{n}, I_{S_{i}}) du^{i}
\leq  \sqrt{3a^2}p^3\frac{2}{\sqrt{2\pi}} p^{-4a^2}
\end{eqnarray*}
with probability greater than $1-10.4\exp\{-\frac{1}{6}p^2\}$. Hence, the desired result follows.
\end{proof}

\begin{lemma}{}{}%
\label{allnorm} For a given $i \in \{1,2,\ldots,p\}$, we have
$$\sqrt{n}J^{i}(\tilde{\theta^{i}}-\theta^{i}_{0})= \Delta^{i}_{n}+ \int u^{i} [\pi_{*}^{i}(u^{i}) - \phi(u^{i}; \Delta^{i}_{n}, I_{S_{i}})] du^{i}$$
where $\pi_{*}^{i}(u^{i})$ is the posterior distribution of $u^{i}$.
\end{lemma}

\begin{lemma}{}{}%
(Parallel to Lemma \ref{Xin})\label{Xin1} Let $\Delta_{n}^{i}$ and $\bar{U}^{i}_{j}$ be as defined in \eqref{delta} of Section 4 and \eqref{U}, respectively. Let $C_{2}$ be defined as in Lemma \ref{boundcondition}. Then under Conditions (2), (4*) and (5), for any arbitrary constant $a$ such that $a^2>1$, we have that if $\frac{C_{2}\log p}{3\sqrt{n}}\leq a-1$, for $||\gamma^{i}|| < \log p$ and $n$ sufficiently large,
$$G_{\Delta_{n}^{i}}^{i}(\gamma^i)=\log \big(E\{\exp[(\gamma^i)^t\Delta_{n}^{i}]\}\big)\leq a^2||\gamma^{i}||^2/2.$$
\end{lemma}

\begin{lemma}{}{}%
 (Parallel to Lemma \ref{max}) \label{max1}Under Condition (2), (4*) and (5), for any $i \in \{1,2,\ldots,p\}$ and $n$ sufficiently large, there exists a constant $a$, $a^2>1$, such that $$P\{||\Delta^{i}_{n}||^2 > 3a^2\log^2 p\} \leq 10.4\exp\{-\frac{1}{6}\log^2 p\}$$
where $\Delta^{i}_{n}$ is defined as in \eqref{delta} of Section 4.
\end{lemma}

\begin{lemma}{}{}%
\label{znubound}
Let $\hat{\theta}^{i}$ be the MLE of $\theta^{i}$ in the $i$-th local model.
Under Condition (2), (4*) and (5), for any $i \in \{1,2,\ldots,p\}$, $$\sqrt{n}||J^{i}(\hat{\theta}^{i}-\theta^{i}_{0})||\leq c'\log p$$ with probability greater than $1-10.4\exp\{-\frac{1}{6}\log^2 p\}$, where $c'=1.2 \sqrt{3}a\frac{\lambda_{max}(F^{i})}{\lambda^2_{min}(F^{i})}$.
\end{lemma}

\section{Supplementary file}
In this section we provide the proofs of the propositions and lemmas that we have used in the Appendix.
\begin{proposition}{}{}%
\label{Fbound}
Let $F^{i}$ be defined in definition \eqref{F} of Section 4 for any $i \in \{1,2,\ldots,p\}$, then under Condition (2), we have that $$\frac{1}{\kappa_{2}^2}\leq \lambda_{min}(F^{i})\leq \lambda_{max}(F^{i})\leq \frac{1}{\kappa^2_{1}}.$$
\end{proposition}
\begin{proof}{\textbf{Proof} }{}%
Let $G^{i}$ be the Fisher information matrix for the uncolored graphical models e.g. $G^{i}=\psi_{u}''(\theta^{i})$ where $\psi_{u}(\theta^{i})=(-\frac{1}{2}\log |K^{i}|+\frac{p_{i}}{2}\log (2\pi))\mathbf{1}_{K^{i} \in P_{G_{i}}}$. Let $\tau$ and $\varpi$ be the numbers of eigenvalues of $G^{i}$ and $F^{i}$. Since $F^{i}$ is a linear projection of $G^{i}$ onto the space of uncolored symmetric matrices, then $\tau>\varpi$. Under Condition (2) and by Proposition \ref{eigen}, for any $l$, $1 \leq l\leq \varpi$, we have $$\frac{1}{\kappa_{2}^2}\leq\min\{\frac{1}{\lambda_{j}(G^{i})\lambda_{k}(G^{i})}|1\leq j, k \leq \tau \}\leq \lambda_{l}(F_{i})\leq \max\{\frac{1}{\lambda_{j}(G^{i})\lambda_{k}(G^{i})}|1\leq j, k \leq \tau\}\leq \frac{1}{\kappa^2_{1}}.$$
\end{proof}

\begin{proposition}{}{}%
\label{entrybound} For any $i \in \{1,2,\ldots,p\}$, let $K^{i,0}_{\alpha\beta}$ be the $(\alpha, \beta)$ entry of $K^{i}_{0}$. Under Condition (2), we have $|K^{i,0}_{\alpha\beta}|\leq \kappa_{2}$.
\end{proposition}
\begin{proof}{\textbf{Proof} }{}%
By Condition (2), we have $\lambda_{max}(K^{i}_{0})\leq \kappa_{2}$ for any $i \in \{1,2,\ldots,p\}$. Therefore, $\kappa_{2}-\lambda_{j}(K^{i}_{0}), j=1,2,\cdots,p_{i}$, are the eigenvalues of $\kappa_{2}I_{p_{i}}-K^{i}_{0}$. Since $\lambda_{max}(K^{i}_{0})\leq \kappa_{2}$, then $\kappa_{2}\geq\lambda_{j}(K^{i}_{0}), j=1,2,\cdots, p_{i}$. It follows that $\kappa_{2}I_{p_{i}}-K^{i}_{0}$ is a positive semidefinite matrix. Since the diagonal elements of a positive semidefinite $\kappa_{2}I_{p_{i}}-K^{i}_{0}$ are all non negative, then $\kappa_{2}-K^{i,0}_{\alpha\alpha}\geq 0, \alpha=1,2, \ldots, p_{i}$. It follows $0< K^{i,0}_{\alpha\alpha}\leq \kappa_{2}$. Since $K^{i}_{0}$ is a positive definite matrix, then each 2 by 2 principal sub matrices \[ \left( \begin{array}{cc}
K^{i,0}_{\alpha\alpha} & K^{i,0}_{\alpha\beta}  \\
K^{i,0}_{\beta\alpha} & K^{i,0}_{\beta\beta}  \\
\end{array} \right)\] of $K^{i}_{0}$ are positive definite. Therefore, $K^{i,0}_{\alpha\alpha}K^{i,0}_{\beta\beta}-(K^{i,0}_{\alpha\beta})^2 >0$, from which we get $|K^{i,0}_{\alpha\beta}|<(K^{i,0}_{\alpha\alpha}K^{i,0}_{\beta\beta})^{1/2}<\kappa_{2}$.
\end{proof}

The next four propositions provide the properties of $\log |F^{i}|$, the colored $G$-Wishart prior, the third and fourth moments of the normalized $Y_{j}^{i}$.

\begin{proposition}{}{}%
\label{trace}
Under Condition (3), for any $i \in \{1,2,\ldots,p\}$, we have the trace of $F^{i}$ satisfies $tr(F^{i})=O(p^2)$ and the determinant $|F^{i}|$ satisfies $\log |F^{i}|=O(p^2)$.
\end{proposition}
\begin{proof}{\textbf{Proof} }{}%
Since $\frac{\partial^2 \psi(\theta^{i})}{\partial \theta_{j}^{i}\partial \theta_{k}^{i}}=\frac{1}{2}tr(\delta^{i}_{j}\Sigma^{i}_{0}\delta^{i}_{k}\Sigma^{i}_{0})$, then $tr(F^{i})=\frac{1}{2}\sum\limits^{S_{i}}_{j=1}tr((\delta^{i}_{j}\Sigma^{i}_{0})^2)$. Furthermore, by Condition (3), $tr(\delta^{i}_{j}\Sigma^{i}_{0})$ is bounded. Therefore, $tr((\delta^{i}_{j}\Sigma^{i}_{0})^2)$ is bounded. It follows $$tr(F^{i})=\frac{1}{2}\sum\limits^{S_{i}}_{j=1}tr((\delta^{i}_{j}\Sigma^{i}_{0})^2)\leq \frac{1}{2}\frac{p_{i}(p_{i}+1)}{2}tr((\delta^{i}_{j}\Sigma^{i}_{0})^2)\leq \frac{1}{2}\frac{p(p+1)}{2}tr((\delta^{i}_{j}\Sigma^{i}_{0})^2)=O(p^2).$$
Next, let us consider $\log |F^{i}|$. Since
$|F^{i}|=\prod_{j=1}^{S_{i}}\lambda_{j}(F^{i}) \leq \Big(\frac{\sum\limits^{S_{i}}_{j=1}\lambda_{j}(F^{i})}{S_{i}}\Big)^{S_{i}}=\Big(\frac{tr(F^{i})}{S_{i}}\Big)^{S_{i}},$ then
$$\log |F^{i}|\leq S_{i} \log \frac{tr(F^{i})}{S_{i}}\leq \frac{p_{i}(p_{i}+1)}{2}\log \frac{\frac{1}{2}\frac{p_{i}(p_{i}+1)}{2}tr((\delta^{i}_{j}\Sigma^{i}_{0})^2)}{\frac{p_{i}(p_{i}+1)}{2}}=O(p^2).$$
The proposition is proved.
\end{proof}

\begin{proposition}{}{}%
\label{prior}Under Condition (2), for any $i \in \{1,2,\ldots,p\}$, we have $$\log \pi_{0}^i(K^{i}_{0})\geq -\frac{1}{2}p_{i}\kappa_{2}+\frac{\delta^{i}-2}{2}p_{i}\log \kappa_{1}$$ when $D^{i}=I_{p_{i}}$.
\end{proposition}
\begin{proof}{\textbf{Proof} }{}%
The non-normalized colored $G$-Wishart distribution can be rewritten as
\begin{eqnarray*}
\pi_{0}^{i}(K^{i}_{0})&=& \exp\big\{-\frac{1}{2}tr(K^{i}_{0}I_{p_{i}})+\frac{\delta^{i}-2}{2}\log |K^{i}_{0}|\big\} \\
&=& \exp \big\{-\frac{1}{2}\sum\limits^{p_{i}}_{j=1}\lambda_{j}(K_{0}^{i})+\frac{\delta^{i}-2}{2}\log \prod^{p_{i}}_{j=1}\lambda_{j}(K_{0}^{i})\big\} \\
&\geq & \exp \big\{-\frac{1}{2}p_{i}\kappa_{2}+\frac{\delta^{i}-2}{2}p_{i}\log \kappa_{1}\big\}.
\end{eqnarray*}
The last inequality due to Condition (2).
\end{proof}

\begin{proposition}{}{}%
\label{lipschiz}(Lipschitz continuity) For any $i \in \{1,2,\ldots,p\}$ and any constant $c$, there exists a constant $M_{1}$ such that $$|\log \pi^{i}(\theta^{i})-\log \pi^{i}(\theta^{i}_{0})| \leq M_{1} p ||\theta^{i}-\theta^{i}_{0}||$$ when $||\theta^{i}-\theta^{i}_{0}|| \leq\sqrt{||(F^{i})^{-1}||cM(p)/n}\rightarrow 0$.
\end{proposition}
\begin{proof}{\textbf{Proof} }{}%
Let $\pi_{0}^i(\theta^{i})$ be the non-normalized colored $G$-Wishart distribution for the local model. By mean value theorem, we have
\begin{eqnarray*}
|\log \pi^{i}(\theta^{i})-\log \pi^{i}(\theta^{i}_{0})|&=&|\log \pi_{0}^i(\theta^{i})-\log \pi_{0}^i(\theta^{i}_{0})|
=|(\theta^{i}-\theta^{i}_{0})^{t}\frac{\partial \log \pi_{0}^i(\theta^{i})}{\partial \theta^{i}}|_{\theta^{i}=\check{\theta}^{i}}|\\
&=& ||\theta^{i}-\theta^{i}_{0}||\cdot\sqrt{\sum\limits^{S_{i}}_{j=1}\Big[-\frac{1}{2}tr(\delta^{i}_{j}D^{i})+\frac{\delta^{i}-2}{2}tr(\delta^{i}_{j} (\check{K}^{i})^{-1})\Big]^2},
\end{eqnarray*}
where $\check{\theta^{i}}$ is the point on the line segment joining $\theta^{i}$ and $\theta^{i}_{0}$.  Since $||\theta^{i}-\theta^{i}_{0}||\rightarrow 0$, then $(\check{K}^{i})^{-1}\rightarrow (K_{0}^{i})^{-1}$. According to Condition (2) and Proposition \ref{entrybound}, each entry of $(K_{0}^{i})^{-1}$ is uniformly bounded, then using the similar proof of Lemma \ref{mark}, each entry $(\check{K}^{i})^{-1}$ is uniformly bounded. Therefore, there exists a constant $M_{1}$ such that $$\sqrt{\sum\limits^{S_{i}}_{j=1}\Big[-\frac{1}{2}tr(\delta^{i}_{j} D^{i})+\frac{\delta-2}{2}tr(\delta^{i}_{j} \Sigma^{i}_{0})\Big]^2}\leq \sqrt{\frac{p_{i}(1+p_{i})}{2}M^2_{1}}\leq \sqrt{\frac{p(1+p)}{2}M^2_{1}}=M_{1}p.$$
\end{proof}

\begin{proposition}{}{}%
\label{exp} For any $i \in V$, let $Y^{i}_{j}$ and $V_{j}^{i}$ be defined in \eqref{sufficient} and \eqref{stand} of Section 4, respectively. Then $B^{i}_{1n}(c)=O(p^{9})$ and $B^{i}_{2n}(c)=O(p^{12})$.
\end{proposition}
\begin{proof}{\textbf{Proof} }{}%
Let $B_{\alpha\beta}$ be the $(\alpha, \beta)$ entry of $(J^{i})^{-1}$. Define $b=\max \{|B_{\alpha\beta}|; \alpha,\beta\in \{1,2,\ldots,S_{i}\}\}$. Then for the vectors $Y_{j}^{i}=(Y^{i}_{j1}, Y^{i}_{j2}, \ldots, Y^{i}_{jS_{i}})^t$ and $a=(a_{1}, a_{2}, \ldots, a_{S_{i}})^t$, the following property holds for $h=1,2,3,4$
\begin{eqnarray}
E_{\theta^{i}}|a^{t}(J^{i})^{-1}Y_{j}^{i}|^h\leq E_{\theta^{i}}\Big[(|a_{1}|, |a_{2}|, \ldots, |a_{S_{i}}|)\left(\begin{array}{c}
b\sum\limits^{S_{i}}_{k=1}|Y^{i}_{jk}|  \\
b\sum\limits^{S_{i}}_{k=1}|Y^{i}_{jk}|\\
\vdots \\
b\sum\limits^{S_{i}}_{k=1}|Y^{i}_{jk}| \\\end{array} \right)\Big]^{h}
=E_{\theta^{i}}\Big[(b\sum\limits^{S_{i}}_{k=1}|Y^{i}_{jk}|)\sum\limits^{S_{i}}_{k=1}|a_{k}|\Big]^{h}.\nonumber\\
\end{eqnarray}
According to Cauchy-Schwarz inequality, we have that
\begin{eqnarray}
E_{\theta^{i}}|a^{t}(J^{i})^{-1}Y_{j}^{i}|^h&\leq& E_{\theta^{i}}\Big[ b (\sum\limits^{S_{i}}_{k=1}|Y^{i}_{jk}|) \sqrt{S_{i}}||a||\Big]^{h}
\leq b^{h}(S_{i})^{h/2}E_{\theta^{i}}\Big[ \sum\limits^{S_{i}}_{k_{1}=1}\ldots\sum\limits^{S_{i}}_{k_{h}=1}|Y^{i}_{jk_{1}}|\cdots|Y^{i}_{jk_{h}}|\Big].\label{N1}\\ \nonumber
\end{eqnarray}
According to Lemma \ref{mark}, each entry of $\theta^{i}$ is bounded when $||J^{i}(\theta^{i}-\theta^{i}_{0})||^2 \leq \frac{cS_{i}}{n}\rightarrow 0$. By Lemma \ref{aa},
we have $E_{\theta^{i}}\Big[|Y^{i}_{jk_{1}}|\cdots|Y^{i}_{jk_{h}}|\Big]$ is bounded for $h=1,2,3,4$. Therefore, $E_{\theta^{i}}|a^{t}(J^{i})^{-1}Y_{j}^{i}|^h= O(p_{i}^{3h})$. Similarly, $|a^{t}(J^{i})^{-1}E_{\theta^{i}}(Y_{j}^{i})|^h=  O(p_{i}^{3h})$. Hence, we have
\begin{eqnarray}
E_{\theta^{i}}|a^{t}V_{j}^{i}|^3
&=&E_{\theta^{i}}|a^{t}(J^{i})^{-1}Y_{j}^{i}-a^{t}(J^{i})^{-1}E_{\theta^{i}}(Y^{i}_{j})|^3\nonumber\\
&\leq &E_{\theta^{i}}|a^{t}(J^{i})^{-1}Y^{i}_{j}|^3+3|a^{t}(J^{i})^{-1}E_{\theta^{i}}(Y^{i}_{j})|E_{\theta^{i}}(a^{t}(J^{i})^{-1}Y^{i}_{j})^2\nonumber\\
&&\hspace{4mm}+3[a^{t}(J^{i})^{-1}E_{\theta^{i}}(Y^{i}_{j})]^2E_{\theta^{i}}|a^{t}(J^{i})^{-1}Y^{i}_{j}|+|[a^{t}(J^{i})^{-1}E_{\theta^{i}}(Y^{i}_{j})]^3|\label{N2}\\
&=& O(p_{i}^{9})=O(p^{9}).\nonumber
\end{eqnarray}
A similar argument deduces $E_{\theta^{i}}|a^{t}V_{j}^{i}|^4=O(p^{12})$. By the definition $B^{i}_{1n}(c)$ and $B^{i}_{2n}(c)$, the desired result follows.
\end{proof}

\begin{proof}{\textbf{Proof of Lemma \ref{allnorm} } }{}%
Let $q_{*}^{i}(\theta^{i})$ be the posterior distribution of $\theta^{i}$. Therefore, we have that
\begin{eqnarray*}
\tilde{\theta^{i}}&=& \int \theta^{i} \cdot  q_{*}^{i}(\theta^{i}) d\theta^{i} \\
&=&\int (\theta^{i}_{0}+n^{-\frac{1}{2}}(J^{i})^{-1}u^{i})q_{*}^{i}(\theta^{i}_{0}+n^{-\frac{1}{2}}(J^{i})^{-1}u^{i})|n^{-1/2}(J^{i})^{-1}| d u^{i}\\
&=&\int (\theta^{i}_{0}+n^{-\frac{1}{2}}(J^{i})^{-1}u^{i})\pi_{*}^{i}(u^{i})du^{i}
= \theta^{i}_{0}+n^{-\frac{1}{2}}(J^{i})^{-1} \int u^{i}\pi_{*}^{i}(u^{i}) du^{i}.
\end{eqnarray*}
It follows $\sqrt{n}J^{i}(\tilde{\theta^{i}}-\theta^{i}_{0}) = \int u^{i}\pi_{*}^{i}(u^{i}) du^{i}. $ On the other hand, the following equations hold
\begin{eqnarray*}
\int u^{i}\phi(u^{i}; \Delta^{i}_{n}, I_{S_{i}})du
=\int (u^{i}-\Delta^{i}_{n})\phi(u^{i}; \Delta^{i}_{n}, I_{S_{i}})du^{i}+\Delta^{i}_{n}\int \phi(u^{i}; \Delta^{i}_{n}, I_{S_{i}}) du^{i}
= \Delta^{i}_{n}.
\end{eqnarray*}
We thus have
\begin{eqnarray*}
\sqrt{n}J^{i}(\tilde{\theta^{i}}-\theta^{i}_{0})-\Delta^{i}_{n}
= \int u^{i} [\pi_{*}^{i}(u^{i})  - \phi(u^{i}; \Delta^{i}_{n}, I_{S_{i}})] du^{i}. \\
\end{eqnarray*}
\end{proof}

\begin{proof}{\textbf{Proof of Lemma \ref{Xin1} } }{}%
Since $||\gamma^{i}||<\log p$ and $\frac{\log p}{\sqrt{n}}\rightarrow 0$ by Condition (4*),
then $||\frac{\gamma^{i}}{\sqrt{n}}||\leq \eta$ where $\eta$ as given in Lemma \ref{boundcondition} is the size of the neighborhood for $\gamma^{i}$. Therefore, by Lemma \ref{boundcondition}, there exists a constant $C_{2}$ such that $\Big|\frac{\partial^3 G_{\bar{U}^{i}_{j}}^{i}(\frac{\gamma^{i}}{\sqrt{n}})}{\partial \gamma^{i}_{k}\partial \gamma^{i}_{l}\partial \gamma^{i}_{m}}\Big|\leq C_{2}$.
We also have $$\frac{1}{3}\frac{C_{2}}{\sqrt{n}}\sum\limits_{m=1}^{S_{i}}\gamma_{m}^{i}=O\big(\frac{\log p}{\sqrt{n}}).$$ According to Condition (4*), $\frac{\log p}{\sqrt{n}}=o(1)$. Therefore, for any arbitrary constant $a$ such that $a^2>1$, $\frac{1}{3}\frac{C_{2}}{\sqrt{n}}\sum\limits_{m=1}^{S_{i}}\gamma_{m}^{i}\leq a^2-1$. Following the argument similar to that of Lemma \ref{Xin}, we obtain $$\log E[e^{(\gamma^{i})^t\eta^{i}}]\leq a^2||\gamma^{i}||^2/2.$$
\end{proof}

\begin{proof}{\textbf{Proof of Lemma \ref{max1} } }{}%
According to Lemma \ref{Xin1}, we have
$$\log \big(E\{\exp [(\gamma^{i})^t \Delta^{i}_{n}]\}\big)\leq a^2 ||\gamma^{i}||^2/2 \hspace{10mm} for \hspace{10mm} ||\gamma^{i}||\leq \log p$$
where $a$ is a constant with $a^2>1$. Condition (5) implies $a^2\log^2 p > S_{i}$. Let $x^{i}_{c}$ be defined as in the proof of Lemma \ref{max}. Then $x^{i}_{c}>\frac{1}{4}a^2\log^2p$ and let $x=\frac{1}{6}\log^2p$, then we have $\frac{S_{i}}{6.6} <x < x^{i}_{c}$. Following similar argument as in the proof of Lemma \ref{max}, we can obtain that $P(||\Delta^{i}_{n}||^2\geq 3a^2\log^2p)\leq 10.4 e^{-\frac{1}{6}\log^2p}$.
\end{proof}

\begin{proof}{\textbf{Proof of Lemma \ref{znubound} } }{}%
Let $B^{i}(\theta^{i})=\psi'(\theta^{i})-\bar{Y}^{i}$ be the negative of the score function. Then the MLE $\hat{\theta}^{i}$ satisfy the likelihood equation $B^{i}(\hat{\theta}^{i})=0$.
Let $b_{n}=\frac{\sqrt{3}a\log p}{\sqrt{n}}\frac{\lambda^{\frac{1}{2}}_{max}(F^{i})}{\lambda_{min}(F^{i})}$ with $a^2>1$.
We are going to show with probability greater than $1-10.4\exp\{-\frac{1}{6}\log^2p\}$, for any $\theta^{i}$ on the ball
$||\theta^{i}-\theta_{0}^{i}||=1.2b_{n},$ we have
\begin{eqnarray}\label{ball}
(\theta^{i}-\theta_{0}^{i})^tB^{i}(\theta^{i})> 0.
\end{eqnarray}
Because that according to Theorem 6.3.4 of \citet{Orte70}, this will imply that there exists a root of $B^{i}(\hat{\theta}^{i})=0$ inside the ball $||\theta^{i}-\theta_{0}^{i}||\leq 1.2b_{n}$ and thus with probability greater than $1-10.4\exp\{-\frac{1}{6}\log^2p\}$, $||J^{i}(\hat{\theta}^{i}-\theta_{0}^{i})||\leq \lambda^{\frac{1}{2}}_{max}(F^{i})1.2b_{n}\leq c'\frac{\log p}{\sqrt{n}}.$ To complete the proof, it now suffices to show the inequality \eqref{ball} holds.
Based on (2.3) in Proposition 2.1 of \citet{Portnoy88}, we have
\begin{eqnarray*}
(\theta^{i}-\theta_{0}^{i})^tB^{i}(\theta^{i})&=&(\theta^{i}-\theta_{0}^{i})^t(\psi'(\theta^{i})-\bar{Y}^{i})\\
&=&(\theta^{i}-\theta_{0}^{i})^t\mu^{i}+(\theta^{i}-\theta_{0}^{i})^t\psi''(\theta^{i}_{0})(\theta^{i}-\theta_{0}^{i})\\
&&+\frac{1}{2}E_{\bar{\theta}^{i}}[(\theta^{i}-\theta_{0}^{i})^tV^{i}_{j}]^3-(\theta^{i}-\theta_{0}^{i})^t\bar{Y}^{i}\\
&=&-(\theta^{i}-\theta_{0}^{i})^t[\bar{Y}^{i}-\mu^{i}]+(\theta^{i}-\theta_{0}^{i})^t\psi''(\theta^{i}_{0})(\theta^{i}-\theta_{0}^{i})\\
&&+\frac{1}{2}E_{\bar{\theta}^{i}}[(\theta^{i}-\theta_{0}^{i})^tV^{i}_{j}]^3\\
&=& term 1+term 2+term 3,
\end{eqnarray*}
where $\mu^{i}=\psi'(\theta^{i}_{0})$, $V^{i}_{j}$ is defined as in \eqref{stand} of Section 4 and $\bar{\theta}^{i}$ is a point on the line segment between $\theta^{i}$ and $\theta_{0}^{i}$.
It is easy to see that
$$term 2 \geq \lambda_{min}(F^{i})||(\theta^{i}-\theta_{0}^{i})||^2.$$
For $term3$, under Condition (5), from \eqref{N1}, \eqref{N2} and Lemma \ref{mark}, we see that
$$\sup\{|E_{\theta}(a^tV_{j}^{i})^3|: ||a||=1, ||\theta^{i}-\theta_{0}^{i}||=1.2b_{n}\}\leq \sup\{|E_{\theta}(a^tV_{j}^{i})^3|: ||a||=1, ||\theta^{i}-\theta_{0}^{i}||\leq 1.2b_{n}\}$$ is bounded.
Since $b_{n}\rightarrow 0$, then $\frac{0.1}{b_{n}}\lambda_{min}(F^{i})\rightarrow \infty$. Therefore,
$$\sup\{|E_{\theta}(a^tV_{j}^{i})^3|: ||a||=1, ||\theta^{i}-\theta_{0}^{i}||=1.2b_{n}\}\leq \frac{0.1}{b_{n}}\lambda_{min}(F^{i}).$$
It follows $$term 3\geq -\frac{0.05}{b_{n}}\lambda_{min}(F^{i}).$$ In $term1$, there is a random term $\bar{Y}^{i}-\psi'(\theta^{i}_{0})$. We will now show that
$$term1 \geq \frac{\sqrt{3}a\log p}{\sqrt{n}}\lambda^{\frac{1}{2}}_{max}(F^{i})||\theta^{i}-\theta_{0}^{i}||$$
with probability greater than $1-10.4\exp\{-\frac{1}{6}\log^2p\}$.
According to Lemma \ref{max1}, we have $||\Delta^{i}_{n}||^2\leq 3a^2\log^2 p$ with probability greater than $1-10.4\exp\{-\frac{1}{6}\log^2p\}$. Furthermore, since $$||\Delta^{i}_{n}||^2=||\sqrt{n}(J^{i})^{-1}(\bar{Y}^{i}-\mu^{i})||^2=n(\bar{Y}^{i}-\mu^{i})^t(F^{i})^{-1}(\bar{Y}^{i}-\mu^{i})\geq \frac{n}{\lambda_{max}(F^{i})}||\bar{Y}^{i}-\mu^{i}||^2,$$
then $\frac{n}{\lambda_{max}(F^{i})}||\bar{Y}^{i}-\mu^{i}||^2\leq 3a^2 \log^2p$ with probability greater than $1-10.4\exp\{-\frac{1}{6}\log^2p\}$.
It implies $||\bar{Y}^{i}-\mu^{i}||^2\leq \frac{\sqrt{3}a\log p}{\sqrt{n}}\lambda^{\frac{1}{2}}_{max}(F^{i})$ with probability greater than $1-10.4\exp\{-\frac{1}{6}\log^2p\}$. Consequently, $term1 \geq \frac{\sqrt{3}a\log p}{\sqrt{n}}\lambda^{\frac{1}{2}}_{max}(F^{i})||\theta^{i}-\theta_{0}^{i}||.$
Combining the above results, on the ball of $||\theta^{i}-\theta_{0}^{i}||=1.2b_{n}$, we have
\begin{eqnarray*}
&&(\theta^{i}-\theta_{0}^{i})^tB^{i}(\theta^{i})\\
&\geq & -\frac{\sqrt{3}a\log p}{\sqrt{n}}\lambda^{\frac{1}{2}}_{max}(F^{i})||\theta^{i}-\theta_{0}^{i}||+\lambda_{min}(F^{i})||(\theta^{i}-\theta_{0}^{i})||^2
-\frac{0.05}{b_{n}}\lambda_{min}(F^{i})||\theta^{i}-\theta_{0}^{i}||^3\\
&=&-3a^2\sqrt{\frac{\log^2p}{n}}\lambda^{\frac{1}{2}}_{max}(F^{i})1.2b_{n}+\lambda_{min}(F^{i})(1.2b_{n})^2
-\lambda_{min}(F^{i})\frac{0.05}{b_{n}}(1.2b_{n})^3\\
&\geq& \lambda_{min}(F^{i})b^2_{n}[-1.2+(1.2)^2-0.05(1.2)^3]>0
\end{eqnarray*}
with probability greater than $1-10.4\exp\{-\frac{1}{6}\log^2p\}$.  Therefore, we proved that $||\hat{\theta}^{i}-\theta_{0}^{i}||\leq 1.2b_{n}$ with probability greater than $1-10.4\exp\{-\frac{1}{6}\log^2 p\}$.
It follows $$||J^{i}(\hat{\theta}^{i}-\theta_{0}^{i})||\leq \lambda^{\frac{1}{2}}_{max}(F^{i})||\hat{\theta}^{i}-\theta_{0}^{i}||\leq \lambda^{\frac{1}{2}}_{max}(F^{i})1.2b_{n}$$
with probability greater than $1-10.4\exp\{-\frac{1}{6}\log^2 p\}$.
\end{proof}

\begin{lemma}{}{}%
\label{mark}Let $\theta_{j}^{i}$ be the $j$-th element of $\theta^{i}$, $i=1,2,\ldots,p$. Under Condition (2), for $||\theta^{i}-\theta_{0}^{i}||\leq \varepsilon_{1}$, we have that $|\theta_{j}^{i}|\leq \varepsilon_{1}+\kappa_{2}$.
\end{lemma}
\begin{proof}{\textbf{Proof} }{}%
Let $\theta_{j,0}^{i}$ be the $j$-th element of $\theta_{0}^{i}$, $i=1,2,\ldots,p$. Since $||\theta^{i}-\theta_{0}^{i}||\leq \varepsilon_{1}$, then $\sqrt{\sum\limits_{j=1}^{S_{i}}(\theta_{j}^{i}-\theta_{j,0}^{i})^2}\leq \varepsilon_{1}$. Therefore, for any $j\in \{1,2,\ldots,S_{i}\}$, we have
$$\sqrt{(\theta_{j}^{i}-\theta_{j,0}^{i})^2}\leq \sqrt{\sum\limits_{j=1}^{S_{i}}(\theta_{j}^{i}-\theta_{j,0}^{i})^2}\leq \varepsilon_{1}.$$ It implies $|\theta_{j}^{i}-\theta_{j,0}^{i}|\leq \varepsilon_{1}$. By Proposition \ref{entrybound}, under Condition (2), we have $|\theta_{j,0}^{i}|\leq \kappa_{2}$.
It follows $|\theta_{j}^{i}|\leq \varepsilon_{1}+\kappa_{2}$.
\end{proof}

\begin{lemma}{}{}%
\label{aa}
Let $Y^{i}_{j}$ be defined in \eqref{sufficient} of Section 4 and denote $Y^{i}_{j}=(Y_{j1},Y_{j2},\cdots,Y_{jS_{i}})^t$, under Condition (2) and $||\theta^{i}-\theta_{0}^{i}||\leq \varepsilon_{1}$, we have $E_{\theta^{i}}\Big[|Y^{i}_{jk_{1}}|\cdots|Y^{i}_{jk_{h}}|\Big]$ is bounded for $h=1,2,3,4$.
\end{lemma}
\begin{proof}{\textbf{Proof} }{}%
According to Lemma \ref{mark}, each element of $\theta^{i}$ is bounded. Since $Y^{i}_{jk}=-\frac{1}{2}tr(\delta_{k}^{i} X^{i}_{j}(X^{i}_{j})^t)$, by Isserlis' Theorem, the moments of every entry of $X^{i}_{j}(X^{i}_{j})^t$ is finite. By Condition (3), $E_{\theta^{i}}(Y^{i}_{jk})$ is bounded and $E_{\theta^{i}}(Y^{i}_{jk})^2$ is also bounded. By H\"{o}lder's inequality, we have $E_{\theta^{i}}[|XY|]\leq (E_{\theta^{i}}[|X|^p])^\frac{1}{p}(E_{\theta^{i}}[|Y|^q])^\frac{1}{q}$. Therefore, when $h=1$, $E_{\theta^{i}}(|Y^{i}_{jk_{1}}|)\leq [E_{\theta^{i}}(Y^{i}_{jk_{1}})^2]^{\frac{1}{2}}$ is bounded. When $h=2$,  we have
$$E_{\theta^{i}}(|Y^{i}_{jk_{1}}||Y^{i}_{jk_{2}}|)\leq [E_{\theta^{i}}(Y^{i}_{jk_{1}})^2]^{\frac{1}{2}}[E_{\theta^{i}}(Y^{i}_{jk_{2}})^2]^{\frac{1}{2}}.$$
It follows $E_{\theta^{i}}(|Y^{i}_{jk_{1}}||Y^{i}_{jk_{2}}|)$ is bounded. When $h=3$,  we have $$E_{\theta^{i}}(|Y^{i}_{jk_{1}}||Y^{i}_{jk_{2}}||Y^{i}_{jk_{3}}|)\leq [E_{\theta^{i}}(|Y^{i}_{jk_{1}}||Y^{i}_{jk_{2}}|)^2]^{\frac{1}{2}}[E_{\theta^{i}}(Y^{i}_{jk_{3}})^2]^{\frac{1}{2}}.$$ Since $E_{\theta^{i}}(|Y^{i}_{jk_{1}}||Y^{i}_{jk_{2}}|)$ is bounded, then $E_{\theta^{i}}(|Y^{i}_{jk_{1}}||Y^{i}_{jk_{2}}|)^2$ is also bounded. Therefore, $E_{\theta^{i}}(|Y^{i}_{jk_{1}}||Y^{i}_{jk_{2}}||Y^{i}_{jk_{3}}|)$ is bounded. Consequently,
$E_{\theta^{i}}\Big[|Y^{i}_{jk_{1}}|\cdots|Y^{i}_{jk_{h}}|\Big]$ is bounded for $h=1,2,3,4$.
\end{proof}

\begin{proposition}{}{}%
\label{eigen}
Let $E$ be a Euclidean space and let $F\subset E$ be a linear subspace. Let $p_{F}$ denote the orthogonal projection of $E$ onto $F$. Let g be a linear symmetric operator $g: E \rightarrow E$ and consider the linear application $f$ of $F$ into itself defined by
$$f: x\in F\rightarrow f(x)=p_{F}\circ g(x)$$
Then, we have that if $\mu_{1}<\mu_{2}<\cdots<\mu_{m}$ are the eigenvalues of $g$ and $\lambda_{1}<\lambda_{2}<\cdots<\lambda_{n}$ are the eigenvalues of $f$, $n<m$, then for any $j=1,2,\cdots,n$, the following inequalities hold
$$\mu_{1}\leq \lambda_{j}\leq \mu_{m}.$$
\end{proposition}
\begin{proof}{\textbf{Proof} }{}%
We prove is first for $m=dim(F)=dim(E)-1$. Let $e=(e_{1},e_{2},\cdots,e_{n})$ be an orthonormal basis of $F$ such that basis the matrix representative of $f$ is a diagonal $[f]_{e}^{e}=diag(\lambda_{1},\lambda_{2},\cdots,\lambda_{n})$ and let $e_{0}\in E$ be such that $e'=(e_{0},e_{1},e_{2},\cdots,e_{n})$ is an orthonormal basis of $E$. Then in that basis, the matrix representative of $g$ is
\begin{eqnarray*}\label{diag2}
[g]_{e'}^{e'}= \left( \begin{array}{cccc}
a&b_{1} & \cdots&b_{n}\\
b_{1}&\lambda_{1} &0 &0\\
\cdots & \cdots &\ddots &0  \\
b_{n} & \cdots & 0&\lambda_{n} \\
\end{array} \right).
\end{eqnarray*}
We see here that the matrix representative of $f$ is a submatrix of the matrix representative of $g$. By the interlacing property of the eigenvalues, we have
$$\mu_{1}\leq \lambda_{j} \leq \mu_{n+1} \hspace{4mm} j=1,2,\ldots,n.$$
If $dim(E)-dim(F)>1$, we iterate the process by induction on $dim(E)-dim(F)$ and complete the proof.
\end{proof}

\begin{lemma}{}{}%
\label{positive} For any $i \in \{1,2,\ldots,p\}$, let $T^i(\gamma^{i})$ be a symmetric matrix with dimension $p_{i}$. Then there exists a constant $\eta$ such that with $||T^{i}(\gamma^{i})||_{F}\leq \eta$, the matrix $I_{p_{i}}+T^{i}(\gamma^{i})\Sigma_{0}^{i}$ is positive definite.
\end{lemma}
\begin{proof}{\textbf{Proof} }{}%
If we want to show $I_{p_{i}}+T^{i}(\gamma^{i})\Sigma_{0}^{i}$ is positive definite, it is equivalent to show for any non zero vector $z$ with dimension $p_{i}$, $z^t(I_{p_{i}}+T^{i}(\gamma^{i})\Sigma_{0}^{i})z$ is positive. By Cauchy-Schwarz inequality, we have
\begin{eqnarray*}
|<T^{i}(\gamma^{i})z,\Sigma_{0}^{i}z>|&\leq& ||T^{i}(\gamma^{i})z||\times||\Sigma_{0}^{i}z||\leq ||T^{i}(\gamma^{i})||\times||z||\times ||\Sigma_{0}^{i}||\times||z||\\
&\leq& ||z||^2||T^{i}(\gamma^{i})||_{F}\times \frac{1}{\kappa_{1}}\leq \eta||z||^2\frac{1}{\kappa_{1}}.
\end{eqnarray*}
Therefore,
\begin{eqnarray*}
z^t[I_{p_{i}}+T^{i}(\gamma^{i})\Sigma_{0}^{i}]z&=&z^tI_{p_{i}}z+z^tT^{i}(\gamma^{i})\Sigma_{0}^{i}z=||z||^2+<T^{i}(\gamma^{i})z,\Sigma_{0}^{i}z>\\
&\geq& ||z||^2-||z||^2\eta\frac{1}{\kappa_{1}}=(1-\eta\frac{1}{\kappa_{1}})||z||^2.
\end{eqnarray*}
We can thus choose a constant $\eta$, such that $\eta < \kappa_{1}$. It follows $z^t(I_{p_{i}}+T^{i}(\gamma^{i})\Sigma_{0}^{i})z \geq ||z||^2>0$ when $||T^{i}(\gamma^{i})||_{F}\leq \eta$.
\end{proof}

\begin{lemma}{}{}%
\label{HL}
Let $K=(K_{ij})_{1\leq i,j\leq p}$ be $p\times p$ positive semi-definite matrix.
Then
$$||K||_F\leq tr^2 (K).$$
\end{lemma}
\begin{proof}{\textbf{Proof} }{}%
Because $K$ is positive semi-definite, we have $K_{ij}^2\leq K_{ii}K_{jj}$. Thus
$$\sum_{1\leq i,j\leq p}K_{ij}^2\leq \sum_{1\leq i,j\leq p}K_{ii}K_{jj}=(\sum_{1\leq i \leq p}K_{ii})(\sum_{1\leq j \leq p}K_{jj})=(\sum_{1\leq i \leq p}K_{ii})^2.$$
\end{proof}

In order to prove the following Theorem , we start from a finite dimensional real linear space $E$ of dimension $n$ (thus isomorphic to $\mathbb{R}^n$ but we prefer  to avoid the use of artificial coordinates).  We denote by $E^*$ its dual space, that means the set of linear applications $\theta: E\mapsto \mathbb{R}$. We denote $\<\theta,x\>=\theta(x).$ If $E$ is Euclidean, the dual $E^*$ is identified with $E$ and $\<\theta,x\>$ is the scalar product.

Consider a non empty open convex cone $C$ with closure $\overline{C}$ such that $C$ is proper, that is to say such that
$$\overline{C}\cap (-\overline{C})=\{0\}.$$ The dual cone of $C$ is
$$C^*=\{\theta\in E^*\ ; \ \<\theta,x\>\ \geq 0\ \forall x\in \overline{C}\setminus \{0\}\}.$$ This is a standard result of convex analysis that $C^*$ is not empty (\citet{Faraut94}). In general the description of $C^*$ is a non trivial matter.

A polynomial $P$ on $E$ is a function $P: E\mapsto \mathbb{R}$ such that if
$e=(e_1,\ldots,e_n)$ is a basis of $E$ and if $x=x_1e_1+\cdots+x_ne_n\in E$ then $P(x)$ is a polynomial with respect to the real variables $(x_1,\ldots,x_n).$ Needless to say the definition does not depend on the particular chosen basis $e$. A polynomial $P$ is homogeneous of degree $k$ if for all $\lambda\in \mathbb{R}$ and all $x\in E$ we have
$$P(\lambda x)=\lambda^k P(x).$$

\begin{theorem}{}{}%
\label{cone}
Let $C$ be an open convex and proper cone of $E$, let $P$ be a homogeneous polynomial on $E$ of degree $k$ and let $\alpha >-n/k.$ We assume that $P(x)>0$ on $C.$ We choose a Lebesgue measure  $dx$ on $E.$ For $\theta\in E^*$ consider the integral
$$L(\theta)=\int_Ce^{-\<\theta,x\>}P(x)^{\alpha }dx\leq \infty.$$ If $\theta\notin  C^*$ the integral $L(\theta)$ diverges. If $\theta\in C^*$
 denote $H_1=\{x\in E\ ;\ \<\theta,x\> =1 \}.$ Then $\overline{C}\cap H_1$ is compact. In this case  $\theta\in C^*$, the integral  $L(\theta)$ is finite if and only if $\int_{C\cap H_1}P(x)^{\alpha}dx$ is finite. Furthermore
\begin{equation}\label{value}L(\theta)=\Gamma(\alpha k+n)\int_{C\cap H_1}P(x)^{\alpha}dx.\end{equation}
\end{theorem}

\begin{proof}{\textbf{Proof (personal communication from G. Letac)} }{}%
Suppose that $\theta_0\in C^*$ and let us show (\ref{value}). Consider the affine hyperplanes $H_1$ and $H_0$ of $E$ defined by
$$H_1=\{x\in E\ ;\ \<\theta_0,x\> =1 \},\ \ H_0=\{x\in E\ ;\ \<\theta_0,x\> =0 \}.$$ The convex set $\overline{C}\cap H_1$ is compact. To see this let us choose an arbitrary scalar product on  $E.$ Observe that the function $u\mapsto \<\theta_0,u\>$ defined on the intersection of $\overline{C}$  with the unit sphere of $E$ is continuous and reaches a minimum $m>0$ since the set of definition is compact. Thus for all $x\in \overline{C}\cap H_1$ we have
$$\|x\|\leq \frac{1}{m}\<\theta_0,x\> =\frac{1}{m} $$
and the closed set $\overline{C}\cap H_1$ is also bounded, thus compact.

We fix now $h_1\in H_1$ and we write any element $x$ of $E$ in a unique way as $x=x_0+x_1h_1$ where $x_1$ is a number and $x_0$ is in $H_0.$ If $E$ is Euclidean, a natural choice for $h_1$ is $\theta_0/\|\theta_0\|^2$ although other choices would be possible. We also write $x=(x_0,x_1)$ for short. We denote by $K\subset H_0$ the set of $x_0$ such that $x_0+h_1=(x_0,1)$ is in $\overline{C}\cap H_1$. Note that $K$ is also compact. We get that $x=(x_0,x_1)$ is in $\overline{C}\setminus\{0\}$ if and only if $y=x_0/x_1\in K$ and $x_1>0.$ To see this denote $$C_1=\{(x_0,x_1)\ ;\ y=x_0/x_1\in K,\ x_1>0\}.$$ The inclusion $C_1\subset \overline{C}\setminus\{0\}$ is obvious as well as $\overline{C}\setminus H_0\subset C_1.$ However if $(x_0,0)$ is in $\overline{C}\cap H_0$ and if $x_0\neq 0$ this implies that $(\lambda x_0,0)$ is in $\overline{C}\cap H_0$ for all $\lambda>0$ and thus $\lambda x_0\in K$ for all $\lambda>0$: this contradicts the compactness of $K.$ As a result $(x_0,0)$  in $\overline{C}\cap H_0$ implies $x_0=0.$ This implies $\overline{C}\setminus\{0\}=\overline{C}\setminus H_0$ and thus $\overline{C}\setminus\{0\}=C_1$

We are now in position to make the change of variable $(x_0,x_1)\mapsto (y=x_0/x_1,x_1)$ in the integral $L(\theta_0)$ with an easy Jacobian, since $\dim H_0=n-1:$
$$dx=dx_0dx_1=x_1^{n-1}dydx_1.$$ We get
$$L(\theta_0)=\int_{C}e^{-x_1}P(x_0,x_1)^{\alpha}dx_0dx_1=\int_KI(y)dy$$ where
 $$I(y)=\int_0^{\infty}e^{-x_1}P(yx_1,x_1)^{\alpha}x_1^{n-1}dx_1=P(y,1)^{\alpha}\Gamma(\alpha k+n)=P(y+h_1)^{\alpha}\Gamma(\alpha k+n)$$ from the homogeneity of the polynomial $P.$ Thus (\ref{value}) is proved.

\vspace{4mm}\noindent Suppose that $\theta_0\notin C^*.$  This is saying that  there exists $x_0\in C$ such that $\<\theta_0,x_0\>\leq 0$. Let us show that $L(\theta_0)=\infty.$
Since $C$ is open we may assume that $\<\theta_0,x_0\><0$. Choose an arbitrary scalar product on  $E.$ There exists $\epsilon$ such that for all $x$ in  $B=\{x\ ;\ \|x-x_0\|< \epsilon\}$ we have $x\in C$ and $\<\theta_0,x\><0.$ Consider the open subcone $C_1=\{\lambda x\ ;\ x\in B,\ \lambda >0\}$ of $C.$ We can write
$$L(\theta_0)\geq \int_{C_1}e^{-\<\theta_0,x\>}P(x)^{\alpha}dx\geq \int_{C_1}P(x)^{\alpha}dx.$$ Clearly the last integral diverges for $\alpha\geq 0.$ For $-n/k<\alpha<0$ we use the same trick: we parameterize $C_1$ with the help of the compact set $\overline{C_1}\cap H_1$ by considering the compact set $K_1$ of $y\in H_0$ such that $y+h_1\in \overline{C_1}\cap H_1$ and we write

$$\int_{C_1}P(x)^{\alpha}dx=\int_{K_1}\int_0^{\infty}P^{\alpha}(x_0,x_1)dx=\int_{K_1}P(y,1)^{\alpha}\left(\int_0^{\infty}x_1^{\alpha k+n-1}dx_1\right)dy=\infty.$$
This proves $\Rightarrow.$
\end{proof}

\begin{lemma}{}{}%
\label{He}For any $i \in \{1,2,\ldots,p\}$, there exists a constant $M_{7}$ such that
$$\log \int_{||\sqrt{n}J^{i}(\theta^{i}-\theta_{0}^{i})||^2>cM(p)}||\sqrt{n}J^{i}(\theta^{i}-\theta_{0}^{i})||\pi_{0}^{i}(\theta^{i})d \theta^{i}\leq \exp [M_{7}p^2 \log p].$$
\end{lemma}
\begin{proof}{\textbf{Proof} }{}%
Without loss of generality, let $\theta^{i}_k, k=1,\ldots,s_{i}$, be the entries of $K^{i}$ on the diagonal and $\theta^{i}_k, k=s_{i}+1,\ldots,S_{i}$, be the off-diagonal entries. We assume that $D^{i}=I_{p_{i}}$, which we need later on anyway.
Then
\begin{eqnarray*}
&&tr(K^{i}D^{i})=tr(K^{i}I_{p_{i}}) =\sum_{k=1}^{s_{i}}\tau^{i}_k\theta^{i}_k=tr(K^{i}),\\
&&\pi^{i}_0(\theta^{i};\delta^{i},I_{p_{i}})=\exp \{-\frac{1}{2}\sum_{k=1}^{s_{i}} \theta^{i}_k \tau^{i}_k+\frac{\delta^{i}-2}{2}\log |K^{i}(\theta^{i})|\}, \hspace{8mm}\text{and}\\
||\theta^{i}||^2&=&\sum_{k=1}^{s_{i}}(\theta^{i}_k)^2+\sum_{k=s_{i}+1}^{S_{i}}(\theta^{i}_k)^2\leq \sum_{k=1}^{S_{i}}\tau^{i}_k(\theta^{i}_k)^2=||K^{i}(\theta^{i})||_{F}\leq (\sum_{k=1}^{s_{i}}\tau^{i}_k\theta^{i}_k)^2\;\mbox{by lemma \ref{HL}},\\
\end{eqnarray*}
where $\tau^{i}_k=|v^{i}_k|$ is the number of elements in the colour class $v^{i}_k$.
We therefore have $||\theta^{i}||\leq \sum\limits_{k=1}^{s_{i}}\tau^{i}_k\theta^{i}_k$. Let $H^{i}$ denote the convex cone $P_{{\cal G}^{i}}$ for short.
\begin{eqnarray}
&&\int_{H^{i}}||\sqrt{n}J^{i}(\theta^{i}-\theta^{i}_0)||\pi^{i}_0(\theta^{i})d \theta^{i}\nonumber\\
&=&\int_{H^{i}}||\sqrt{n}J^{i}(\theta^{i}-\theta^{i}_0)||\exp \{-\frac{1}{2}\sum_{k=1}^{s_{i}} \theta^{i}_k \tau^{i}_k+\frac{\delta^{i}-2}{2}\log |K^{i}(\theta^{i})|\}d \theta^{i}\nonumber\\
&\leq&\sqrt{||F^{i}||}\sqrt{n}\int_{H^{i}}\Big[||\theta^{i}||+||\theta^{i}_0||\Big]\exp \{-\frac{1}{2}\sum_{k=1}^{s_{i}} \theta^{i}_k \tau^{i}_k+\frac{\delta^{i}-2}{2}\log |K^{i}(\theta^{i})|\}d \theta^{i}\nonumber\\
&=&\sqrt{||F^{i}||}\sqrt{n}||\theta^{i}_0||\int_{H^{i}}\exp \{-\frac{1}{2}\sum_{k=1}^{s_{i}} \theta^{i}_k \tau^{i}_k+\frac{\delta^{i}-2}{2}\log |K^{i}(\theta^{i})|\}d \theta^{i}\label{thetanot}\\
&&\hspace{1cm}+\sqrt{||F^{i}||}\sqrt{n}\int_{H^{i}}||\theta^{i}||\exp \{-\frac{1}{2}\sum_{k=1}^{s_{i}} \theta^{i}_k \tau^{i}_k+\frac{\delta^{i}-2}{2}\log |K^{i}(\theta^{i})|\}d \theta^{i}.\label{theta}
\end{eqnarray}
By Proposition \ref{entrybound}, we have $||\theta^{i}_0||^2\leq S_{i}\kappa_{2}^2.$ Furthermore, according to Proposition \ref{Fbound}, we have $||F^{i}|| \leq \frac{1}{\kappa^2_{1}}$. We therefore need to find upper bounds for the integrals in \eqref{thetanot} and \eqref{theta}. These two integrals are of the type $\int_{H^{i}}f(\theta^{i})^{\alpha^{i}}e^{-tr(\theta^{i} D^{i})} d \theta^{i}$ where $f(\theta^{i})$ is a homogeneous function of order $k^{i}$. If $n^{i}$ is the dimension of the space in which $H^{i}$ sits, we use the result of Theorem \ref{cone}. Let $\bar{D}^{i}$ be the $s_{i}$-dimensional vector with entries $\frac{\tau^{i}_k}{2}$. We have
\begin{eqnarray*}
\int_{H^{i}}\exp \{-\frac{1}{2}\sum_{k=1}^{s_{i}} \theta^{i}_k \tau^{i}_k +\frac{\delta^{i}-2}{2}\log |K^{i}(\theta^{i})|\}d \theta^{i}&=&\int_{H^{i}}e^{-tr(\bar{D}^{i} \theta^{i})}|K^{i}(\theta^{i})|^{\frac{\delta^{i}-2}{2}}d\theta^{i},
\\
\int_{H^{i}}||\theta^{i}||\exp \{-\frac{1}{2}\sum_{k=1}^{s_{i}} \theta^{i}_k \tau^{i}_k+\frac{\delta^{i}-2}{2}\log |K^{i}(\theta^{i})|\}d \theta^{i}&\leq &\int_{H^{i}}(\sum_{k=1}^{s_{i}}\tau^{i}_k\theta^{i}_k)e^{-tr( \bar{D}^{i}\theta^{i})}|K^{i}(\theta^{i})|^{\frac{\delta^{i}-2}{2}}d\theta^{i}
\end{eqnarray*}
and therefore since $K^{i}(\theta^{i})$ is homogeneous of order $p_{i}$, $\sum\limits_{k=1}^{S_{i}}\tau^{i}_k\theta^{i}_k$ is homogeneous of order $1$ and $1$ is homogeneous of order $0$, we have, for $\alpha^{i}=\frac{\delta^{i}-2}{2}$
\begin{eqnarray}
&&\int_{H^{i}}e^{-tr(\bar{D}^{i} \theta^{i})}|K^{i}(\theta^{i})|^{\frac{\delta^{i}-2}{2}}d\theta^{i}=\Gamma(\alpha^{i} p_{i}+S_{i})\int_{H^{i}\cap H^{i}_1}|K^{i}(\theta^{i})|^{\frac{\delta^{i}-2}{2}}d\theta^{i},\nonumber\\
&&\int_{H^{i}}(\sum_{k=1}^{s_{i}}\tau^{i}_k\theta^{i}_k)e^{-tr(\bar{D}^{i}\theta^{i})}|K^{i}(\theta^{i})|^{\frac{\delta^{i}-2}{2}}d\theta^{i}\nonumber\\
&=&\Gamma(\alpha^{i} p_{i}+1+S_{i})\int_{H^{i}\cap H^{i}_1}(\sum_{k=1}^{s_{i}}\tau^{i}_k\theta^{i}_k)|K^{i}(\theta^{i})|^{\frac{\delta^{i}-2}{2}}d\theta^{i}.\label{all}
\end{eqnarray}

However, we do not know how to compute the integrals $\int_{H^{i}\cap H^{i}_1}|K^{i}(\theta^{i})|^{\frac{\delta^{i}-2}{2}}d\theta^{i}$ and
$\int_{H^{i}\cap H^{i}_1}(\sum\limits_{k=1}^{s_{i}}\tau^{i}_k\theta^{i}_k)|K^{i}(\theta^{i})|^{\frac{\delta^{i}-2}{2}}d\theta^{i}.$
The set $H^{i}_1=\{\theta^{i}\mid tr(\bar{D}^{i}\theta^{i})=1\}$ is $H^{i}_1=\{\theta^{i}\mid \sum\limits_{k=1}^{s_{i}}\tau^{i}_k\theta^{i}_k=2\}$. So, we only have one integral, $\int_{H^{i}\cap H^{i}_1}|K^{i}(\theta^{i})|^{\frac{\delta^{i}-2}{2}}d\theta^{i}$, to compute. But
$$\sum\limits_{k=1}^{s_{i}}\tau^{i}_k\theta^{i}_k=tr(K^{i}(\theta^{i})I_{p_{i}})=tr(K^{i}(\theta^{i}))=\sum_{j=1}^{p_{i}}\lambda_j$$
where the $\lambda_j$ are the eigenvalues of $K^{i}(\theta^{i})$.
Following the inequality between the arithmetic mean and the geometric mean, on $H^{i}_1\cap H^{i}$, we have
$$|K^{i}(\theta^{i})|=(\prod_{i}^{p_{i}}\lambda_j)\leq \frac{(\sum_{j=1}^{p_{i}}\lambda_j)^{p_{i}}}{p_{i}^{p_{i}}}=\frac{2^{p_{i}}}{p_{i}^{p_{i}}}$$
and thus
\begin{equation}
\label{k}
\int_{H^{i}\cap H^{i}_1}|K^{i}(\theta^{i})|^{\frac{\delta^{i}-2}{2}}d\theta^{i}\leq \frac{2^{p_{i}}}{p_{i}^{p_{i}}}\int_{H^{i}\cap H^{i}_1}d\theta^{i}.
\end{equation}
We are now going to use Theorem \ref{cone} in the reverse direction with $f(\theta^{i})=1$ in order to evaluate $\int_{H^{i}\cap H^{i}_1}d\theta^{i}$. We have
\begin{equation}
\label{Htrace}
\int_{H^{i}}e^{-tr(\bar{D}^{i}\theta^{i})} d\theta^{i}=\Gamma(0+S_{i})\int_{H^{i}\cap H^{i}_1}d\theta^{i}
\end{equation}
and we are going to majorize $\int_{H^{i}}e^{-tr(\bar{D}^{i}\theta^{i})} d\theta^{i}$. We now use the fact that the matrices in $H^{i}$ are positive definite, thus we have that, for $l=s_{i}+1,\ldots,S_{i}$, $(\theta^{i}_l)^2\leq \theta^{i}_{t_l}\theta^{i}_{u_l}$ whenever $\theta^{i}_l=K^{i}_{jk},\;j\not =k$ and $\theta^{i}_{t_l}=K^{i}_{jj}, \theta^{i}_{u_l}=K^{i}_{kk}$. Since the cone $H^{i}$ is included in the cone $P^{i}$ of positive definite matrices, we have that, for $l=s_{i}+1,\ldots,S_{i}$, $(\theta^{i}_l)^2\leq \theta^{i}_{t_l}\theta^{i}_{u_l}$ whenever $\theta^{i}_l=K^{i}_{jk},\;j\not =k$ and $\theta^{i}_{t_l}=K^{i}_{jj}, \theta^{i}_{u_l}=K^{i}_{kk}$ and thus we can write
\begin{eqnarray*}
\int_{H^{i}}e^{-tr(\bar{D}^{i}\theta^{i})}d\theta^{i}
&\leq&\int_0^{+\infty}\ldots  \int_0^{+\infty}e^{-\sum_{r=1}^{s_{i}}\bar{D}^{i}_r\theta^{i}_r} \Big[\prod_{l=s_{i}+1}^{S_{i}}\int_{-\sqrt{\theta^{i}_{t_l}\theta^{i}_{u_l}}}^{-\sqrt{\theta^{i}_{t_l}\theta^{i}_{u_l}}} d\theta^{i}_l\Big]\prod_{r=1}^{s_{i}} d\theta^{i}_r\\
&=&\int_0^{+\infty}\ldots  \int_0^{+\infty}e^{-\sum_{r=1}^{s_{i}}\bar{D}^{i}_r\theta^{i}_r} \Big[\prod_{l=s_{i}+1}^{S_{i}}
2\sqrt{\theta^{i}_{t_l}\theta^{i}_{u_l}}\Big]\prod_{r=1}^{s_{i}} d\theta^{i}_r.
\end{eqnarray*}
Since we have assumed that $D^{i}$ is equal to the identity, $\bar{D}^{i}_r=\frac{\tau^{i}_r}{2}$ with the $\tau^{i}_r$ being bounded. Then
\begin{eqnarray*}
&&\int_0^{+\infty}\ldots  \int_0^{+\infty}e^{-\sum_{r=1}^{s_{i}}\bar{D}^{i}_r\theta^{i}_r}\prod_{r=1}^{s_{i}} \Big[\prod_{l=s_{i}+1}^{S_{i}}
2\sqrt{\theta^{i}_{t_l}\theta^{i}_{u_l}}\Big]d\theta^{i}_r\\
&=&
2^{S_{i}-s_{i}}\prod_{r=1}^{s_{i}}\int_0^{+\infty}(\theta^{i}_r)^{\frac{k^{i}_r}{2}}e^{-\frac{\tau^{i}_r\theta^{i}_r}{2}}d\theta^{i}_r=2^{S_{i}-s_{i}}\prod_{r=1}^{s_{i}}(\frac{2}{\tau^{i}_r})^{\frac{k^{i}_r}{2}+1}\Gamma(\frac{k^{i}_r}{2}+1)
\end{eqnarray*}
where $k^{i}_r$ is the number of $t_l$ or $u_l$ equal to $r$ in the $i$-th local model. From the majorization above, \eqref{Htrace}, \eqref{k} and \eqref{all} successively, we obtain the following inequalities
\begin{eqnarray*}
\int_{H^{i}\cap H^{i}_1}d\theta^{i}&\leq& \frac{1}{\Gamma(S_{i})}2^{S_{i}-s_{i}}\prod_{r=1}^{s_{i}}(\frac{2}{\tau^{i}_r})^{\frac{k^{i}_r}{2}+1}\Gamma(\frac{k^{i}_r}{2}+1),\\
\int_{H^{i}\cap H^{i}_1}|K^{i}(\theta^{i})|^{\frac{\delta^{i}-2}{2}}d\theta^{i}&\leq& \frac{2^{S_{i}-s_{i}+p_{i}}}{p_{i}^{p_{i}}\Gamma(S_{i})}\prod_{r=1}^{s_{i}}(\frac{2}{\tau^{i}_r})^{\frac{k^{i}_r}{2}+1}\Gamma(\frac{k^{i}_r}{2}+1),\\
\int_{H^{i}}(\sum_{k=1}^{S_{i}}\tau^{i}_k\theta^{i}_k)e^{-tr(\bar{D}^{i} \theta^{i})}|K^{i}(\theta^{i})|^{\frac{\delta^{i}-2}{2}}d\theta^{i}&\leq &\frac{2^{S_{i}-s_{i}+p_{i}+1}\Gamma(\alpha^{i} p_{i}+1+S_{i})}{p_{i}^{p_{i}}\Gamma(S_{i})}\prod_{r=1}^{s_{i}}(\frac{2}{\tau^{i}_r})^{\frac{k^{i}_r}{2}+1}\Gamma(\frac{k^{i}_r}{2}+1),\\
\text{and} \hspace{4mm}\int_{H^{i}}e^{-tr(\bar{D}^{i}\theta^{i})}|K^{i}(\theta^{i})|^{\frac{\delta^{i}-2}{2}}d\theta^{i}&\leq &\frac{2^{S_{i}-s_{i}+p_{i}}\Gamma(\alpha^{i} p_{i}+S_{i})}{p_{i}^{p_{i}}\Gamma(S_{i})}\prod_{r=1}^{s_{i}}(\frac{2}{\tau^{i}_r})^{\frac{k^{i}_r}{2}+1}\Gamma(\frac{k^{i}_r}{2}+1).
\end{eqnarray*}
It follows that
\begin{eqnarray*}
&&\int ||\sqrt{n}J(\theta^{i}-\theta^{i}_{0})||\pi^{i}_{0}(\theta^{i}) d \theta^{i}\\
&\leq & n^{\frac{1}{2}}\frac{1}{\kappa_{1}}\frac{2^{S_{i}-s_{i}+p_{i}}}{p_{i}^{p_{i}}\Gamma(S_{i})}[\prod^{s_{i}}_{r=1}(\frac{2}{\tau^{i}_{r}})^{\frac{k^{i}_{r}}{2}+1}\Gamma(\frac{k^{i}_{r}}{2}+1)](M_{0}p_{i}\Gamma(\alpha^{i} p_{i}+S_{i})+2\Gamma(\alpha^{i} p_{i} +1+S_{i}))\\
&\leq &n^{\frac{1}{2}}\frac{1}{\kappa_{1}}\frac{2^{S_{i}-s_{i}+p_{i}}}{p_{i}^{p_{i}}\Gamma(S_{i})}[\prod^{s_{i}}_{r=1}(\frac{2}{\tau^{i}_{r}})^{\frac{k^{i}_{r}}{2}+1}\Gamma(\frac{k^{i}_{r}}{2}+1)]M_{2}p_{i}\Gamma(\alpha ^{i} p_{i} +1+S_{i}),
\end{eqnarray*}
where $M_{0}$ and $M_{2}$ are constants.
Therefore,
\begin{eqnarray*}
&&\log \big\{\int ||\sqrt{n}J^{i}(\theta^{i}-\theta^{i}_{0})||\pi^{i}_{0}(\theta^{i})d \theta^{i}\big\}\\
&\leq &\frac{1}{2}\log n-\log \kappa_{1}+(S_{i}-s_{i}+p_{i})\log 2-p_{i}\log p_{i}-\log \Gamma(S_{i})+\log M_{2}+\log p_{i}\\
&&+\log \Gamma(\alpha^{i} p_{i} +1+S_{i})+\sum\limits^{s_{i}}_{r=1}[(\frac{k^{i}_{r}}{2}+1)\log \frac{2}{\tau^{i}_{r}} + \log \Gamma(\frac{k^{i}_{r}}{2}+1)].\\
\end{eqnarray*}
Since $\log n$ and $\log p$ is the same order and $k^{i}_{r} \leq p_{i}$, we have that
\begin{eqnarray*}
&&\exp[\log \big\{\int ||\sqrt{n}J^{i}(\theta^{i}-\theta_{0}^{i})||\pi_{0}^{i}(\theta^{i})d \theta^{i}\big\}]\\
&\leq &\exp[\frac{1}{2}\log p+(S_{i}-s_{i}+p_{i})\log 2+\log p_{i}+\log \Gamma(\alpha^{i} p_{i} +1+S_{i})+p_{i}(\frac{p_{i}}{2}+1)\log 2\\
&&+ p_{i}\log \Gamma(\frac{p_{i}}{2}+1)+M_{3}],\\
\end{eqnarray*}
where $M_{3}$ is a constant.
By Sterling's approximation, we have $\log n!=n\log n+O(\log n)$. Therefore, there exist two constant $M_{5}$ and $M_{6}$ such that
\begin{eqnarray*}
\log \Gamma(\alpha^{i} p_{i} +1+S_{i}) \leq \log \Gamma(\alpha^{i} p_{i} +1+\frac{p_{i}(p_{i}+1)}{2})\leq\log (2\alpha^{i} p_{i}+p_{i}^2)!\leq M_{5}p_{i}^2 \log p_{i}
\end{eqnarray*}
and
\begin{eqnarray*}
\log \Gamma(\frac{p_{i}}{2}+1)\leq \log p_{i}!\leq M_{6}p_{i} \log p_{i}.
\end{eqnarray*}
Combining all results above, we obtain that
\begin{eqnarray*}
&&\exp\Big(\log \big\{\int ||\sqrt{n}J^{i}(\theta^{i}-\theta^{i}_{0})||\pi^{i}_{0}(\theta^{i})d \theta^{i}\big\}\Big)\\
&\leq &\exp\Big(\frac{1}{2}\log p+[\frac{p_{i}(p_{i}+1)}{2}+p_{i}]\log 2+\log p_{i}+M_{5}p_{i}^2 \log p_{i}+p_{i}(\frac{p_{i}}{2}+1)\log 2\\
&&+M_{6}p_{i}^2 \log p_{i}+M_{3}\Big)\leq  \exp [M_{7}p^2 \log p],
\end{eqnarray*}
where $M_{7}$ is a constant.
\end{proof}

\end{document}